\documentclass[letterpaper,11pt]{article}
\usepackage[margin=1in]{geometry}
\usepackage{amsmath, amsfonts, amssymb, amsthm}
\usepackage{thm-restate}
\usepackage[colorlinks=true,urlcolor=black,linkcolor=black,citecolor=black]{hyperref}
\usepackage[capitalize]{cleveref}
\usepackage{bbm,url}
\usepackage{graphicx}
\usepackage{xcolor}
\usepackage{xspace}
\usepackage{colortbl}
\usepackage{subcaption}
\usepackage{hhline}
\usepackage{pifont}
\usepackage{multirow}
\usepackage{multicol}
\usepackage{enumitem}
\usepackage{comment}
\usepackage{nicefrac,bm}
\usepackage{algorithm}
\usepackage[noend]{algpseudocode}
\usepackage{thmtools}
\usepackage{thm-restate}
\usepackage[nottoc]{tocbibind}
\usepackage{tikz}
\usetikzlibrary{positioning,arrows.meta,decorations.pathreplacing}
\usepackage[numbers]{natbib}

\allowdisplaybreaks

\newtheorem{theorem}{Theorem}
\newtheorem*{theorem*}{Theorem}

\newtheorem{proposition}[theorem]{Proposition}
\newtheorem{lemma}{Lemma}
\newtheorem{claim}{Claim}

\newtheorem{example}{Example}

\newtheorem{definition}{Definition}

\newcommand{\s}[1]{#1}
\newcommand{\mpb}{{\textup{MPB}}\xspace}

\newcommand{\child}[1]{\texttt{child}(#1)}
\newcommand{\parent}[1]{\texttt{parent}(#1)}
\newcommand{\efx}{{\textup{EFX}}\xspace}
\newcommand{\pefx}{{\textup{pEFX}}\xspace}
\newcommand{\efone}{{\textup{EF1}}\xspace}
\newcommand{\eftwo}{{\textup{EF2}}\xspace}
\newcommand{\pefk}{{\textup{pEF}$k$}\xspace}
\newcommand{\po}{{\textup{PO}}\xspace}
\newcommand{\fpo}{{\textup{fPO}}\xspace}

\newcommand{\x}{\mathbf{x}}
\newcommand{\y}{\mathbf{y}}
\newcommand{\z}{\mathbf{z}}
\newcommand{\p}{\mathbf{p}}
\newcommand{\q}{\mathbf{q}}

\newcommand{\N}{\mathbb N}
\newcommand{\Z}{\mathbb Z}

\newcommand{\R}{\mathbb R}
\newcommand{\poly}[1]{\mathsf{poly}(#1)}

\renewcommand{\b}{\mathbf{b}}

\renewcommand{\v}{\mathbf{v}}
\renewcommand{\r}{\mathbf{r}}
\renewcommand{\P}{\mathcal{P}}

\newcommand{\bbeta}{\boldsymbol{\beta}}

\title{Constant-Factor EFX Exists for Chores\thanks{Work supported by NSF Grants CCF-1942321 and CCF-2334461.}}
\author{Jugal Garg\footnote{University of Illinois at Urbana-Champaign, USA} \\
\texttt{\small jugal@illinois.edu} 
\and
Aniket Murhekar\footnote{University of Illinois at Urbana-Champaign, USA}\\
\texttt{\small aniket2@illinois.edu}
\and
John Qin\footnote{University of Illinois at Urbana-Champaign, USA}\\
\texttt{\small johnqin2@illinois.edu}
}

\date{}
\begin{document}
\renewcommand{\arraystretch}{1.2}
\maketitle
\thispagestyle{empty}

\begin{abstract}
We study the problem of \emph{fair} allocation of chores among agents with additive preferences. In the discrete setting, envy-freeness up to any chore (EFX) has emerged as a compelling fairness criterion. However, establishing its (non-)existence or achieving a meaningful approximation remains a major open question in fair division. The current best guarantee is the existence of $O(n^2)$-EFX allocations, where $n$ denotes the number of agents, obtained through a sophisticated algorithm~\cite{zhou22efx}. In this paper, we show the existence of $4$-EFX allocations, providing the first constant-factor approximation of EFX. 

We further investigate the existence of allocations that are both fair and \emph{efficient}, using Pareto optimality (PO) as our efficiency criterion. For the special case of bivalued instances, we establish the existence of allocations that are both $3$-EFX and PO, thereby improving upon the current best factor of $O(n)$-EFX without any efficiency guarantees. For general additive instances, the existence of allocations that are $\alpha$-EF$k$ and PO has remained open for any constant values of $\alpha$ and $k$, where EF$k$ denotes envy-freeness up to $k$ chores. We provide the first positive result in this direction by showing the existence of allocations that are $2$-EF$2$ and PO. 

Our results are obtained via a novel economic framework called \textit{earning restricted (ER) competitive equilibrium} for fractional allocations, which imposes limits on the earnings of agents from each chore. We show the existence of ER equilibria by carefully formulating a linear complementarity problem (LCP) that captures all ER equilibria, and then prove that the classic complementary pivot algorithm applied to this LCP terminates at an ER equilibrium. By carefully setting earning limits and leveraging the properties of ER equilibria, we design algorithms that involve rounding the fractional solutions and then performing swaps and merges of bundles to meet the desired fairness and efficiency criteria. We expect that the concept of ER equilibrium will play a crucial role in deriving further results on related problems. 
\end{abstract}

\newpage
\setcounter{tocdepth}{2}
\tableofcontents
\thispagestyle{empty}

\newpage
\setcounter{page}{1}
\section{Introduction}
Allocation problems frequently arise in various contexts such as task assignment, partnership dissolution, and the division of inheritance. The fair allocation of indivisible items has been extensively studied across multiple disciplines, including computer science, economics, social choice theory, and multi-agent systems. In this paper, we consider the fundamental problem of \textit{fairly} allocating a set of \textit{indivisible chores} --- items that impose a cost on agents who receive them --- to agents. Formally, given a set $M$ of $m$ indivisible chores and $n$ agents, our objective is to achieve an allocation, represented by a partition $\x = (\x_1, \dots, \x_n)$ of chores among agents, that satisfies a notion of \emph{fairness}. Here, $\x_i \subseteq M$ denotes the bundle of chores assigned to agent $i$. We assume that each agent $i$ has an additive disutility function, defined as $d_i(\x_i) = \sum_{j\in\x_i} d_{ij}$, where $d_{ij} > 0$ represents the \textit{disutility} agent $i$ incurs from receiving chore $j$.

Among various fairness criteria, \textit{envy-freeness} (\s{EF})~\cite{foleyEF} is the most natural concept, requiring that every agent (weakly) prefers their allocated items over those assigned to others. However, for indivisible items, an \s{EF} allocation may not always exist (e.g., if a single task must be allocated between two agents), prompting the need for relaxed notions of \s{EF} in discrete settings. A prominent and compelling relaxation is envy-freeness up to any item (\efx). 

\paragraph{EFX for chores.} An allocation of chores is said to be \efx if no agent $i$ envies another agent $h$ \textit{after} the counter-factual removal of \textit{any} single chore from $i$'s bundle. Thus, an allocation $\x$ is \efx if, for any pair of agents $(i, h)$, we have $d_i(\x_i \setminus \{j\}) \le d_i(\x_h)$ for every chore $j\in \x_i$. 

~\newline
This makes \efx the closest discrete analog of \s{EF}. While \efx allocations may not always exist for super-additive preferences~\cite{efxnonexistence2024}, the question of their existence for more than two agents with additive preferences remains a major open question in fair division.

The existence of \efx allocations has been established only for very specific instances. For example, they are known to exist when the number of chores does not exceed twice the number of agents~\cite{mahara2023efxmatching}, or when there are two types of chores\footnote{The instance only has multiple copies of two chores}~\cite{aziz2022twotypes} (see \cref{sec:related-work} for an expanded discussion). Instead of imposing such restrictions on the instances, another popular approach is to explore the existence of approximately-\efx allocations for \textit{all} instances. In an $\alpha$-EFX allocation of chores, the disutility of each agent after the removal of any chore from her bundle is at most $\alpha$ times her disutility for the chores assigned to any other agent, for some factor $\alpha\ge 1$. That is, in an $\alpha$-\efx allocation $\x$, for any pair of agents $(i, h)$, we have $d_i(\x_i \setminus \{j\}) \le \alpha\cdot d_i(\x_h)$ for every chore $j\in \x_i$. Currently, the best-known approximation is the existence of $O(n^2)$-\efx allocations, achieved through a sophisticated algorithm~\cite{zhou22efx}.
Our first main result improves this approximation factor to $4$.

\begin{restatable}{theorem}{thmEFXMain}\label{thm:efx-main}
Any chore allocation instance admits a $4$-\efx allocation.   
\end{restatable}

While a fair allocation ensures equitable distribution, it can compromise on overall efficiency. Ideally, we aim for allocations that are both fair and economically efficient. The standard criterion for economic efficiency is Pareto-optimality (\po): an allocation is \po if no re-allocation can make at least one agent better off without making any other agent worse off. Thus, we seek allocations that are simultaneously \po and satisfy some relaxation of envy-freeness. Given that establishing the existence of (approximate-)EFX allocations is already a challenging task without considering efficiency guarantees, and since verifying \po is known to be coNP-hard~\cite{keijzer}, we focus on structured instances or weaker envy-freeness relaxations to attain allocations that are both fair and efficient.

We consider \textit{bivalued} instances, where the disutility of any chore for any agent is one of two given positive numbers $\{a,b\}$. These instances model soft and hard preferences and have been widely studied in discrete fair division (e.g.,  \cite{amanatidis2020mnwefx,GM23TCS,ebadian2021bivaluedchores,Garg_Murhekar_Qin_2022,GMQ23chores,zhou22efx,aziz2019bivaluedpo}). For bivalued instances, it has been established that \efx and \po allocations exists for $n=3$ agents~\cite{GMQ23chores}. Additionally, for $n\ge 4$ agents, $(n-1)$-\efx allocations are known to exist, albeit without any efficiency guarantees~\cite{zhou22efx}. While \cref{thm:efx-main} already establishes the existence of $4$-EFX allocations for these instances, we improve this result to show the existence of $3$-EFX allocations, and importantly, we do so while also achieving \po.

\begin{restatable}{theorem}{thmefxpomain}\label{thm:efxpo-main}
Any chore allocation instance where agents have bivalued disutilities admits a $3$-\efx and \po allocation.
\end{restatable}

For general additive instances, we explore the concept of envy-freeness up to $k$ chores (EF$k$), which requires that for any two agents $i$ and $\ell$, the disutility of $i$ after the removal of \emph{some} $k$ chores assigned to her is at most her disutility for the chores assigned to $\ell$. It is important to note that while any arbitrary \po allocation is trivially \s{EF}$m$ for $m$ chores, it may not satisfy $\alpha$-\s{EF}$k$ for any $\alpha\ge 1$ with a constant $k$.
Moreover, although the existence of an \efone and \po allocation has been established for special instances, such as $n=3$ agents \cite{GMQ23chores}, bivalued disutilities~\cite{Garg_Murhekar_Qin_2022,ebadian2021bivaluedchores}, two types of chores \cite{aziz2022twotypes}, and three types of agents~\cite{GMQ24wef1po}, it remains unclear whether allocations that are $\alpha$-\s{EF}$k$ and PO exist for any $\alpha \geq 1$ and constant $k$ across all instances.
Our next result provides the first positive result in this direction by showing the existence of $2$-EF$2$ and \po allocations for all instances.

\begin{restatable}{theorem}{thmEFTwoMain}\label{thm:ef2-main}
Any chore allocation instance admits a $2$-\eftwo and \po allocation. More precisely, for every agent the allocation is either $2$-\efone or \eftwo. 
\end{restatable}

We note that the (non-)existence of an \efone and \po allocation of chores is another major open question in fair division. Indeed, an \efone and \po allocation may not be achievable in general, and our result of $2$-\eftwo and PO could represent the strongest attainable guarantee for general instances. 

\paragraph{Earning-restricted equilibrium.} To establish our results, we introduce a novel concept of \textit{earning-restricted} competitive equilibrium. In a (unrestricted) competitive equilibrium (CE) for chores, each agent $i$ aims to earn an amount $e_i > 0$ by performing chores in exchange for payment. Each chore $j$ pays an amount $p_j > 0$ for the completion of the chore; thus, if agent $i$ performs a fraction $x_{ij}\in[0,1]$ of chore $j$, she earns $p_j\cdot x_{ij}$ from $j$. An allocation $\x = \{x_{ij}\}_{i\in[n], j\in[m]}$ and a set of chore payments $\p = \{p_j\}_{j\in[m]}$ constitute a CE if all chores are fully allocated and each agent meets her earning requirement $e_i$, while only performing chores that yield a minimum disutility per unit of payment.

We define the concept of earning restriction (ER) by imposing a limit $c_j > 0$ on the collective earning that agents can derive from each chore $j$. Consequently, agents earn $\min\{p_j,c_j\}$ in total from chore $j$, and only a $\min\{p_j,c_j\}/p_j$ fraction of chore $j$ can be allocated. This model has a natural economic interpretation: an enterprise intends to execute a project $j$ that costs $p_j$, but due to financial constraints or contractual agreements, it can only allocate up to $c_j$ towards the completion of the project. 

Since earning-restricted equilibria generalize the concept of unrestricted CE, a natural approach is to try adapting existing CE existence proofs and algorithms~\cite{chaudhury22cechores,boodaghians22cechores,chaudhury2024competitiveequilibriumchoresdual} for ER equilibria. However, the earning restriction imposes strong constraints, making it challenging to directly apply the existing techniques. In fact, it is not immediately clear if an ER equilibrium exists at all. Indeed, if agents demand more aggregate earning than can be disbursed by performing all of the chores, an ER equilibrium cannot exist. Our next main result establishes existence under a natural condition: 

\begin{restatable}{theorem}{thmerexistence}\label{thm:er-existence}
An earning-restricted competitive equilibrium exists if and only if $\sum_i e_i \le \sum_j c_j$.
\end{restatable}

In our fair division results, we specifically set each agent's earning requirement $e_i = 1$ and use a uniform earning limit $c_j$ across all chores. However, Theorem~\ref{thm:er-existence} establishes the existence of ER competitive equilibrium for all instances with arbitrary agent requirements $e_i > 0$ and chore limits $c_j > 0$, provided that the aforementioned condition is met. One of the key steps in our approach involves rounding a fractional ER equilibrium to obtain an integral, Pareto-optimal solution. We show that the fractional ER equilibrium reveals more information about agent preferences compared to  an unrestricted competitive equilibrium for the same instance, making it particularly valuable for computing fair allocations.

All our existence results are accompanied by polynomial-time algorithms that require an ER equilibrium as input. Consequently, if ER equilibria can be computed (even approximately) in polynomial time, our existence results (Theorems~\ref{thm:efx-main}, \ref{thm:efxpo-main}, and \ref{thm:ef2-main}) would directly translate into polynomial-time algorithms. Towards this, we show that an ER equilibrium can indeed be computed in polynomial time when the number of agents is constant, which implies polynomial-time algorithms for all our results in this case. 

\begin{restatable}{theorem}{thmConstER}\label{thm:const-er}
An earning-restricted equilibrium can be computed in polynomial time when the number of agents is constant.
\end{restatable}

\paragraph{Additional results.} We expect that the concept of earning-restricted equilibrium and the techniques developed to derive our main results will find broader applications. Building on these techniques, we derive the following additional results:
\begin{itemize}[leftmargin=*]
\item Existence of an allocation that is $(n-1)$-\efone and \po for all chore instances (\cref{thm:ef1-main}).
\item A polynomial time algorithm for computing a \po allocation that is \textit{balanced}, i.e., the number of chores assigned to each agent differs by at most one (\cref{thm:balanced}).  
\item A polynomial time algorithm for computing an \efx allocation when the number of chores is at most twice the number of agents (\cref{thm:efx-small}). While this result is already known \cite{mahara21efx}, we present a new algorithm that is faster, arguably simpler, and provides specific properties that are crucial for obtaining \cref{thm:efx-main}.
\item A polynomial time algorithm for computing an \efx and \po allocation for bivalued instances where the number of chores is at most twice the number of agents (\cref{thm:bivalued-small}).
\end{itemize}

\subsection{Related Work}\label{sec:related-work}
We discuss other related literature that is most relevant to \efone/\efx and \po allocations. For additional related work, we refer the reader to excellent surveys \cite{aziz2022survey,amanatidis2023survey,liu2024mixed}. 

The existence and polynomial time computation of EFX allocations of chores is known for the following special cases: (i) two agents, (ii) instances where agents have identical preference order (IDO) over the chores \cite{boli2022wpropx}, (iii) two types of chores \cite{aziz2022twotypes}, (iv) the number of chores is at most twice the number of agents \cite{mahara2023efxmatching}, (v) all but one agent have IDO disutility functions \cite{mahara2023efxmatching}, and (vi) there are three agents with 2-ary disutilities \cite{mahara2023efxmatching}. For $n=3$ agents with bivalued disutility functions, \cite{GMQ23chores} proved an \efx and \po allocation can be computed in polynomial time. Recently, \cite{afshinmehr2024approximateefxexacttefx} showed the existence of $2$-\efx allocations for $n=3$ agents. The existence and polynomial time computation of EF1 and PO allocations of chores is also known for (i) two agents \cite{aziz2019chores} and (ii) three types of agents \cite{GMQ24wef1po}. 

\paragraph{Results on goods.} Goods are items that provide non-negative utilities to agents receiving them. For goods, \efx allocations are known to exist for two agents (via cut-or-choose), identically ordered (IDO) instances \cite{plaut2018efx}, three agents \cite{chaudhury2020efx}, two types of agents~\cite{mahara21efx}, two types of goods~\cite{gorantla23efxmultiset} and three types of agents~\cite{v2024efxexiststypesagents}. However, the existence of EFX allocations is open for $n\ge 4$ agents. In terms of approximation, the best known result is the existence of $0.618$-\efx allocations \cite{amanatidis2020efxmms}. Recent works \cite{barman2023parameterizedguaranteesenvyfreeallocations, amanatidis2024pushingfrontierapproximateefx} improve the approximation guarantee beyond $0.618$ in certain special cases. Another relaxation for achieving EFX allocations is charity, where some goods are left unallocated; see e.g.,~\cite{caragiannis2019efxmnw,efxcharity2021,berger2022efx,akrami23efxsimpler}.

The existence of EF1 and PO allocations is known for goods; \cite{caragiannis16nsw-ef1} showed that an allocation with maximum Nash welfare (product of agent utilities) satisfies both EF1 and PO. However, computing such an allocation is NP-hard \cite{keijzer}, even approximately \cite{lee2015nsw-apx,garg2017nswhardness}. Bypassing this hardness, \cite{Barman18FFEA} gave a pseudo-polynomial time algorithm for computing an EF1 and PO allocation. Polynomial time computation is open in full generality but is known for binary instances \cite{barman2018binarynsw} and for a constant number of agents \cite{GM24jair}. A balanced PO allocation is known to exist and can be computed in polynomial time \cite{conitzer2017publicdecision}. For bivalued instances, an EFX and PO allocation can be computed in polynomial time \cite{GM23TCS}.

The concept analogous to ER equilibrium for goods is spending-restricted (SR) equilibrium \cite{cole2015nswapprox}, which imposes a limit on the amount agents can collectively spend on a good. There are crucial differences between SR and ER equilibria, both in terms of computation and applications. First, similar to CE for goods, SR equilibria can be computed efficiently: they are captured by a convex program formulation, and there are polynomial-time flow-based algorithms for its computation~\cite{coledevanur2017SR,cole2015nswapprox}. Second, SR equilibria have mainly been used to develop approximation algorithms for maximum Nash welfare \cite{cole2015nswapprox,coledevanur2017SR} and to the best of our knowledge have not found applications for computing envy-based fair allocations of goods. Furthermore, while all applications of SR equilibria utilize the same spending limit of $\beta = 1$, our results leverages ER equilibria with earning limits set at $\beta=\frac{1}{2}$, introducing additional challenges.

\paragraph{Share-based fairness notions.} Apart from envy-based fairness, several works study share-based fairness, where a fair allocation gives every agent some amount of (dis-)utility irrespective of other agents' bundles. Examples of such fairness notions include the proportional share (e.g., \cite{barman2019prop1po,aziz20propchores,boli2022wpropx}), maximin share (e.g., \cite{huang2023reduction, FeigeST21, ghodsi2018fair, huang21mmschores,AkramiG24,barman2020approximation}), any-price share (e.g.,  \cite{babaioff21aps,babaioff22ecshare}), and quantile shares \cite{babichenko24quantileshare}.

\paragraph{\normalfont \textbf{Organization of the remainder of the paper.}}
\cref{sec:overview} provides technical overview of the main results. \cref{sec:prelim} defines the problem formally and introduces the earning-restricted model. \cref{sec:rounding} presents our results on the existence of allocations that are (i) $2$-\eftwo and \po and (ii) $(n-1)$-\efone and \po (deferred to \cref{app:efone}). \cref{sec:efx} presents our algorithms proving the existence $4$-\efx allocations. The case of $m\le 2n$ is discussed in \cref{sec:efx-small}, while the general case is discussed in  \cref{sec:efx-overview} and \cref{sec:efx-analysis}. \cref{sec:bivalued} presents the existence of $3$-\efx and \po allocations for bivalued instances, with the case of $m\le 2n$ deferred to \cref{app:bivalued-small}. Finally, we show the existence of ER equilibria in \cref{sec:er}, with \cref{sec:er-const} showing polynomial time computation for constantly many agents. \cref{app:examples} contains illustrative examples. 

\section{Technical Overview}\label{sec:overview}
\subsection{ER Equilibria for Fair and Efficient Allocation}
Competitive equilibrium (CE) is a well-established solution concept for efficient allocation. A fractional allocation $\x = \{x_{ij}\}_{i\in[n], j\in [m]}$ and a set of chore payments $\p \in \R_{> 0}^m$ together constitute a CE if (i) every agent $i$ receives a bundle that minimizes her disutility among all bundles and payments that meet her earning requirement, and (ii) all chores are allocated. For additive disutilities, the first condition is equivalent to requiring that $(\x,\p)$ satisfies the \textit{minimum pain-per-buck (MPB) condition}:
\[
x_{ij} > 0 \Rightarrow \frac{d_{ij}}{p_j} = \min_{c\in[m]} \frac{d_{ic}}{p_c}. 
\]

The First Welfare theorem \cite{mas1995microeconomic} guarantees that for any $(\x, \p)$ satisfying the above MPB condition, the allocation $\x$ is Pareto-optimal. Moreover, the MPB condition also allows one to express the disutility $d_i(\x_i)$ of an agent in the CE $(\x,\p)$ in terms of the money $\p(\x_i) = \sum_j x_{ij}\cdot p_j$ she earns in the CE as $d_i(\x_i) = \alpha_i \cdot \p(\x_i)$, where $\alpha_i = \min_{c\in[m]} d_{ic}/p_c$ is the MPB ratio of $i$. In other words, the earning of an agent serves as a proxy for her disutility in the allocation. With this, roughly speaking, it suffices to balance agents' earnings to obtain a fair allocation, while the allocation being a CE ensures that it is \po. Indeed, a competitive equilibrium $(\y,\p)$ with equal earnings (CEEE) obtained by setting $e_i = 1$ for all agents is envy-free and \po. Since we are interested in finding allocations of indivisible chores, a natural approach would be to \textit{round} the fractional CEEE solution $\y$ to obtain a rounded allocation $\x$ which retains its fairness and efficiency properties. The rounding preserves the MPB condition, ensuring that $(\x,\p)$ is a CE and implying $\x$ is \po. However, the following example illustrates that no rounding of a CE can give any approximation to envy-freeness.

\begin{example}\label{ex:er-eq} 
\normalfont
Consider an instance with three agents $\{a_1,a_2,a_3\}$ and four chores $\{j_1, j_2, j_3, j_4\}$, with disutilities given in the following table.
\begin{center}
\begin{tabular}{ |c||c|c|c|c| } 
\hline
& $j_1$ & $j_2$ & $j_3$ & $j_4$\\\hhline{|=#=|=|=|=|}
$a_1$ & 2 & 1 & 2 & 2 \\\hline
$a_2$ & 4 & 1 & 1 & 2 \\\hline
$a_3$ & 9 & 2 & 1 & 1 \\\hline
\end{tabular}
\end{center}
Suppose each agent has an earning requirement of $e_i = 1$. Consider the allocation $\x = (\x_1, \x_2, \x_3)$ given by $\x_1 = \{\frac{1}{2}j_1\}$, $\x_2 = \{\frac{1}{2}j_1\}$ and $\x_3 = \{j_2, j_3, j_4\}$, and payments $\p = (2, 0.5, 0.25, 0.25)$. It can be checked that $(\x,\p)$ is a CE, and is illustrated in \cref{fig:er-example}. However note that any rounding of $(\x,\p)$ leaves some agent $a_i$ for $i\in\{1,2\}$ with no chores, which causes $a_3$ to have multiplicatively unbounded envy towards $a_i$. A generalization of the above example shows that no rounding of a CE with equal agent earnings can guarantee $\alpha$-\s{EF}$k$, for any $\alpha,k\ge 1$ (see \cref{ex:ce-rounding}). 
\end{example}

\begin{figure}
\centering 
\begin{tikzpicture}[
    agent/.style={circle, draw, fill=blue!70, text=white, minimum size=0.7cm, font=\small},
    job/.style={rectangle, draw, fill=red!60, text=white, minimum size=1cm, font=\large},
    job2/.style={rectangle, draw, fill=green!50, text=black, minimum size=0.7cm, font=\large},
    arrow/.style={-Stealth, line width=1pt, blue!80},
    edge/.style={thick}
]

\node[agent] (a1) at (0, 3) {$a_1$};
\node[agent] (a2) at (0, 1.5) {$a_2$};
\node[agent] (a3) at (0, 0) {$a_3$};
\node[job] (j1) at (3, 3) {$j_1$};
\node[job2] (j2) at (3, 1.5) {$j_2$};
\node[job2] (j3) at (3, 0) {$j_3$};
\node[job2] (j4) at (3, -1.5) {$j_4$};

\draw[edge] (a1) -- node[above] {1} (j1);
\draw[edge] (a2) -- node[above] {1} (j1);
\draw[edge] (a3) -- node[above] {0.5} (j2);
\draw[edge] (a3) -- node[above] {0.25} (j3);
\draw[edge] (a3) -- node[above] {0.25} (j4);

\node at (4, 3) {2};
\node at (4, 1.5) {0.5};
\node at (4, 0) {0.25};
\node at (4, -1.5) {0.25};

\node[agent] (a1_r) at (8, 3) {$a_1$};
\node[agent] (a2_r) at (8, 1.5) {$a_2$};
\node[agent] (a3_r) at (8, 0) {$a_3$};
\node[job] (j1_r) at (11, 3) {$j_1$};
\node[job2] (j2_r) at (11, 1.5) {$j_2$};
\node[job2] (j3_r) at (11, 0) {$j_3$};
\node[job2] (j4_r) at (11, -1.5) {$j_4$};

\draw[edge] (a1_r) -- node[above] {1} (j1_r);
\draw[edge] (a2_r) -- node[above] {2/3} (j2_r);
\draw[edge] (a2_r) -- node[above] {1/3} (j3_r);
\draw[edge] (a3_r) -- node[above] {1/3} (j3_r);
\draw[edge] (a3_r) -- node[above] {2/3} (j4_r);

\node at (12, 3) {4/3};
\node at (12, 1.5) {2/3};
\node at (12, 0) {2/3};
\node at (12, -1.5) {2/3};

\node[align=center, font=\small] at (4, 4) {\underline{$p_j$}};
\node[align=center, font=\small] at (1.45, 4) {\underline{$q_{ij}$}};
\node[align=center, font=\small] at (12, 4) {\underline{$p_j$}};
\node[align=center, font=\small] at (9.45, 4) {\underline{$q_{ij}$}};

\node[align=center] at (1.45, -2.5) {\underline{CE with $e_i=1$}};
\node[align=center] at (9.85, -2.5) {\underline{ER with $e_i=1$, $c_j=1$}};

\end{tikzpicture}
\caption{Illustrating the difference between an unrestricted CE and an ER equilibrium of the instance from \cref{ex:er-eq}. The chore payments $p_j$ are indicated to the right of the chore, and chore specific earnings $q_{ij}$ are indicated above the edges between agent $i$ and chore $j$. Note that in the unrestricted CE, chore $j_1$ pays out $1$ each to agents $a_1$ and $a_2$, whereas in the ER equilibrium the earning restriction of $1$ forces agent $a_2$ to do chores $j_2$ and $j_3$ in the ER equilibrium.}
\label{fig:er-example}
\end{figure}

\paragraph{ER equilibria for fair and efficient allocation.} 
The example highlights the main issue with rounding an unrestricted CE: the presence of high paying chores ($j_1$) which results in agents ($a_1$ and $a_2$) meeting their earning requirement by only doing such chores. An earning-restricted CE circumvents this issue by placing a limit $c_j$ on the amount that chore $j$ can disburse to the agents. Such an earning restriction on a lucrative chore forces agents to seek out less lucrative chores in the equilibrium. In doing so, the ER equilibrium reveals more information about agent preferences than an unrestricted equilibrium. For a concrete example, we present an ER equilibrium for the instance of \cref{ex:er-eq} with limits $c_j = 1$ for all chores in \cref{fig:er-example}. We see that in the ER equilibrium, $a_2$ is allocated her favorite chores $j_2$ and $j_3$, while she is allocated her least favorite chore $j_1$ in the unrestricted CE. Moreover, it can be checked that any rounding of the ER equilibrium results in an EFX and PO allocation for the instance! This highlights the utility of ER equilibria for fair and efficient chore allocation. In \cref{thm:er-existence}, we establish the existence of ER equilibria under the \textit{feasible earning} condition $\sum_i e_i \le \sum_j c_j$. Next, we discuss applications of ER equilibria for fair and efficient chore allocation. 

\paragraph{\cref{alg:rounding}: 2-EF2 and PO allocation.}
We use ER equilibria to obtain a 2-\eftwo and PO allocation for all chore allocation instances (\cref{thm:ef2-main}). We uniformly set agent earning requirements to $e_i = 1$, and chore earning limits to $\beta = \frac{1}{2}$. We assume $m \ge 2n$ to ensure that the feasible earning condition is satisfied and an ER equilibrium exists. Our polynomial time algorithm --- \cref{alg:rounding} --- rounds an ER equilibrium $(\z,\p)$ of such an instance to an integral allocation $(\x,\p)$ with certain bounds on the earnings of the agents. Recall that in the fractional solution $(\z,\p)$ the earning of every agent equals $1$. 

\cref{alg:rounding} partitions the set of chores into low paying chores $L = \{j: p_j \le \frac{1}{2}\}$, and high paying chores $H = \{j: p_j > \frac{1}{2}\}$, and allocates them separately to ensure that in the resulting integral allocation $(\x,\p)$:
\begin{itemize}[leftmargin=*]
\item The earning of every agent is at least $\frac{1}{2}$. Thus, no agent loses too much earning due to the rounding. \cref{alg:rounding} achieves this by ensuring that every agent loses chores which pay at most $\frac{1}{2}$ in total, or receives a chore from $H$ which pays at least $\frac{1}{2}$.
\item The earning of every agent is at most $1$ after the removal of her two highest paying chores. Thus, no agent receives too much earning due to the rounding. \cref{alg:rounding} achieves this by ensuring that every agent gets at most two chores from $H$.
\end{itemize}
By establishing the above bounds on the earnings of agents, we conclude that $(\x,\p)$ is a $2$-\eftwo and \po allocation. Moreover, we remark that our rounding algorithm is tight: \cref{ex:ef2-tight} shows an instance where no rounding of an ER equilibrium with $\beta=\frac{1}{2}$ can return a $(2-\delta)$-\eftwo and \po allocation, for any $\delta > 0$. 

\paragraph{\cref{alg:balanced}: Balanced PO allocation.} To address the case of $m\le 2n$ left out by the above approach, we design \cref{alg:balanced}: a polynomial time algorithm that gives a balanced \po allocation for any number of chores, i.e., every agent gets roughly $\frac{m}{n}$ chores. When $m\le 2n$, a balanced allocation is \eftwo, since every agent gets at most two chores. Similarly, when $m\le n$, a balanced allocation is \efone, since every agent gets at most one chore. 

\cref{alg:balanced} starts with an imbalanced allocation and repeatedly tries to transfer chores from the agent with the highest number of chores to an agent with the least number of chores, until the allocation becomes balanced. \cref{alg:balanced} performs such transfers while ensuring that the allocation is PO. Doing so requires carefully raising the payments for a subset of chores so that the MPB condition can be maintained before and after transfers. To show termination in polynomial time, we prove that there can be at most $n$ payment raises between two transfer steps, and at most $m$ transfer steps overall before the allocation is balanced.

\paragraph{Algorithms for approximately-EF1 and PO allocations.} Having shown the existence of 2-\eftwo and \po allocations, an important next question is investigating the existence of $\alpha$-EF1 and \po allocations for $\alpha\ge 1$. Following the ideas developed in \cref{sec:errounding} which round a fractional solution of an ER equilibrium, the natural approach towards obtaining an approximate-\efone guarantee would be to ensure that every agent gets at most one high paying chore in the rounded solution. Clearly, this requires the number of high paying chores to be at most $n$, which cannot be guaranteed for any earning limit $\beta < 1$. However, \cref{alg:rounding} cannot ensure good lower bounds on the agent earnings when $\beta = 1$. 

To fix this, we design \cref{alg:efone-rounding} by modifying the rounding procedure of \cref{alg:rounding}. For $m\ge n$, an ER equilibrium with $\beta=1$ exists. Given such an equilibrium $(\y,\p)$, \cref{alg:efone-rounding} defines $L$ to be the set of chores with payment at most $\frac{\beta}{2} = \frac{1}{2}$, and $H$ to be the set of chores with payment exceeding $\frac{1}{2}$. We prove that by using the same rounding procedure but with $L$ and $H$ defined this way, we obtain an integral MPB allocation $(\x,\p)$ where $\min_{j\in \x_i} \p(\x_i\setminus \{j\}) \le 1$ and $\p(\x_i) \ge \frac{1}{2(n-1)}$ for all agents $i\in N$. This implies that $\x$ is $2(n-1)$-\efone and \po. Finally, for $m\le n$ note that \cref{alg:balanced} returns an EF1 and \po allocation in polynomial time.

We improve upon our result by designing another algorithm in \cref{app:rebalancing} which returns an $(n-1)$-\efone and \po allocation. Essentially, our improved algorithm \textit{unrolls} \cref{alg:efone-rounding} and carefully identifies the events which caused the earning of an agent to drop below $\frac{1}{n-1}$. We argue that this must have happened due to sub-optimal rounding choices, and that they can be identified and corrected in polynomial time. Our algorithm thus returns an allocation in which every agent earns at least $\frac{1}{n-1}$, and at most $1$ up to the removal of one chore, and is therefore $(n-1)$-\efone and \po. Once again, we remark that our algorithm is tight: \cref{ex:ef1-tight} shows an instance where no rounding of the ER equilibrium with $\beta = 1$ is $(n-1-\delta)$-\efone, for any $\delta > 0$.

\subsection{Existence of 4-\efx Allocations}
Our main result showing the existence of 4-\efx allocations relies on two algorithms: \cref{alg:efx-small} which computes an \efx allocation for instances with $m \le 2n$ in polynomial time, and \cref{alg:const-efx} which computes a 4-\efx allocation for instances with $m\ge 2n$. We discuss \cref{alg:efx-small} later and focus on instances with $m \ge 2n$.

\subsubsection{\cref{alg:const-efx}: 4-EFX Allocation} \cref{alg:const-efx} is our most technically involved algorithm and relies on several novel ideas to obtain the existence of constant approximation of EFX. 

First, given an ER equilibrium of such an instance with $e_i=1$ and $\beta=\frac{1}{2}$, we compute a $2$-\eftwo and \po allocation $(\x,\p)$ using \cref{alg:rounding}. As before, we partition the chores into two sets based on their payments: the set $L$ of low paying chores with payment at most $\frac{1}{2}$, and the set $H$ of high paying chores with payment greater than $\frac{1}{2}$. We then partition the bundle of each agent $i$ as $\x_i = S_i \cup H_i$, where $S_i \subseteq L$ contains low paying chores and $H_i\subseteq H$ contains high paying chores. Let us partition the set of agents as $N = N_H \sqcup N_0$, where agents in $N_H$ receive one or two high paying chores while agents in $N_0$ receive none. From the analysis of \cref{alg:rounding}, we obtain the following lower and upper bounds on the earning of agents:
\begin{itemize}
\item $\p(\x_i)\ge \frac{1}{2}$ for all agents $i\in N$,
\item $\p(S_i) \le 1$ for all agents in $N_H$,
\item $\p(S_i) \le 2$ for all agents in $N_0$\footnote{The earning of any agent is at most $1$ up to the removal of their two highest paying chores. For agents in $N_0$, these chores are not in $H$, and their payment is at most $\frac{1}{2}$. Then, for such an agent $i$, $\p(S_i) \le 1 + 2\cdot\frac{1}{2} = 2$.}.
\end{itemize}

Thus, we see that agents in $N_0$ are already $4$-\efx! This indicates that the allocation $\x$ may not be $O(1)$-\efx due to agents in $N_H$. Recall that such agents have either one or two high paying chores from $H$. For simplicity, let us assume for the moment that all agents in $N_H$ are assigned a single high paying chore. We introduce the idea of \textit{`chore swaps'} to address the \efx-envy of such agents.

\begin{figure}
\centering 
\begin{tikzpicture}[
    node_style/.style={circle, draw, fill=blue!70, text=white, minimum size=0.7cm, font=\large},
    job_style/.style={rectangle, draw, fill=red!60, text=white, minimum size=1cm, font=\large},
    bar/.style={rectangle, draw, fill=green!50, minimum width=0.4cm, minimum height=0.8cm},
    bar2/.style={rectangle, draw, fill=yellow!50, minimum width=0.4cm, minimum height=0.8cm},
    edge/.style={thick},
    brace/.style={decorate, decoration={brace, amplitude=5pt}},
]

\node[node_style] (i) at (3, 3) {$i$};
\node[job_style] (j) at (4.5, 3) {$j_i$}; 
\node[above] at (4.5, 3.5) {$p_{j_i} > \frac{1}{2}$};
\node[node_style] (l) at (3, 1) {$\ell$};
\node[job_style] (j_l) at (4.5, 1) {$j_\ell$};

\foreach \x in {0,0.5,1,1.5,2} {
    \node[bar] at (\x, 3) {};
}

\foreach \x in {0.5,1,1.5,2} {
    \node[bar2] at (\x, 1) {};
}

\draw[decorate, decoration={brace}] (-0.25, 3.5) -- (2.25, 3.5) node[midway, above] {$\p(S_i)\le 1$};
\draw[decorate, decoration={brace,mirror}] (0.2, 0.4) -- (5.2, 0.4) node[midway, below] {$\p(S_i) + p_{j_i} > 3 \cdot d_i(\x_\ell)$};

\draw[edge] (j) to (i);
\draw[edge] (l) -- (2.2,1);
\draw[edge] (i) -- (2.2,3);
\draw[edge] (l) -- (j_l);

\node[node_style] (i_r) at (12, 3) {$i$};
\node[job_style] (j_r) at (14, 1) {$j_i$};
\node[node_style] (l_r) at (12, 1) {$\ell$};
\node[job_style] at (15.5, 3) {$j_\ell$};
\foreach \x in {9,9.5,10,10.5,11} {
    \node[bar] at (\x, 3) {};
}

\foreach \x in {13,13.5,14,14.5} {
    \node[bar2] at (\x, 3) {};
}

\draw[decorate, decoration={brace}] (8.75, 3.5) -- (11.25, 3.5) node[midway, above] {$\p(S_i)\le 1$};
\draw[decorate, decoration={brace}] (12.7, 3.6) -- (16.1, 3.6) node[midway, above] {$d_i(\x_\ell) < \frac{\p(S_i) + p_{j_i}}{3}$};

\draw[edge] (i_r) -- (11.2,3);
\draw[edge] (j_r) -- (l_r);
\draw[edge] (i_r) -- (12.8,3);
\node[below] at (14, 0.5) {$p_{j_i} > \frac{1}{2}$};

\node[align=center] at (3, -0.8) {\underline{Allocation before $(i,\ell)$ swap}};
\node[align=center] at (12, -0.8) {\underline{Allocation after $(i,\ell)$ swap}};

\draw[->, thick] (6.1,1.7) -- (8.4,1.7) node[midway, above] {$(i,\ell)$ swap};

\end{tikzpicture}
\caption{Illustrating an $(i,\ell)$ chore swap. In the allocation before the swap, the agent $i\in N_H$ with a single high paying chore $j_i$ envies agent $\ell$ the most. The swap transfers the entire bundle of agent $\ell$ to agent $i$, and transfers the single chore $j_i$ to agent $\ell$. Using the bounds on payments and disutilities one can argue that after the $(i,\ell)$ swap, the agents are $3$-\efx.}
\label{fig:choreswap}
\end{figure}

\paragraph{Chore swaps.} Consider an agent $i$ with $H_i = \{j_i\}$, who is not $3$-\efx in  allocation $\x$, and let $\ell$ be the agent who $i$ envies the most, i.e. $\ell = \arg\min\{h\in N: d_i(\x_h)\}$. An $(i,\ell)$ swap on the allocation $\x$ results in an allocation $\x'$ obtained by transferring all the chores of $\ell$ to $i$ and transferring the chore $j_i$ from $i$ to $\ell$. That is, $\x'_i = \x_i \cup \x_\ell \setminus \{j_i\}$, $\x'_\ell = \{j_i\}$, and $\x'_h = \x_h$ for all $h\neq \{i,\ell\}$. See \cref{fig:choreswap} for an illustration of a chore swap.

We prove that an $(i,\ell)$ chore swap locally resolves the $3$-\efx-envy of agent $i$; that is, agent $i$ is 3-\efx after the swap. Let $\alpha_i$ be the MPB ratio of agent $i$. Now observe that agent $i$ is  $3$-\efx towards all agents $h\neq \ell$ after the swap:
\begin{align*}
d_i(\x'_i) &= d_i(S_i) + d_i(\x_\ell) \\
&= \alpha_i\p(S_i) + d_i(\x_\ell) \tag{using the MPB condition} \\ 
&\le \alpha_i + d_i(\x_\ell) \tag{since $\p(S_i) \le 1$} \\
&\le 2\alpha_i\p(\x_\ell) + d_i(\x_\ell) \tag{since $\p(\x_\ell) \ge \frac{1}{2}$} \\
&\le 3\cdot d_i(\x_\ell) \tag{using the MPB condition}\\
&\le 3\cdot d_i(\x_h) \tag{by choice of agent $\ell$} \\
&= 3\cdot d_i(\x'_h).
\end{align*}

Similarly, the fact that $i$ is $3$-\efx envious of the bundle $\x_\ell$ establishes a lower bound on the disutility of $j_i$ for $i$, which we can use to prove that $i$ will not $3$-\efx envy $\x'_\ell = \{j_i\}$ after the swap. 
\begin{align*}
d_i(\x'_i) &= d_i(S_i) + d_i(\x_\ell) \\
&< d_i(S_i) + \frac{d_i(\x_i)}{3} \tag{since $i$ is not $3$-\efx towards $\ell$} \\
&= \frac{4}{3} d_i(S_i) + \frac{1}{3} d_i(j_i) \tag{since $\x_i = S_i \cup \{j_i\}$}\\
&= \frac{4}{3} \alpha_i \p(S_i) + \frac{1}{3} \alpha_i p_{j_i} \tag{using the MPB condition}\\
&< \frac{8}{3} \alpha_i p_{j_i} + \frac{1}{3} \alpha_i p_{j_i} \tag{using $\p(S_i) \le 1 < 2\cdot p_{j_i}$, since $i\in N_H$ and $j_i\in H$} \\
&= 3 \alpha_i p_{j_i} \\
&\le 3\cdot d_i(\x'_\ell). \tag{using the MPB condition}
\end{align*}
In conclusion, after an $(i,\ell)$ chore swap, agent $i$ is $3$-\efx, i.e., the $3$-\efx envy of agent $i$ is temporarily resolved. See \cref{fig:choreswap} for an illustration of the above arguments. Moreover, agent $\ell$ is \efx after the swap since she has a single chore. The above idea suggests repeatedly performing chore swaps until the allocation is $3$-\efx. 

However, an agent $i$ who underwent a swap may re-develop $O(1)$-\efx-envy subsequently in the run of the algorithm! Consider an $(i,\ell)$ swap performed between an agent $i\in N_H$ who was not $3$-\efx and the agent $\ell$ who $i$ envied the most. This resulted in an allocation $\x'$ in which the high paying chore $j_i\in H$ of $i$ was transferred to $\ell$. Now consider a subsequent swap $(h, k)$ between different agents $h\in N_H$ and $k\in N$, after which the high paying chore $j_h \in H$ of agent $h$ is (the only chore) assigned to $k$. Roughly speaking, since $i$ was $3$-\efx after the $(i,\ell)$ swap, $i$ does not $3$-\efx envy $k$'s bundle before the swap. Hence, $i$ will not envy $h$ after the $(h,k)$ swap. However, it could happen that $i$ develops $O(1)$-\efx envy towards $k$ after the $(h,k)$ swap if $d_i(j_h) < O(1)\cdot d_i(\x'_i)$.

\paragraph{Ordering the chore swaps.} To ensure this does not happen, \textit{our algorithm performs chore swaps in a carefully chosen order}. Recall that we argued that after the swap $i$ is $3$-\efx towards $\ell$. Thus, $d_i(\x'_i) \le 3\cdot d_i(j_i)$. If we had that $d_i(j_i) \le d_i(j_h)$, then we could show that $i$ remains $3$-\efx towards $k$ in the allocation $\x''$ after the $(h,k)$ swap as follows:
\[
d_i(\x''_i) = d_i(\x'_i) \le 3 d_i(j_i) \le 3 d_i(j_h) = 3 d_i(\x''_k).
\]
This observation suggests that for \textit{any} two agents $i,h\in N_H$ such that $i$ undergoes a swap before $h$, we should aim to have $d_i(j_i) \le d_i(j_h)$. To meet this strong condition comparing the disutilities of the high paying chores, we prove that it is \textit{sufficient to perform swaps in order of the payments of the high paying chores}. That is, at each time step $t$, among all the agents in $N_H$ who are not $3$-\efx, we pick the agent $i$ with the high paying chore with the minimum payment $p_{j_i}$ and perform an $(i, \ell)$ swap. An involved analysis shows that this design choice ensures \cref{alg:const-efx} does not cause an agent in $N_H$ to re-develop $3$-\efx envy. Clearly, this procedure terminates in at most $n$ steps with a $3$-\efx allocation.

\begin{figure}
\centering 
\begin{tikzpicture}[
    node_style/.style={circle, draw, fill=blue!70, text=white, minimum size=0.7cm, font=\large},
    job_style/.style={rectangle, draw, fill=red!60, text=white, minimum size=0.8cm, font=\large},
    bar/.style={rectangle, draw, fill=green!50, minimum width=0.4cm, minimum height=0.5cm},
    brace/.style={decorate, decoration={brace, amplitude=5pt}},
]

\node[node_style] (i_1) at (2.5, 3) {$i_1$};
\node[node_style] (i_2) at (2.5, 2) {$i_2$};
\node[node_style] (i_3) at (2.5, 1) {$i_3$};
\node[node_style] (i_4) at (2.5, 0) {$i_4$};

\node[job_style] (j_1) at (3.75, 3) {$j_1$};
\node[job_style] (j_2) at (4.75, 3) {$j_2$};

\node[job_style] (j_3) at (3.75, 2) {$j_3$};

\node[job_style] (j_4) at (3.75, 1) {$j_4$};
\node[job_style] (j_5) at (4.75, 1) {$j_5$};

\foreach \x in {0.5,1,1.5} {
    \node[bar] at (\x, 3) {};
}

\foreach \x in {1,1.5} {
    \node[bar] at (\x, 2) {};
}

\foreach \x in {0.5,1,1.5} {
    \node[bar] at (\x, 1) {};
}

\foreach \x in {0,0.5,1,1.5} {
    \node[bar] at (\x, 0) {};
}

\draw (i_1) -- (j_1);
\draw (i_2) -- (j_3);
\draw (i_3) -- (j_4);

\draw (i_1) -- (1.7,3);
\draw (i_2) -- (1.7,2);
\draw (i_3) -- (1.7,1);
\draw (i_4) -- (1.7,0);

\draw[decorate, decoration={brace}] (-0.2, 3.6) -- (1.8, 3.6) node[midway, above] {$L$};
\draw[decorate, decoration={brace}] (3.2, 3.6) -- (5.3, 3.6) node[midway, above] {$H$};

\node[node_style] (i_1_r) at (12.5, 3) {$i_1$};
\node[node_style] (i_2_r) at (12.5, 2) {$i_2$};
\node[node_style] (i_3_r) at (12.5, 1) {$i_3$};
\node[node_style] (i_4_r) at (12.5, 0) {$i_4$};

\node[job_style] (j_1_r) at (13.75, 3) {$j_5$};

\node[job_style] (j_2_r) at (13.75, 2) {$j_1$};
\node[job_style] (j_3_r) at (14.75, 2) {$j_3$};

\node[job_style] (j_4_r) at (13.75, 1) {$j_4$};
\node[job_style] (j_5_r) at (13.75, 0) {$j_2$};

\foreach \x in {10.5,11,11.5} {
    \node[bar] at (\x, 3) {};
}

\foreach \x in {11,11.5} {
    \node[bar] at (\x, 2) {};
}

\foreach \x in {10.5,11,11.5} {
    \node[bar] at (\x, 1) {};
}

\foreach \x in {10,10.5,11,11.5} {
    \node[bar] at (\x, 0) {};
}

\draw[decorate, decoration={brace}] (9.8, 3.6) -- (11.8, 3.6) node[midway, above] {$L$};
\draw[decorate, decoration={brace}] (13.3, 3.6) -- (15.3, 3.6) node[midway, above] {$\z'$: EFX};

\draw (i_1_r) -- (j_1_r);
\draw (i_2_r) -- (j_2_r);
\draw (i_3_r) -- (j_4_r);
\draw (i_4_r) -- (j_5_r);

\draw (i_1_r) -- (11.7,3);
\draw (i_2_r) -- (11.7,2);
\draw (i_3_r) -- (11.7,1);
\draw (i_4_r) -- (11.7,0);

\draw[->, thick] (5.5,1.75) -- (9.7,1.75) node[midway, above] {EFX re-allocation of $H$};

\end{tikzpicture}
\caption{Illustrating the re-allocation of chores in $H$.}
\label{fig:realloc}
\end{figure}

\paragraph{Handling agents with two high paying chores.} In the above discussion, we ignored agents in $N_H$ with two high paying chores. One may try to perform two chore swaps for each such agent. However it is not clear what the right order of swaps should be, and it turns out that \efx-envy can re-develop in subsequent swaps. Therefore, as one expects, the presence of two high paying chores in an agent's bundle seems to make the problem harder than if the agent had only one high paying chore.

However, we take advantage of the following crucial insight. Observe that $\p(S_i) \le 2$ for any agent $i\in N$, and $p_j > \frac{1}{2}$ for any high paying chore $j\in H$. This means that for any agent $i$, the chores in $S_i$ have cumulatively less payment than \textit{any} single high paying chore, up to a constant factor. Thus we should `balance' out the envy created among the agents due to an imbalanced allocation of the high paying chores. Note that the total number of high paying chores is at most $2n$. For $m\le 2n$, we can compute an exact EFX allocation using our \cref{alg:efx-small}. Thus we compute an \efx \textit{re-allocation} $\z'$ of the high paying chores $H$. We then add back the chores from $L$ to obtain the allocation $\x'$ given by $\x'_i = S_i \cup \z'_i$ for all $i$. This re-allocation is illustrated in \cref{fig:realloc}.

We now prove a surprising property of the allocation $\x'$: the agents who have two or more\footnote{We say `two or more' since an agent may receive more than two chores in the EFX re-allocation $\z'$.} high paying chores are actually $6$-\efx! To see why, consider an agent $i$ with $|\z'_i| \ge 2$ and any agent $h\in N$. Let $j = \arg\min_{j'\in \z'_i} d_{ij'}$. Then we have:
\begin{align*}
d_i(\x'_i) &= d_i(S_i) + d_i(\z'_i) \\
&= d_i(S_i) + d_{ij} + d_i(\z'_i\setminus \{j\}) \\
&\le d_i(S_i) + 2 \cdot d_i(\z'_i\setminus \{j\}) \tag{using $|\z'_i| \ge 2$ and choice of $j$} \\
&\le d_i(S_i) + 2 \cdot d_i(\z'_h) \tag{since $\z'$ is EFX} \\
&= \alpha_i \cdot \p(S_i) + 2 \cdot d_i(\z'_h) \tag{using the MPB condition} \\
&\le 2\alpha_i + 2\cdot d_i(\z'_h) \tag{since $\p(S_i) \le 2$} \\
&< 4\alpha_i\cdot \p(\z'_h) + 2\cdot d_i(\z'_h) \tag{since $\z'$ is \efx, $|\z'_h|\ge 1$, and $\p(\z'_h) > \frac{1}{2}$ since $\z'_h \subseteq H$} \\
&\le 6\cdot d_i(\z'_h) \tag{using the MPB condition} \\
&\le 6\cdot d_i(\x'_h).
\end{align*}

The \efx re-allocation of the $H$ chores thus leaves us to tackle the agents in with exactly one high paying chore. For these agents, we can try to use chore swaps as described earlier. Unfortunately, this does not work as it is: due to the re-allocation, we cannot use the payments $\p$ to determine the order of swaps. We show, however, that a new order of swaps can be determined that ensures that $O(1)$-\efx envy does not re-develop due to swaps. Moreover, we show that such swaps do not cause agents with two or more high paying chores to become $6$-\efx envious. Ultimately, \cref{alg:const-efx} terminates after $n$ swaps with a $6$-\efx allocation.

\paragraph{Improving the constant of approximation.} We use two ideas to tighten the approximation factor and obtain the existence of $4$-\efx allocations. First, we prove the tighter bound of $\p(S_i) \le \frac{1}{2}$ in the event that $i$ has two high paying chores. This is a consequence of a finer analysis of our ER rounding algorithm (\cref{alg:rounding}). Second, we compute a \textit{specific} EFX re-allocation $\z'$ of chores in $H$ so an agent $i$ with $\p(S_i) > 1$ has at most one chore in $\z'$. We show that our \cref{alg:efx-small} can be used to return allocations with such properties.

We provide an overview of \cref{alg:const-efx} in \cref{sec:efx-overview} and its analysis in \cref{sec:efx-analysis}. We expect that the ideas developed in obtaining this result will have wider applications. Below, we use the idea of chore swaps to obtain improved results for two structured classes: when $m\le 2n$, and for bivalued instances.

\subsubsection{Algorithms for Structured Instances}

\paragraph{\cref{alg:efx-small}: EFX allocation for $m\le 2n$.} \cref{alg:efx-small} fixes an order of the agents, say agent $1$ to agent $n$, and first allocates chores following a picking sequence. Following the sequence, each agent picks their least-disutility chore among the remaining chores in their turn. When  $m\le n$, the picking sequence is simply $1, \dots, n$. In this case, the resulting allocation is EFX since each agent gets at most one chore. When $m = n+r$ for $r\in [n]$, the picking sequence is $r, r-1, \dots, 1$, followed by $1, 2, \dots, n$.  The resulting allocation may not be EFX as the agents in $[r]$ get two chores. To fix their EFX-envy, we perform \textit{chore swaps} between the bundles of an agent $i\in [r]$ and the agent $\ell$ who $i$ most envies. In such a swap, $i$ receives the entire bundle of $\ell$, and $\ell$ receives the higher-disutility chore of $i$. We argue that after such a swap, agents $i$ and $\ell$ are both EFX. We carefully argue that each agent in $[r]$ undergoes a swap at most once, and becomes EFX after the swap. Thus, \cref{alg:efx-small} terminates with an EFX allocation after at most $r\le n$ swaps.

We note that although the existence of EFX allocations for $m\le 2n$ is known from prior work \cite{mahara2023efxmatching}, our algorithm is faster and arguably simpler as it does not repeatedly compute matchings. More importantly, the EFX allocation returned by our algorithm has certain special properties that are useful in \cref{alg:const-efx} for computing a $4$-\efx allocation in the general case.

\paragraph{\cref{alg:bivalued}: 3-EFX and PO for bivalued instances.} We next design \cref{alg:bivalued}, which returns a $3$-\efx and \po allocation for bivalued instances given an ER equilibrium with $\beta=\frac{1}{2}$; for this we assume $m > 2n$. \cref{alg:bivalued} uses the ideas of chore swaps used in \cref{alg:const-efx}, but the additional structure offered by the bivalued nature of the instance  allows us to improve the approximation guarantee to $3$-\efx while also maintaining \po.

\paragraph{\cref{alg:bivalued-small}: EFX and PO for bivalued instances with $m\le 2n$.} For bivalued instances with $m\le 2n$, we design \cref{alg:bivalued-small} which computes an \efx and \po allocation in polynomial time. \cref{alg:bivalued-small} begins with a balanced allocation computed using \cref{alg:balanced}, and then essentially runs \cref{alg:bivalued}. Since the number of chores is limited, a careful analysis shows that the guarantee of the resulting allocation can be improved to \efx and \po.

\subsection{Existence of Earning Restricted Equilibria}
Proving the existence of an ER equilibrium turns out to be quite challenging. Existing flow-based algorithms designed for computing a competitive equilibrium (CE) \cite{orlin2010fisher} or a spending-restricted equilibrium in the context of goods \cite{cole2015nswapprox,coledevanur2017SR} do not apply straightforwardly to chores. Although computing a (unrestricted) CE for chores is more difficult than for goods, there are several combinatorial algorithms \cite{chaudhury22cechores,boodaghians22cechores,chaudhury2024competitiveequilibriumchoresdual} that can compute an approximate CE. However, these algorithms also do not seem to extend to establish the existence and computation of ER equilibria for chores.

Our approach employs Lemke's complementary pivoting method on a polyhedron~\cite{Lemke65}, which is similar in spirit to the simplex algorithm for linear programming. This is a powerful approach that has been used earlier for computing CE in goods (e.g., \cite{GMSV}) and chores (e.g., \cite{ChaudhuryGMM21}). This process involves two key steps. First, we need to formulate a linear complementarity program (LCP) that captures ER equilibria. Second, we must ensure that the algorithm terminates at an ER equilibrium, which typically involves proving the absence of \emph{secondary rays} (a specific type of unbounded edges) in the LCP polyhedron; further details can be found in Section \ref{sec:lemke-prelim}.

It is important to note that both steps must work \emph{simultaneously}. Our LCP formulation captures ER equilibria, but it also captures some non-equilibrium solutions, adding complexity to our task. The most significant challenge lies in this second step. The main issue with Lemke's scheme is that it does not always guarantee termination at a solution; this occurs when the algorithm's path leads to a secondary ray.

Typically, to establish the convergence of a complementary pivot algorithm to a solution, one must prove that there are no secondary rays in the LCP polyhedron. However, our LCP formulation does contain secondary rays. This necessitates that we show that the algorithm never reaches a secondary ray, to ensure its termination. Additionally, we need to show that the final output of the algorithm is indeed an equilibrium, rather than a non-equilibrium solution to the LCP. This complicates the analysis of our algorithm.

Nevertheless, the LCP-based approach has several noteworthy features. It not only shows the (constructive) existence of an ER competitive equilibrium but also yields new structural results as simple corollaries. For example, it shows that a rational-valued equilibrium exists if all input parameters are rational, and it establishes that the problem belongs to the complexity class PPAD. Furthermore, even if computing an ER equilibrium turns out to be computationally intractable, this LCP-based method still provides a practical algorithm that performs fast in practice.

\paragraph{Polynomial time algorithm for constant $n$.} We present a polynomial time algorithm for computing an ER equilibrium when the number of agents $n$ is constant. Our algorithm effectively explores the space of all competitive allocations in $O(m^{n^2})$ time. Given that an ER competitive equilibrium is guaranteed to exist due to Theorem~\ref{thm:er-existence}, this ensures that an equilibrium will be found in polynomial time for constant $n$. Further details on this algorithm are provided in  \cref{sec:er-const}.

\subsection{Discussion and Future Directions}
In this paper, we established the existence of $4$-\efx allocations of indivisible chores, improving the previous existence result of $O(n^2)$-\efx allocations. We also proved the existence of allocations that are fair and efficient, namely (i) $2$-\eftwo and \po allocations, (ii) $(n-1)$-\efone and \po allocations, and (iii) $3$-\efx and \po allocations for bivalued instances. We introduced the framework of earning restricted (ER) competitive equilibria as a tool for obtaining informative fractional allocations with desirable fairness and efficiency properties. Our algorithms relied on rounding the ER equilibrium to a suitable integral allocation, and several techniques of splitting, swapping, and merging bundles to obtain our guarantees.

\medskip
\noindent We conclude with some concrete directions for future investigation that arise from our work.
\begin{enumerate}[leftmargin=*]
\item Perhaps the most important direction is investigating the computational complexity of computing an ER equilibrium. Our existence results (Theorems~\ref{thm:efx-main}, \ref{thm:efxpo-main}, \ref{thm:ef2-main}) are accompanied by polynomial-time algorithms that require an ER equilibrium as input. Consequently, if ER equilibria can be computed (even approximately) in polynomial time, our existence results would immediately translate into polynomial-time algorithms (with potentially a small loss in the guarantees). Indeed, we show in \cref{thm:const-er} that an ER equilibrium can be computed in polynomial time when the number of agents is constant, but the question is open in its full generality.
\item We believe that the idea of chore swaps can be used as a unifying framework to obtain both improved results and/or simpler algorithms for structured instances. For example, the existence of 2-\efx allocations for $n=3$ agents was recently shown by \cite{afshinmehr2024approximateefxexacttefx} through extensive case-analysis. We conjecture that this result can be obtained through a simpler algorithm which utilizes our ideas of ER equilibrium and chore swaps.
\item Lastly, our techniques suggest that advances on the problem of computing an \efone and \po allocation, even approximately, could drive progress for computing an $\alpha$-EFX allocation for $\alpha < 4$. In this direction, an important question is the existence of \efone and \po for $n=4$ agents.
\end{enumerate}

\section{Preliminaries}\label{sec:prelim}

\paragraph{Problem instance.} An instance $(N, M, D)$ of the chore allocation problem consists of a set $N = [n]$ of $n$ agents, a set $M = [m]$ of $m$ indivisible chores, and a list $D = \{d_i\}_{i \in N}$, where $d_i : 2^M \rightarrow \R_{\ge 0}$ is agent $i$'s \textit{disutility} function over the chores. 
Let $d_{ij} > 0$ denote the disutility of chore $j$ for agent $i$. We assume that the disutility functions are additive; thus for every $i \in N$ and $S \subseteq M$, $d_i(S) = \sum_{j \in S} d_{ij}$. An instance is said to be \textit{bivalued} if there exist $a,b\in\R_{> 0}$ such that $d_{ij} \in \{a,b\}$ for all $i\in N, j\in M$.

\paragraph{Allocation.} An \textit{integral allocation} $\x = (\x_1, \x_2, \ldots, \x_n)$ is an $n$-partition of the chores; here $\x_i \subseteq M$ is the set of chores assigned to agent $i$, who receives disutility $d_i(\x_i)$. In a \textit{fractional allocation} $\x \in [0, 1]^{n \times m}$, chores are divisible and $x_{ij} \in [0, 1]$ denotes the fraction of chore $j$ assigned to agent $i$, who receives disutility $d_i(\x_i) = \sum_{j \in M} d_{ij} \cdot x_{ij}$. We assume allocations are integral unless specified.

\paragraph{Fairness and efficiency notions.}
An allocation $\x$ is said to be:
\begin{enumerate}
\item $\alpha$-Envy-free up to $k$ chores ($\alpha$-$\textup{EF}k$) if for all $i,h\in N$, there exists $S\subseteq \x_i$ with $|S|\le k$ such that $d_i(\x_i \setminus S) \leq \alpha\cdot d_i(\x_h)$. An allocation is simply denoted by EF$k$ if it is $1$-EF$k$.
\item $\alpha$-Envy-free up to any chore ($\alpha$-\efx) if for all $i,h\in N$ and $j\in \x_i$, $d_i(\x_i \setminus \{j\}) \leq \alpha\cdot d_i(\x_h)$. An allocation is simply denoted by EFX if it is $1$-EFX.
\item Pareto optimal (\po) if there is no allocation $\y$ that dominates $\x$. An allocation $\y$ dominates allocation $\x$ if for all $i \in N$, $d_i(\y_i) \leq d_i(\x_i)$, and there exists $h \in N$ such that $d_h(\y_h) < d_h(\x_h)$. 
\item Fractionally Pareto-optimal (\fpo) if there is no fractional allocation that dominates $\x$. An \fpo allocation is clearly \po, but not vice-versa.
\end{enumerate}

\subsection{Competitive Equilibrium} An instance $(N,M,D,e)$ of a Fisher model for chores consists of a set $N$ of agents, set $M$ of chores, list $D=\{d_i\}_{i\in N}$ specifying the disutility functions of the agents, as well as an \textit{earning requirement} $e_i > 0$ for each agent $i \in N$. We associate payments $\p = (p_1, \ldots, p_m) \in \R_{> 0}^m$ with the chores, i.e. chore $j$ pays $p_j$. Each agent $i$ aims to earn at least $e_i$ by performing chores in exchange for payment. In a (fractional) allocation $\x$ with payments $\p$, the \textit{earning} of agent $i$ is $\p(\x_i) = \sum_{j \in M} p_j \cdot x_{ij}$. 

An allocation $(\x,\p)$ is said to be a competitive equilibrium if all chores are allocated and all agents earn their earning requirement subject to performing chores of least possible disutility. For additive disutilities, the latter condition can be expressed in terms of disutility-to-payment ratios as follows.

\begin{definition}\label{def:mpb}(MPB allocation) \normalfont For each agent $i$, the \textit{pain-per-buck} ratio $\alpha_{ij}$ of chore $j$ is defined as $\alpha_{ij} = d_{ij} / p_j$, and the \textit{minimum-pain-per-buck} (MPB) ratio of agent $i$ is then given by $\alpha_i = \min_{j \in M} \alpha_{ij}$. Let $\mpb_i = \{j \in M \mid d_{ij} / p_j = \alpha_i\}$ denote the set of chores which are MPB for agent $i$ for payments $\p$. An allocation $(\x,\p)$ is called an \textit{MPB allocation} if for all $i\in N$ and $j\in M$, $x_{ij} > 0$ implies $j\in\mpb_i$, i.e., agents are only assigned chores which are MPB for them.
\end{definition}

With the above definition, competitive equilibria for agents with additive disutilities can be defined as follows.

\begin{definition}\label{def:ce}(Competitive equilibrium) \normalfont
We say that $(\x, \p)$ is a \textit{competitive equilibrium} (CE) for the instance $(N,M,D,e)$ if (i) for all $j \in M$, $\sum_{i \in N} x_{ij} = 1$, i.e., all chores are completely allocated, (ii) for all $i \in N$, $\p(\x_i) = e_i$, i.e., each agent receives her earning requirement, and (iii) $(\x,\p)$ is an MPB allocation.
\end{definition}

The First Welfare Theorem \cite{mas1995microeconomic} shows that for a competitive equilibrium $(\x, \p)$ of instance $(N,M,D,e)$, the allocation $\x$ is \fpo. Using this fact, we can argue:

\begin{restatable}{proposition}{thmfpo}\label{thm:fpo}
Let $(\x,\p)$ be an MPB allocation. Then $\x$ is \fpo.
\end{restatable}
\begin{proof}
We create an associated Fisher market instance $I = (N,M,D,e)$ by defining $e_i = \p(\x_i)$ for each $i\in N$. It is easy to see that $(\x,\p)$ is a CE for $I$. By the First Welfare Theorem, $\x$ is \fpo.
\end{proof}

The above proposition shows that MPB allocations are useful in ensuring efficiency. We now discuss how such allocations can be utilized for fairness. For an MPB allocation $(\x, \p)$ where $\x$ is integral, we let $\p_{-k}(\x_i) := \min_{S \subseteq \x_i, |S| \leq k} \p(\x_i \setminus S)$ denote the payment agent $i$ receives from $\x_i$ excluding her $k$ highest paying chores. Likewise, we let $\p_{-X}(\x_i) := \max_{j\in \x_i} \p(\x_i\setminus \{j\})$ denote the payment $i$ receives from $\x_i$ excluding her lowest paying chore. 

\begin{definition}[Payment \s{EF}$k$ and Payment \efx]\label{def:pef1}
\normalfont An allocation $(\x, \p)$ is said to be $\alpha$-\textit{payment envy-free up to $k$ chores} ($\alpha$-\pefk) if for all $i, h \in N$ we have $\p_{-k}(\x_i) \leq \alpha\cdot\p(\x_h)$. Agent $i$ \textit{$\alpha$-\pefk-envies} $h$ if $\p_{-k}(\x_i) > \alpha\cdot\p(\x_h)$. 

An allocation $(\x, \p)$ is said to be $\alpha$-\textit{payment envy-free up any chore} ($\alpha$-\pefx) if for all $i, h \in N$ we have $\p_{-X}(\x_i) \leq \alpha\cdot\p(\x_h)$. Agent $i$ \textit{$\alpha$-\pefx-envies} $h$ if $\p_{-X}(\x_i) > \alpha\cdot\p(\x_h)$. 
\end{definition}

We derive a sufficient condition for computing an $\alpha$-\s{EF}$k$/$\alpha$-\efx and \po allocation.
\begin{lemma}\label{lem:pEF1impliesEF1}
Let $(\x,\p)$ be an MPB allocation where $\x$ is integral. 
\begin{itemize}
\item[(i)] If $(\x,\p)$ is $\alpha$-\pefk, then $\x$ is $\alpha$-\textup{EF}$k$ and \fpo.
\item[(ii)] If $(\x,\p)$ is $\alpha$-\pefx, then $\x$ is $\alpha$-\efx and \fpo.
\end{itemize}
\end{lemma}
\begin{proof}
Since $(\x,\p)$ is an MPB allocation,  \cref{thm:fpo} shows $\x$ is \fpo. Let $\alpha_i$ be the MPB ratio of agent $i$ in $(\x,\p)$. Consider any pair of agents $i,h\in N$. 
\begin{itemize}
\item[(i)] If $(\x,\p)$ is $\alpha$-\pefk, then:
\begin{align*}
\min_{S\subseteq\x_i, |S|\le k} d_i(\x_i\setminus S) &= \alpha_i\cdot\p_{-k}(\x_i) \tag{since $(\x,\p)$ is on MPB} \\
&\le \alpha_i\cdot \alpha \cdot \p(\x_h) \tag{since $(\x,\p)$ is $\alpha$-pEF$k$ (\cref{def:pef1})} \\
&\le \alpha \cdot d_i(\x_h). \tag{since $(\x,\p)$ is on MPB}
\end{align*}
Thus, $\x$ is $\alpha$-\s{EF}$k$. 
\item[(ii)] If $(\x,\p)$ is $\alpha$-\pefx, then:
\begin{align*}
\max_{j\in\x_i} d_i(\x_i\setminus \{j\}) &= \alpha_i\cdot\p_{-X}(\x_i) \tag{since $(\x,\p)$ is on MPB} \\
&\le \alpha_i\cdot \alpha \cdot \p(\x_h) \tag{since $(\x,\p)$ is $\alpha$-\pefx (\cref{def:pef1})} \\
&\le \alpha \cdot d_i(\x_h). \tag{since $(\x,\p)$ is on MPB}
\end{align*}
Thus, $\x$ is $\alpha$-\efx. \qedhere
\end{itemize}
\end{proof}

\subsection{Earning-Restricted Equilibrium} 
We introduce the concept of \textit{earning-restricted (ER) competitive equilibrium} for chores. An instance $(N,M,D,e,c)$ of the ER competitive equilibrium problem consists of a set $N = [n]$ of $n$ agents, a set $M = [m]$ of $m$ chores, a list $D = \{d_i\}_{i\in N}$ of additive agent disutility functions, a list $e = \{e_i\}_{i\in N}$ of agent earning requirements, and a list $c = \{c_j\}_{j\in M}$ of chore earning-restrictions. As before, each agent $i\in N$ aims to earn at least $e_i > 0$ by performing chores in exchange for payment from the chores. However, the money that agents can collectively earn from a chore $j\in M$ is capped, and this cap is specified by the earning limit $c_j\ge 0$.

Thus, an equilibrium $(\x,\p)$ of an ER instance consists of a partial fractional allocation $\x$ and a set of chore payments $\p$ such that each agent $i$ earns her earning requirement $e_i$ while performing chores of least possible disutility, with the restriction that the earning from each chore $j$ is at most $c_j$. Once a chore has paid $c_j$ to the agents, the rest of the chore is not assigned. Define the \textit{earning vector} $\q\in\R^{n\times m}$ associated with $(\x,\p)$ given by $q_{ij} := p_j x_{ij}$ which denotes the amount agent $i$ earns from chore $j$. Let $q_j = \sum_i q_{ij} = \sum_i p_j x_{ij}$ denote the total earning from chore $j$. We now formally define an ER equilibrium $(\x,\p)$.

\begin{definition}[Earning-restricted equilibrium]\label{def:er-main}
Let $\q$ be the earning vector associated with an allocation $(\x,\p)$. Then $(\x,\p)$ is an earning-restricted equilibrium of an ER instance $(N,M,D,e,c)$ if 
\begin{itemize}
\item[(i)](Agents) $(\x,\p)$ is an MPB allocation, i.e., for $i\in N, j\in M$, $x_{ij} > 0$ implies $j\in \mpb_i$. Moreover, for each $i\in N$, $\sum_j q_{ij} = e_i$.
\item[(ii)] (Chores) For each $j\in M$, either $\sum_i x_{ij} = 1$ and $q_j = p_j \le c_j$, or $\sum_i x_{ij} < 1$ and $q_j = c_j < p_j$. In other words, for each $j$, $q_j = \min\{p_j, c_j\}$.
\end{itemize}
\end{definition}

In the above definition, the first condition expresses that agents are assigned MPB chores and each agent earns their earning requirement. The second condition expresses that for each chore $j$, if the payment $p_j$ of the chore is at most the earning limit $c_j$, then the chore is fully assigned and pays out $q_j = p_j$ to the agents. On the other hand, if the payment $p_j$ exceeds the earning limit $c_j$, then the chore will pay out $q_j = c_j$ to the agents, and only a $c_j/p_j$ fraction of the chore will be assigned. For notational convenience, we often use both $(\x,\p)$ and $(\x,\p,\q)$ to denote an ER equilibrium. 

Clearly, an ER equilibrium can exist only if $\sum_i e_i \le \sum_j c_j$, i.e., the chores must collectively pay enough so all agents can earn their earning requirements. In \cref{sec:er}, we prove that this condition is in fact sufficient for existence.

\thmerexistence*

\section{Existence of 2-EF2 and PO Allocations}\label{sec:rounding}
In this section, we prove the existence of $2$-\eftwo and \fpo allocations for all chore allocation instances.

\thmEFTwoMain*

We prove \cref{thm:ef2-main} through two algorithms: \cref{alg:rounding} which returns a $2$EF$2$ and fPO allocation for instances with $m \ge 2n$, and \cref{alg:balanced} which returns a EF$2$ and fPO allocation for instances with $m\le 2n$.

\paragraph{\cref{alg:rounding}: 2-\eftwo and \po for $m\ge 2n$.} The main idea is to use ER equilibria to compute a fair and efficient allocation. Given a chore allocation instance, we uniformly set agent earning requirements $e_i = 1$ and impose a uniform earning limit of $\beta\in[\frac{1}{2},1)$ on all chores. Since $m\ge 2n$ and $\beta\ge \frac{1}{2}$, we have $m\cdot\beta\ge n$. Thus, the feasible earning condition is satisfied and an ER equilibrium $(\z,\p)$ exists by \cref{thm:er-existence}. We design a polynomial time algorithm \cref{alg:rounding} which carefully rounds the fractional ER equilibrium allocation $\z$ to an integral allocation $\x$ that is approximately-EF$k$ and \fpo. With different choices of $\beta$, the rounded integral allocation satisfies different fairness guarantees. In particular, setting $\beta=\frac{1}{2}$ gives a $2$-\eftwo and \fpo allocation. We present \cref{alg:rounding} and its analysis in \cref{sec:errounding}.

\paragraph{\cref{alg:balanced}: \eftwo and \po for $m\le 2n$.}
To handle the case of $m \le 2n$, we design a polynomial time algorithm (\cref{alg:balanced}), which computes an \eftwo and \fpo allocation. Specifically, for any number of chores, \cref{alg:balanced} produces an \fpo allocation in which the number of chores in agent bundles differ by at most one, i.e., is balanced. Thus, for $m\le 2n$, each agent gets at most two chores and hence the allocation is \eftwo. The algorithm starts with an imbalanced allocation and transfers chores from agents with a higher number of chores to agents with a lower number chores until the allocation is balanced, while preserving \fpo. We present \cref{alg:balanced} and its analysis in \cref{sec:balanced}.

\begin{algorithm}[!t]
\caption{Earning Restricted Rounding}\label{alg:rounding}
\textbf{Input:} Instance $(N,M,D)$ with $m\beta\ge n$; ER equilibrium $(\y,\p)$ with earning limit $\beta\in [\frac{1}{2}, 1)$ \\
\textbf{Output:} An integral allocation $\x$
\begin{algorithmic}[1]
\State $(\z,\p) \gets \texttt{MakeAcyclic}(\y,\p)$
\State Let $G = (N, M, E)$ be the payment graph associated with $(\z,\p)$
\State Root each tree of $G$ at some agent and orient edges
\State $\x_i \gets \emptyset$ for all $i\in N$ \Comment{Initialize empty allocation}
\State $L = \{j \in M: p_j \le \beta\}$, $H = \{j \in M: p_j > \beta\}$ \Comment{Low, High paying chores}
\Statex \textit{--- Phase 1: Round leaf chores ---}
\For{all leaf chores $j$} 
\State $\x_i \gets \x_i \cup \{j\}$ for $i= \parent{j}$; delete $j$ from $G$
\EndFor
\Statex \textit{--- Phase 2: Allocate $L$ ---}
\For{every tree $T$ of $G$}
\For{every agent $i$ of $T$ in BFS order}
\If{$\p(\x_i) > 1$}
\For{every $j\in\child{i}\cap H$}
\State Assign $j$ to agent $h\in\child{j}$ earning most from $j$ among $\child{j}$; delete $j$
\EndFor
\EndIf
\While{$\exists j\in\child{i}\cap L$ s.t. $\p(\x_i\cup \{j\}) \le 1$}
\State $\x_i \gets \x_i \cup \{j\}$; delete $j$ from $G$
\EndWhile
\For{every $j\in\child{i}\cap L$}
\State Assign $j$ to arbitrary agent $h\in\child{j}$; delete $j$ from $G$
\EndFor
\EndFor
\EndFor
\Statex \textit{--- Phase 3: Pruning trees ---}
\For{chore $j\in V(G) \cap M$}
\If{agent $i\in \child{j}$ does not earn the most from $j$ among agents in $\child{j}$}
\State Delete edge $(j, i)$ from $G$
\EndIf
\EndFor
\Statex \textit{--- Phase 4: Matching to allocate $H$ ---}
\For{every tree $T = (N(T) \cup M(T), E(T))$ of $G$}
\State $h\gets \arg\max_{i\in N(T)} \p(\x_i)$
\State Compute a matching $\sigma$ of $i\in N(T)\setminus\{h\}$ to $M(T)$
\For{$i\in N(T)\setminus \{h\}$}
\State $\x_i \gets \x_i \cup \{\sigma(i)\}$
\EndFor
\EndFor
\State \Return $\x$
\end{algorithmic}
\end{algorithm}

\subsection{Earning-Restricted Rounding}\label{sec:errounding}
We now describe \cref{alg:rounding}, which rounds a fractional ER equilibrium $(\y,\p)$ of an instance with uniform chore earning limit $\beta$ to an approximately-\s{EF}$k$ and \fpo allocation. Our algorithm modifies the allocation by manipulating its payment graph, as defined below.

\begin{definition}[Payment graph] The payment graph $G = (N,M,E)$ associated with an allocation $(\x,\p)$ is a weighted bipartite graph with vertex set $V(G) = N\sqcup M$, and edge set $E(G) = \{(i,j) : i\in N, j\in M, x_{ij}>0\}$. The weight of edge $(i,j)$ is $p_j\cdot x_{ij}$, which is the earning of agent $i$ from chore $j$.
\end{definition}

\cref{alg:rounding} first transforms the given equilibrium into one whose payment graph is acyclic, i.e., is a collection of trees. This is due to the following lemma:

\begin{restatable}{lemma}{lemacyclic}\label{lem:acyclic}
There is a polynomial time algorithm \emph{$\texttt{MakeAcyclic}$} which takes as input an ER equilibrium of instance $I$ and returns an another ER equilibrium of $I$ whose payment graph is acyclic.
\end{restatable}
\begin{proof}
The algorithm \texttt{MakeAcyclic} begins with the payment graph $G = (N\sqcup M, E)$ of the ER equilibrium $(\y,\p)$. If $G$ is acyclic, it returns $(\y,\p)$. Otherwise, suppose an intermediate allocation $(\z,\p)$ has a cycle $C = (i_1, j_1, i_2, j_2, \dots, i_k, j_k, i_1)$, where $i_\ell \in N$ are agents and $j_\ell\in M$ are chores, and $C$ contains the edges $(i_\ell, j_\ell)\in E$ and $(j_\ell,i_{\ell+1})\in E$ for $1\le \ell \le k$, with the notation that $i_{k+1} = i_1$. The earning of an agent $i$ from chore $j$ is $q_{ij} = p_j\cdot y_{ij}$. Without loss of generality, assume $(i_1,j_1)$ is the edge with minimum $q_{ij}$ among the edges $(i,j)$ in $C$. Let $s = q_{i_1j_1}$. 

Now consider the allocation $(\z,\p,\q')$, where for all $\ell\in[k]$, $q'_{i_\ell j_\ell} = q_{i_\ell j_\ell} - s$, and $q'_{i_\ell j_{\ell-1}} = q_{i_\ell j_{\ell-1}} + s$, and $q'_{ij} = q_{ij}$ for all $(i,j)\notin C$. This has the effect of circulating agent earnings around the cycle $C$ and the edge $(i_1,j_1)$ is no longer present in the payment graph of $(\z,\p)$. \texttt{MakeAcyclic} updates the allocation to $(\z,\p)$ and continues deleting cycles until the payment graph becomes acyclic. Since each step strictly decreases the number of edges in the payment graph and cycles can be found efficiently, \texttt{MakeAcyclic} terminates in polynomial time. 

We prove using induction that the resulting allocation is an ER equilibrium. The initial allocation $(\y,\p)$ is an ER equilibrium and suppose the claim holds at some iteration with an updated allocation $(\y,\p)$. Let $(\z,\p,\q')$ be the next allocation. Notice for each agent $i$, $\sum_j q'_{ij} = \sum_j q_{ij} = e_i$. Next for each chore, we have $\sum_i q'_{ij} = \sum_i q_{ij} = \min\{p_j, c_j\}$. Lastly if $z_{ij} > 0$ then $y_{ij} > 0$ as well. Thus the conditions of \cref{def:er-main} is satisfied, implying that $(\z,\p,\q')$ is an ER equilibrium. 
\end{proof}

Given the ER equilibrium $(\z,\p)$ whose payment graph $G$ is acyclic, \cref{alg:rounding} roots each tree of $G$ at some agent and orients its edges. For a node $v\in V(G)$, let $\child{v}$ denote the children nodes of $v$ and $\parent{v}$ denote the parent node of $v$. Note that the root nodes of trees in $G$ are agents and the leaf nodes are chores. We let $\x$ denote the integral allocation of chores to agents made by \cref{alg:rounding}, which is initially empty. We classify chores into two sets: $L = \{j\in M: p_j\le \beta\}$ comprising of low paying chores, and $H = \{j\in M: p_j > \beta\}$ comprising of high paying chores. \cref{alg:rounding} proceeds in four phases.

Phase 1 rounds every leaf chore $j$ to their parent agent $\parent{j}$ and then deletes $j$ from $G$. After this, all chores in $G$ have edges to at least two agents, i.e., are \textit{shared} chores. Note that there can be at most $(n-1)$ shared chores, since $G$ is acyclic.

Phase 2 assigns chores in $L$. In each tree $T$ of $G$, we visit agents in breadth-first order starting from the root. At agent $i$, we first check if $\p(\x_i) > 1$. Note that this can happen only if $i$ received the $\parent{i}$ chore. If so, we assign every chore $j\in \child{i}$ to a child agent of $j$. A chore $j\in \child{i}\cap L$ is assigned to an arbitrary child of $j$, while $j\in \child{i}\cap H$ is assigned to an agent who earns the most from $j$ among children of $j$. After this, such an agent $i$ is not assigned any further chores in the algorithm. Otherwise, if $\p(\x_i) \le 1$ when visiting $i$, we iteratively assign the child chores of $i$ in $L$ as long as $\p(\x_i)\le 1$. Any remaining child chore $j\in L$ is assigned to an arbitrary child agent of $j$. Thus at the end of phase 2, all chores in $L$ have been allocated, and the graph $G$ is a collection of `Phase 2 trees' whose vertices are agents and chores from $H$.

Phase 3 prunes Phase 2 trees by deleting certain edges. For every shared chore $j\in H$, we delete the edge $(j,i)$ for $i\in\child{j}$ if $i$ does not earn the most from $j$ among the child agents of $j$. As a result, we obtain `Phase 3' trees in which each chore $j\in H$ is adjacent to exactly two agents. 

Phase 4 assigns the remaining shared chores in $H$. Due to the pruning phase, each Phase 3 tree $T = (N(T)\cup M(T), E(T))$ with $|N(T)| = r$ agents has exactly $|M(T)| = r-1$ shared chores from $H$. We identify an agent $h\in N(T)$ with the highest earning $\p(\x_h)$, and then assign the $(r-1)$ chores of $M(T)$ to the $(r-1)$ agents of $N(T)\setminus\{h\}$ via a matching. Such a matching is possible because in a Phase 3 tree, each shared chore is adjacent to exactly two agents. Thus during Phase 4, every agent gets at most one chore from $H$, and all chores are allocated.

\begin{restatable}{lemma}{lemRoundRuntime}\label{lem:rounding-runtime}
Given an ER equilibrium $(\y,\p)$ of an instance $(N,M,D)$, \cref{alg:rounding} returns an integral allocation $\x$ in $\poly{n,m}$ time.
\end{restatable}
\begin{proof}
\cref{lem:acyclic} shows that the procedure \texttt{MakeAcyclic} results in an allocation with an acyclic payment graph in polynomial time. In the following phases, \cref{alg:rounding} assigns all chores to agents. Each phase takes polynomial time since they involve polynomial time operations such as BFS in the payment graph or computing a matching in a tree.
\end{proof}

We now analyze the properties of the allocation $\x$ returned by \cref{alg:rounding}. We first show:
\begin{restatable}{lemma}{lemRoundingfPO}\label{lem:rounding-fpo}
The allocation $\x$ returned by \cref{alg:rounding} is \fpo.
\end{restatable}
\begin{proof}
Since $(\z,\p)$ is an ER equilibrium, $(\z,\p)$ is an MPB allocation. Let $Z_i = \{j : z_{ij} > 0\}$. Note that throughout \cref{alg:rounding}, $\x_i \subseteq Z_i$. Hence, $(\x,\p)$ is also an MPB allocation. Consequently, \cref{thm:fpo} implies that $\x$ is \fpo.
\end{proof}

To analyze fairness properties of $\x$, we first prove upper bounds on agent earnings. Essentially, the following lemma states that the earning up to one chore of each agent is at most 1, except when the agent has two chores from $H$; in the latter case the agent earns at most $1-\beta$ from other chores. 
\begin{restatable}{lemma}{lemEarningUB}\label{lem:earning-ub}
Let $(\x,\p)$ be the allocation returned by \cref{alg:rounding} with earning restriction $\beta \in [\frac{1}{2},1)$. Then for each $i\in N$, either $\p_{-1}(\x_i) \le 1$, or $|\x_i\cap H| = 2$ and $\p_{-2}(\x_i) \le 1-\beta$.
\end{restatable}
\begin{proof}
Let $\x^t$ denote the allocation after Phase $t$, for $t\in[4]$; note that $\x^4 = \x$. Consider an agent $i\in N$. Let $\hat{\x}_i$ be the allocation when \cref{alg:rounding} visits $i$ in Phase 2. Suppose $\p(\hat{\x}_i) \le 1$. Then we have $\p(\x^2_i) \le 1$ at the end of Phase 2 after $i$ is assigned a subset of $\child{i}\cap L$. Subsequently, $i$ could be assigned one more chore in Phase 4. Hence we have $\p_{-1}(\x_i) \le 1$ in this case. 

On the contrary, suppose $\p(\hat{\x}_i) > 1$. Then \cref{alg:rounding} will not allocate any chore to $i$ in Phase 4, and hence $\x_i = \x^2_i = \hat{\x}_i$. Note that either $\hat{\x}_i = \x^1_i$ or $\hat{\x}_i = \x^1_i \cup \{j\}$, where $j = \parent{i}$. That is, $\hat{\x}_i$ includes the chores $\x^1_i$ allocated to $i$ in Phase 1, and may include $i$'s parent chore $j$. Recall that Phase 1 rounds leaf chores to their parent agents, hence $\x^1_i$ comprises of the leaf chores that are also child chores of $i$. Due to the earning restriction of $\beta$, agent $i$ earns exactly $\beta$ from any chore in $\x^1_i\cap H$. Since $\beta \ge \frac{1}{2}$ and $e_i = 1$, we have $|\x^1_i\cap H|\le 2$. We consider three scenarios:

\paragraph{Case 1: $|\x^1_i \cap H| = 0$.} In this case, we have $\p(\x^1_i)\le 1$. Hence $\p_{-1}(\hat{\x}_i) \le \p(\hat{\x}_i\setminus\{j\}) \le  \p(\x^1_i) \le 1$.

\paragraph{Case 2: $|\x^1_i \cap H| = 1$.} Let $\x^1_i\cap H = \{j_1\}$. Then $\p(\x^1_i\setminus\{j_1\}) \le 1-\beta$, since the earning of $i$ from the $j_1$ is exactly $\beta$. We have three possibilities depending on the payment of $i$'s parent chore $j$.
\begin{itemize}
\item If $j\notin H$, then $p_j\le \beta$. Then observe that:
\begin{align*}
\p_{-1}(\hat{\x}_i) &= \p(\hat{\x}_i\setminus\{j_1\}) \\
&\le \p((\x^1_i\setminus\{j_1\}) \cup \{j\}) \\
&\le (1-\beta) + \beta = 1.
\end{align*}
\item If $j\in H$ and $j\notin\hat{\x}_i$, then $\p_{-1}(\hat{\x}_i) = \p(\x^1_i\setminus\{j_1\}) \le 1-\beta < 1$.
\item If $j\in H$ and $j\in\hat{\x}_i$, then $|\x_i \cap H| = 2$, and $\p_{-2}(\hat{\x}_i\setminus\{j,j_1\}) \le \p(\x^1_i\setminus\{j_1\}) \le 1-\beta$.  
\end{itemize}

\paragraph{Case 3: $|\x^1_i \cap H| = 2$.} In this case, $\p(\x_i^1) \ge 2\beta$, since $i$ earns exactly $\beta$ from each chore in $\x_i^1\cap H$. However since $\beta\ge \frac{1}{2}$ and $e_i=1$, this case can only arise if $\beta = \frac{1}{2}$, in which case $i$ can only be earning from the two chores in $\x^1_i$. Thus $i$ has no parent chore, and $\x_i = \x_i^1$. Hence, $|\x_i \cap H| = 2$ and $\p_{-2}(\x_i) = 0$.

This proves the lemma. 
\end{proof}

We next establish lower bounds on agent earnings. These bounds are derived by investigating the allocation computed by the matching phase (Phase 4). We say agent $i$ \textit{loses} a chore $j$ if $i$ is earning from $j$ in the fractional solution $\z$ but not in the integral allocation $\x$, i.e., $z_{ij} > 0$ but $j\notin \x_i$.

\begin{restatable}{lemma}{lemEarningLB}\label{lem:earning-lb}
Let $(\x,\p)$ be the allocation returned by \cref{alg:rounding}. Then for each agent $i\in N$, $\p(\x_i) \ge \min\{\beta, 1-\beta\}$.
\end{restatable}
\begin{proof}
Let $(\z,\p)$ be the ER equilibrium whose payment graph is acyclic, which is computed before Phase 1 begins. Let $\x^t$ denote the allocation after Phase $t$ of \cref{alg:rounding}, for $t\in[4]$. Note that $\x^2 = \x^3$ since Phase 3 does not assign any chores and only deletes edges in $G$. Also note $\x^4 = \x$.

Let $T = (N(T) \cup M(T), E(T))$ be a Phase 3 tree rooted at agent $i_0$. Since $T$ is a Phase 3 tree, $T$ has exactly $|N(T)|-1$ chores, all of which belong to $H$. Phase 4 identifies the agent $h\in\arg\max_{i\in N(T)} \p(\x^3_i)$, and assigns a chore $\sigma(i) \in H$ to every agent $i\in N(T)\setminus\{h\}$ by computing a matching of $M(T)$ to $N(T)\setminus \{h\}$. Since $p_j > \beta$ for $j\in H$, we have $\p(\x_i)\ge p_{\sigma(i)} > \beta$ for all $i\in N(T)\setminus \{h\}$. Hence we only need to prove lower bounds on the earning $\p(\x_h)$ of the agent $h$. Note that $\x_h = \x^3_h = \x^2_h$, since $h$ is not allocated any chores in Phase 3 or 4. By choice of $h$, we also have that $\p(\x_h) \ge \p(\x^3_i) = \p(\x^2_i)$ for all $i\in N(T)$. We now analyze three scenarios.

\begin{itemize}[]
\item[(i)] Some agent $i\in N(T)$ lost a child chore $j\in \child{i}$. Suppose $i$ lost $j$ in Phase 2. If $j\in H$, then it must be that $\p(\x^2_i) > 1$. If $j\in L$, then it must be that $\p(\x^2_i) \ge 1-\beta$; otherwise we could have assigned $j$ to $i$ in Phase 2. In either case, we have $\p(\x^2_i) \ge 1-\beta$, and hence $\p(\x_h) \ge \p(\x^2_i) \ge 1-\beta$ by choice of $h$. Note that $i$ cannot lose $j\in \child{i}$ in Phase 3 since Phase 3 only deletes edges from a chore to some of its child agents. Thus, $\p(\x_h) \ge 1-\beta$ in this case.

\item[(ii)] No agent in $N(T)$ lost a child chore. In this case, no agent in $N(T)\setminus \{i_0\}$ has lost any chore they were earning from in $(\z,\p)$; the root agent $i_0$ could have potentially lost its parent chore $j_0 = \parent{i_0}$. We evaluate the amount of earning $i_0$ loses due to losing $j_0$. Suppose $j_0 \in H$. Then $i_0$ must have lost $j_0$ in either Phase 2 or 3 to some agent $i'\in\child{j_0}$ since $i_0$ was not earning the most from $j_0$ among agents in $\child{i_0}$. Due to the earning limit, agents can earn at most $\beta$ from $j_0$. Hence the earning from $i_0$ from $j_0$ is at most $\frac{\beta}{2}$. On the other hand, if $j_0\in L$, then $i_0$ earns at most $p_{j_0} \le \beta$ from $j_0$. In either case, we find that $i_0$ has only lost $\beta$ in earning. Hence the total earning of agents in $N(T)$ is at least $|N(T)|-\beta$, while that from the chores in $M(T)$ is at most $\beta\cdot(|N(T)|-1)$. Hence there is at least one agent $i\in N(T)$ whose earning $\p(\x^2_i)$ satisfies:
\[ \p(\x^2_i) \ge \frac{|N(T)| - \beta - \beta\cdot(|N(T)|-1)}{|N(T)|} = 1-\beta.\]
Since $\p(\x_h)\ge \p(\x^2_i)$ by choice of $h$, this implies $\p(\x_h) \ge 1-\beta$.\qedhere
\end{itemize}

To conclude, we established that $\p(\x_i)\ge \beta$ for the agents $i$ that are matched to a chore in Phase 4, or $\p(\x_h) \ge 1-\beta$ for the agents $h$ that are not matched. Thus, for all $i\in N$, $\p(\x_i)\ge \min\{\beta,1-\beta\}$.
\end{proof}

Note that the maximum value of the lower bound on agent earnings is given by \cref{lem:earning-lb} is obtained at $\beta = \frac{1}{2}$. We now prove the main theorem of this section.

\begin{restatable}{theorem}{thm3EFTwo}\label{thm:3EF2}
Given an ER equilibrium of an instance $(N,M,D)$ where $m \ge 2n$, \cref{alg:rounding} returns a $2$-\eftwo and \fpo allocation in polynomial time. More precisely, for every agent the allocation is either $2$-\efone or \eftwo.
\end{restatable}
\begin{proof}
Let $(\x,\p)$ be the allocation returned by running \cref{alg:rounding} with an ER equilibrium of the instance with $\beta = \frac{1}{2}$. \cref{lem:earning-lb} then implies that for every agent $i\in N$, $\p(\x_i) \ge \frac{1}{2}$.

\cref{lem:earning-ub} implies that for every agent $h\in N$, either $\p_{-1}(\x_h) \le 1$, or $\p_{-2}(\x_h) \le 1 - \beta = \frac{1}{2}$. Thus for any agent $h$:
\begin{itemize}
\item If $\p_{-1}(\x_h) \le 1$, then $\x$ is $2$-\efone for agent $h$, as $\p_{-1}(\x_h) \le 1 \le 2\cdot \p(\x_i)$.
\item If $\p_{-2}(\x_h) \le \frac{1}{2}$, then $\x$ is \eftwo for agent $h$, as $\p_{-2}(\x_h) \le \frac{1}{2} \le \p(\x_i)$.
\end{itemize}
Using \cref{lem:pEF1impliesEF1}, this shows that every agent $h$ is either $2$-\efone or \eftwo towards any another agent $i$. Overall, the allocation is $2$-\eftwo. \cref{lem:rounding-fpo} implies $\x$ is \fpo and \cref{lem:rounding-runtime} shows \cref{alg:rounding} runs in polynomial time. 
\end{proof}

We remark that our rounding algorithm is tight: \cref{ex:ef2-tight} shows an instance where no rounding of an ER equilibrium with $\beta=\frac{1}{2}$ can return a $(2-\delta)$-\eftwo and \po allocation, for any $\delta > 0$.

\subsection{Algorithm for Balanced Chore Allocation}\label{sec:balanced}
In this section, we show the existence of an \eftwo and \fpo allocation when $m\le 2n$. For this, we design \cref{alg:balanced}, which computes a \textit{balanced} \fpo allocation for any given instance. Formally, an allocation $\x$ is balanced iff $||\x_i| - |\x_h|| \le 1$ for all agents $i,h\in N$, i.e., the sizes of agent bundles differ by at most one. In particular, when $m \leq 2n$, every agent in a balanced allocation has at most two chores, so the allocation returned by \cref{alg:balanced} is \eftwo and \fpo. Similarly, when $m \leq n$, \cref{alg:balanced} returns an \efone and \fpo allocation since every agent has at most one chore.

Our algorithm relies on a useful structure associated with an MPB allocation called the MPB graph.
\begin{definition}[MPB graph]
The augmented MPB graph $G = (N,M,E)$ associated with an integral MPB allocation $(\x,\p)$ is a directed bipartite graph with vertex set $V(G) = N\sqcup M$, and edge set $E(G) = \{(i,j) : i\in N, j\in M, j\in\x_i\} \cup \{(j,i) : i\in N, j\in M, j\in \mpb_i\setminus \x_i\}$. 
\end{definition}

\begin{algorithm}[t]
\caption{Balanced \po allocation}\label{alg:balanced}
\textbf{Input:} Chore allocation instance $(N,M,D)$ \\
\textbf{Output:} An balanced \po allocation $\x$
\begin{algorithmic}[1]
\State For some agent $h\in N$, set $\x_h \gets M$, and $\x_i \gets \emptyset$ for $i \neq h$
\State For each $j \in M$, set $p_j \gets d_{hj}$
\State $\ell \gets \arg\min_{i \in N} \lvert \x_i \rvert$ \Comment{Agent with fewest number of chores}
\While{$\lvert \x_h \rvert > \lvert \x_\ell \rvert + 1$}
    \State $C_h \gets$ Vertices reachable from $h$ in the MPB graph of $(\x,\p)$ 
    \If{$\ell \in C_h$} \Comment{Transfer chores along a path in the MPB graph}
        \State $P \gets (h = i_0, j_1, i_1, j_2, \dots, i_{k-1}, j_k, i_k = \ell)$ \Comment{Path from $h$ to $\ell$}
        \For{$1 \leq r \leq k$}
            \State $\x_{i_{r-1}} \gets \x_{i_{r-1}} \setminus \{j_r\}$, $\x_{i_r} \gets \x_{i_r} \cup \{j_r\}$ \Comment{Chore transfers along path $P$}
        \EndFor
        \State $\ell \gets \arg\min_{i \in N} \lvert \x_i \rvert$
    \Else \Comment{Raise payments of chores in $C_h$}
        \State $\gamma \gets \min_{i \in N \setminus C_h, j \in C_h \cap M} \frac{\alpha_i}{d_{ij} / p_j}$
        \For{$j \in C_h$}
            \State $p_j \gets \gamma \cdot p_j$
        \EndFor
    \EndIf
\EndWhile
\State \Return $\x$
\end{algorithmic}
\end{algorithm}

Algorithm~\ref{alg:balanced} first allocates the entire set of chores $M$ to an arbitrary agent $h$, with the payment of each chore $j$ set as $p_j=d_{hj}$ to ensure the initial allocation is MPB. We make progress towards achieving a balanced allocation by reducing the number of chores assigned to $h$, the agent with the most chores, and increasing the number of chores assigned to $\ell$, the agent with the fewest chores. Let $C_h$ denote the set of vertices in the MPB graph reachable from $h$. If $\ell\in C_h$, i.e., $\ell$ is reachable from $h$ in the MPB graph via a path $P = (h = i_0, j_1, i_1, j_2, \dots, j_{k-1}, i_{k-1}, j_k, i_k=\ell)$, then chore $j_r$ is transferred from $i_{r-1}$ to $i_r$, for all $r\in[k]$. The result of such transfers is that $h$ has one fewer chore, $\ell$ has one more chore, and all other agents maintain the same number of chores. Observe that since $j_r\in \mpb_{i_r}$ for all $r\in[k]$, the allocation after the transfers is also MPB. If $\ell\notin C_h$, then we uniformly raise the payments of all chores in $C_h$, until a chore $j\in C_h$ becomes MPB for an agent $i\notin C_h$. Such a payment raise respects the MPB condition for all agents, hence the resulting allocation remains MPB. We repeat payment raises until $\ell$ is reachable from $h$. After this, we transfer a chore from $h$ to $\ell$ along a path in the MPB graph, and repeat this process until the allocation is balanced. \cref{alg:balanced} runs in polynomial time, as there are at most $m$ transfers, and there are at most $n$ payment raises between two transfers.

\begin{restatable}{theorem}{thmBalanced}\label{thm:balanced}
For any instance $(N,M,D)$, \cref{alg:balanced} returns a balanced \fpo allocation $\x$ in polynomial time. In particular, $\x$ is \eftwo and \fpo when $m\le 2n$, and \efone and \fpo when $m\le n$.
\end{restatable}
\begin{proof}
Let $(\x,\p)$ be an allocation in the run of \cref{alg:balanced} before termination. We claim that the allocation excluding $\x_h$ is always a balanced allocation, i.e., for $i_1, i_2 \in N \setminus \{h\}$, $|\lvert \x_{i_1} \rvert - \lvert \x_{i_2} \rvert| \leq 1$. We prove this by induction. Indeed, the initial allocation excluding $h$ is trivially balanced. A chore transfer step from $h$ to $\ell$ along a path $P$ only changes the number of chores assigned to $h$ and $\ell$ and no other agents in $P$. Since $\ell$ has the fewest chores and gains only one additional chore, the allocation excluding $h$ remains balanced. Since the payment raise step does not change the allocation, the claim holds by induction. We next show that:
\begin{claim}
While the allocation $\x$ is not balanced, $h$ is the unique agent with the highest number of chores.    
\end{claim}
\begin{proof}
Since \cref{alg:balanced} has not terminated, $|\x_h| > |\x_\ell| + 1$. Since the allocation excluding $h$ is balanced, we have for any $i \in N \setminus \{h\}$ that $\lvert \x_i \rvert - \lvert \x_\ell \rvert \leq 1$, implying that $|\x_i| \le |\x_\ell| +1$. This shows $|\x_i| < |\x_h|$ for any $i\in N\setminus\{h\}$, thus proving the claim.    
\end{proof} 

We can now complete the proof for the termination of Algorithm~\ref{alg:balanced}. In any chore transfer step, agent $h$ loses exactly one chore, so $\lvert \x_h \rvert$ decreases by one. Additionally, $\lvert \x_h \rvert$ does not change in any payment raise iteration. Algorithm~\ref{alg:balanced} therefore terminates after at most $m$ chore transfer steps. We claim that there can be at most $n$ payment raise iterations between chore transfer iterations. A payment raise step results in an agent $i\notin C_h$ getting added to $C_h$, and does not remove any agents from $C_h$. Thus, after at most $n$ payment raise steps it must be that $\ell$ is reachable from $h$, and \cref{alg:balanced} performs a chore transfer step. Thus, Algorithm~\ref{alg:balanced} terminates after at most $mn + m$ iterations. On termination with $\x$, it must be that $|\x_h|-|\x_\ell|\le 1$, implying that the final allocation is balanced. In particular, if $m\le 2n$, each agent has at most two chores and $\x$ is \eftwo. Likewise, if $m\le n$, each agent has at most one chore and $\x$ is \efone.

Since the initial allocation is MPB, and every transfer step and payment raise preserves the MPB condition, the resulting allocation is MPB as well. Thus, \cref{thm:fpo} shows that the final allocation is \fpo.
\end{proof}

\subsection{Algorithms for Computing Approximately-EF1 and PO Allocations}\label{sec:efone}
We now turn to the existence of approximately-\efone and \po allocations. The main result of this section is:
\begin{restatable}{theorem}{thmEFOneMain}\label{thm:ef1-main}
Any chore allocation instance with $n$ agents admits an $(n-1)$-\efone and \fpo allocation.
\end{restatable}

First, we observe from \cref{thm:balanced} that if $m \le n$, an \efone and \fpo allocation can be computed in polynomial time using \cref{alg:balanced}. Hence we assume $m \ge n$ in the remainder of the section.  

\paragraph{\cref{alg:efone-rounding}: $2(n-1)$-EF1 and PO for $m\ge n$.} Following the ideas developed in \cref{sec:errounding} which rounds a fractional solution of an ER equilibrium, the natural approach towards obtaining an approximate-\efone guarantee is to ensure that every agent gets at most one high paying chore in the rounded solution. Clearly, this requires the number of high paying chores to be at most $n$, which cannot be guaranteed for earning limit $\beta < 1$. However, \cref{lem:earning-lb} does not show good lower bounds on the agent earnings when $\beta = 1$. 

To fix this, we design \cref{alg:efone-rounding} by modifying the rounding procedure of \cref{alg:rounding}. Since $m\ge n$, an ER equilibrium with $\beta=1$ exists. Given such an equilibrium $(\y,\p)$, \cref{alg:efone-rounding} defines $L$ to be the set of chores with payment at most $\frac{\beta}{2} = \frac{1}{2}$, and $H$ to be the set of chores with payment exceeding $\frac{1}{2}$. We prove that by using the same rounding procedure but with $L$ and $H$ defined this way, we obtain an integral MPB allocation $(\x,\p)$ where $\p_{-1}(\x_i) \le 1$ and $\p(\x_i) \ge \frac{1}{2(n-1)}$ for all agents $i\in N$. Using \cref{lem:pEF1impliesEF1}, this implies that $\x$ is $2(n-1)$-\efone and \fpo.

\begin{restatable}{theorem}{thmTwonEFone}\label{thm:2nEF1}
Given an ER equilibrium for an instance $(N,M,D)$ where $m\ge n$, \cref{alg:efone-rounding} returns a $2(n-1)$-\efone and \fpo allocation in polynomial time.
\end{restatable}

The pseudocode of \cref{alg:efone-rounding} and its analysis is presented in \cref{app:efone}. In \cref{ex:ef1-tight}, we show a lower bound on our approach with $\beta = 1$ by presenting an instance for which no rounding of the ER equilibrium is $(n-1-\delta)$-\efone, for any $\delta > 0$. This shows that the lower and upper bounds from this approach have a gap of factor $2$, and leaves open the question of whether there is a better algorithm for rounding the ER equilibrium to achieve an approximate-EF1 guarantee.

\paragraph{Improved algorithm guaranteeing $(n-1)$-\efone and \po.}
We show that this gap in the approximation factor can be closed designing an improved algorithm which returns an $(n-1)$-\efone and \po allocation given an ER equilibrium. At a high level, our algorithm first considers the $2$-\eftwo and \fpo allocation $(\z,\p)$ returned by \cref{alg:efone-rounding}. We obtained this fairness guarantee by showing that $\p_{-1}(\z_h) \le 1$ and $\p(\z_i) \ge \frac{1}{2(n-1)}$ for all agents $i, h \in N$. Improving the lower bound to $\p(\z_i) \ge \frac{1}{n-1}$ for all $h\in N$ would imply that $\z$ is $(n-1)$-\efone and \po. Our improved algorithm aims to construct such an allocation in the event that $\z$ is not already $(n-1)$-\efone. To do so, we \textit{unroll} \cref{alg:efone-rounding} and carefully identify the events which caused the earning $\p(\z_i)$ to go below $\frac{1}{n-1}$ in \cref{alg:efone-rounding} for some agent $i$. This must have happened due to an suboptimal choice of rounding, and our algorithm corrects this. We prove that such suboptimal choices can be identified and fixed in polynomial time so that in the final allocation $(\z,\p)$, every agent $i$ satisfies $\p(\z_i) \ge \frac{1}{n-1}$, as well as $\p_{-1}(\z_i)\le 1$. This proves that the allocation returned by our improved algorithm is $(n-1)$-\efone and \fpo. We present and discuss our improved algorithm in detail in \cref{app:rebalancing}.
\begin{restatable}{theorem}{thmRebalancing}\label{thm:rebalancing}
Given an ER equilibrium of an instance with $m\ge n$, an $(n-1)$-\efone and \fpo allocation can be found in polynomial time.
\end{restatable}
\cref{thm:balanced} and \cref{thm:rebalancing} together prove \cref{thm:ef1-main} on the existence of $(n-1)$-\efone and \po allocations for all chore allocation instances with $n$ agents.

\section{Existence of 4-EFX Allocations}\label{sec:efx}
In this section, we prove the main result of our paper showing the existence of approximately-\efx allocations of chores for all instances.

\thmEFXMain*

We prove \cref{thm:efx-main} through two algorithms: \cref{alg:efx-small} which returns an exact EFX allocation for instances with $m \le 2n$, and \cref{alg:const-efx} which returns a $4$-EFX for instances with $m\ge 2n$.

\paragraph{\cref{alg:efx-small}: \efx for $m\le 2n$.} We algorithmically show that when $m\le 2n$, an EFX allocation exists and can be computed in polynomial time. The existance of EFX allocation for this case is previously known via an $O(n^3)$-time algorithm that uses matching based techniques \cite{mahara2023efxmatching}. Our algorithm, \cref{alg:efx-small}, is faster (runs in $O(n^2)$ time) and arguably simpler than that of \cite{mahara2023efxmatching}. More importantly, our algorithm returns an EFX allocation with certain special properties and introduces the idea of \textit{chore swaps}, both of which are important of our algorithm computing a $4$-\efx allocation in the general case. We present and analyze \cref{alg:efx-small} in \cref{sec:efx-small}. Formally, we prove that:
\begin{restatable}{theorem}{thmEFXSmall}\label{thm:efx-small}
For a chore allocation instance with $n$ agents ordered $1$ through $n$, and $m\le 2n$ chores, \cref{alg:efx-small} returns in $O(n^2)$ time an allocation $\x$ s.t. 
\begin{itemize}
\item[(i)] $\x$ is \efx.
\item[(ii)] If $m > n$, then $|\x_i| = 1$ for all $i>m-n$. That is, the last $n$ agents in the order receive a single chore each.
\end{itemize}
\end{restatable}

\paragraph{\cref{alg:const-efx}: 4-\efx for $m\ge 2n$.} With the above result, it only remains to establish \cref{thm:const-efx} for instances with $m\ge 2n$. For this setting, we design a polynomial time algorithm \cref{alg:const-efx}, which computes a $4$-\efx allocation for a chore allocation instance with $m\ge 2n$, given its ER equilibrium with earning limit $\beta=\frac{1}{2}$. \cref{alg:const-efx} uses \cref{alg:efx-small} as a subroutine and crucially relies on the properties outlined in \cref{thm:efx-small}. We provide a detailed overview of \cref{alg:const-efx} in \cref{sec:efx-overview} and then present its analysis in \cref{sec:efx-analysis}. 

\begin{restatable}{theorem}{thmConstEFX}\label{thm:const-efx}
Given an ER equilibrium of a chore allocation instance with $m\ge 2n$ and earning limit $\beta=\frac{1}{2}$, \cref{alg:const-efx} returns a $4$-\efx allocation in polynomial time.
\end{restatable}

\noindent Thus, Theorems~\ref{thm:efx-small} and \ref{thm:const-efx} together prove \cref{thm:efx-main}, and are proved in the next two sections.

\begin{algorithm}[!t]
\caption{Computes an EFX allocation for instances with $m\le 2n$}\label{alg:efx-small}
\textbf{Input:} Instance $(N,M,D)$ with $m \le 2n$\\
\textbf{Output:} An integral allocation $\x$   
\begin{algorithmic}[1]
\State $r \gets \max\{0, m-n\}$
\State $M' \gets M$, $\x_i \gets \emptyset$ for all $i\in [n]$
\Statex \textit{--- Phase 1: Agents in $N_2$ pick chores in order from $r$ to $1$ ---}
\For{$i=r$ down to $1$}
\State $e_i \gets \arg\min_{j\in M'} d_i(j)$
\State $\x_i \gets \x_i \cup \{e_i\}$, $M' \gets M' \setminus \{e_i\}$
\EndFor
\Statex \textit{--- Phase 2: Agents pick chores in order from $1$ to $n$  ---}
\For{$i=1$ to $n$}
\State $j_i \gets \arg\min_{j\in M'} d_i(j)$
\State $\x_i \gets \x_i \cup \{j_i\}$, $M' \gets M' \setminus \{j_i\}$
\EndFor
\Statex \textit{--- Phase 3: Chore swaps ---}
\While{$\x$ is not EFX}
\State $i \gets \arg\min\{i'\in N: i' \text{ is not EFX}\}$
\State $\ell \gets \arg\min\{d_i(\x_h): h\in N\}$
\State Perform $(i, \ell)$ swap: $\x_i \gets \x_i \cup \x_\ell \setminus \{j_i\}$, $\x_\ell \gets \{j_i\}$
\EndWhile
\State \Return $\x$
\end{algorithmic}
\end{algorithm}

\subsection{EFX for $m\le 2n$}\label{sec:efx-small}
This section proves \cref{thm:efx-small} by describing and analyzing \cref{alg:efx-small}, which computes an EFX allocation for instances with $m\le 2n$. \cref{alg:efx-small} iteratively allocates chores to agents in three phases. We initialize $M'$ to $M$ and update this set as chores are allocated. For simplicity, we first assume $m>n$, letting $r = m-n$.

In Phase 1, proceeding in the order $r, r-1, \dots, 1$, each agent $i$ iteratively picks their least disutility chore $e_i$ among the set of remaining items $M'$. Let $L = \{e_1,\dots, e_r\}$. Then in Phase 2, proceeding in the order $1$ to $n$, each agent $i$ picks their least disutility chore $j_i$ in $M'$. Let $H = \{j_1,\dots, j_n\}$. Let $\x^0$ be the allocation at the end of Phase 2. Let $N_2 = [r]$ be the set of agents with two chores in $\x^0$, and let $N_1 = [n]\setminus [r]$ be the set of agents with one chore in $\x^0$. Clearly, agents in $N_1$ are EFX, and hence only agents in $N_2$ may be EFX-envious.  

Starting with $\x^0$, Phase 3 of \cref{alg:efx-small} performs \textit{chore swaps} between an agent $i\in N_2$ who is not EFX, and the agent $\ell$ who $i$ envies the most. We refer to such a swap as an $(i, \ell)$-swap. In an $(i, \ell)$ swap in an allocation $\x$, the bundle $\x_\ell$ is transferred to $i$, and the higher disutility chore $j_i$ is transferred from $i$ to $\ell$. When there are multiple envious agents $i$, we break ties following the agent ordering. We argue that after an $(i, \ell)$ swap, the agent $i$ becomes EFX and remains EFX throughout the subsequent execution of the algorithm. Thus, every agent $i\in N_2$ undergoes an $(i, \ell)$ swap at most once, and these swaps happen in the order of agents $1$ through $r$. This implies that \cref{alg:efx-small} terminates in at most $r$ steps with an EFX allocation.

Finally, we note that $r=0$ when $m\le n$. Thus, \cref{alg:efx-small} skips Phase 1 and only executes Phase 2 and returns an allocation in which each agent gets a single chore, and hence is EFX.

\paragraph{Analysis of \cref{alg:efx-small}.} Since it is clear that \cref{alg:efx-small} returns an EFX allocation when $m\le n$, we assume $m > n$ in the following analysis. Purely for the purpose of analysis, we implement Phase 3 as follows:

\begin{algorithm}
\begin{algorithmic}[1]
\setcounter{ALG@line}{8}
\For{$i=1$ to $r$}
\If{$i$ is not EFX}
\State $\ell \gets \arg\min\{d_i(\x_h): h\in N\}$
\State Perform $(i, \ell)$ swap: $\x_i \gets \x_i \cup \x_\ell \setminus \{j_i\}$, $\x_\ell \gets \{j_i\}$
\EndIf
\EndFor
\end{algorithmic}        
\end{algorithm}

\cref{lem:efx-small-invariants} below refers to the above implementation of Phase 3 of \cref{alg:efx-small}. Let $\x^i$ denote the allocation after iteration $i$ of Phase 3.

\begin{restatable}{lemma}{lemEFXSmallInvariants}\label{lem:efx-small-invariants}
For each $i\in [r]$, 
\begin{itemize}
\item[(i)] Before iteration $i$, agents $N \setminus [i-1]$ do not participate in any swap.
\item[(ii)] In iteration $i$, if agent $i$ participates in an $(i, \ell)$ swap, then $i$ is \efx after the swap. Moreover, $\max_{j\in \x^i_i} d_i(\x^i_i\setminus \{j\}) < d_i(j_i)$ immediately after the swap.
\item[(iii)] After iteration $i$, agents in $N_2\cap [i]$ are \efx. Agents in $N_1$ have a single chore, and are \efx.
\end{itemize}
\end{restatable}
\begin{proof}
We prove the invariants inductively, beginning with $i=1$. Consider the allocation $\x^0$ before Phase 3 begins. We show that the invariants hold for $i=1$ as follows.
\begin{itemize}
\item[(i)] Invariant (i) holds trivially, since no agent has participated in any swap before iteration $1$.
\item[(ii)] Note that $\x^0_1 = \{e_1, j_1\}$ and $j_h \in \x^0_h$ for any $h\ge 2$ in the allocation $\x^0$ before iteration $1$. By the order in which agent $1$ picks chores, we have $d_1(e_1) \le d_1(j_1) \le d_1(j_h)$ for any $h\ge 2$. Hence, $\max_{j\in \x^0_1} d_1(\x^0_1\setminus\{j\}) = d_1(j_1) \le d_1(j_h) \le d_1(\x^0_h)$. This shows that agent $1$ is EFX in the allocation $\x^0$, and hence no swap takes place in iteration $1$. Invariant (ii) thus holds vacuously.
\item[(iii)] As argued above, agent $1$ is EFX in $\x^0$, and since no swap takes place in iteration $1$, we have $\x^1 = \x^0$. In $\x^1$, agent $1$ is EFX, and agents in $N_1$ have a single chore, and hence are EFX. Thus, invariant (iii) holds.
\end{itemize}

Assume that invariants (i)-(iii) hold for some $i\in [r-1]$. We will prove that the invariants hold for $i+1$ as well. 
\begin{itemize}
\item[(i)] Let $i$ and $\ell$ be the agents participating in an $(i, \ell)$ swap in iteration $i$. By invariant (i) of the inductive hypothesis, agents $\{i, i+1, \dots, n\}$ have not undergone a swap before iteration $i$. Hence the allocation before iteration $i$ satisfies $\x^{i-1}_h = \x^0$ for any $h\in N_1 \cup N_2\setminus [i-1]$. By the order in which agent $i$ picks chores, we have:
\[
\max_{j\in \x^{i-1}_i} d_i(\x^{i-1}_1\setminus\{j\}) = d_i(j_i) \le d_i(j_h) \le d_i(\x_h^{i-1}),
\]
which shows that agent $i$ does not EFX-envy any agent in $\{i+1, \dots, n\}$. Since $i$ is EFX-envious before iteration $i$, we must have $\ell \in [i-1]$. Thus, after iteration $i$, only agents in $[i]$ have participated in swaps, establishing invariant (i).
\item[(ii)] Suppose agent $(i+1)$ participates in an $(i+1, k)$ swap with agent $k$ in iteration $(i+1)$ resulting in the allocation $\x^{i+1}$. Before iteration $(i+1)$, we know from invariant (i) that agents in $N\setminus [i]$ have not participated in any swap. Thus, $\x^i_{i+1} = \{e_{i+1}, j_{i+1}\}$. Using the fact that agent $(i+1)$ is not EFX in the allocation $\x^i$, and the choice of agent $k$, we have $d_{i+1}(j_{i+1}) > d_{i+1}(\x^i_k)$. Once again, using the order in which agent $(i+1)$ picked chores, we see that $k\in [i]$. 

Next, we claim that $d_{i+1}(e_{i+1}) \le d_{i+1}(j)$ for any $j\in \x^i_k$. To see this, note that by the order in which agent $(i+1)$ picks chores, the only chores that have disutility less than $d_{i+1}(e_{i+1})$ for agent $(i+1)$ could be the chores $e_{i+2}, \dots, e_r$. However, since agents $(i+1), \dots, r$ have not undergone any swap step, these chores cannot belong to $\x^i_k$. The claim thus holds.

We now prove invariant (ii). After iteration $(i+1)$, we have $\x^{i+1}_{i+1} = \{e_{i+1}\} \cup \x^i_k$, and $\x^{i+1}_k = \{j_{i+1}\}$. Observe that:
\[
\max_{j\in \x^{i+1}_{i+1}} d_{i+1}(\x^{i+1}_{i+1}\setminus\{j\}) = d_{i+1}(\x^i_k) < d_{i+1}(j_{i+1}),
\]
where we use the claim that $d_{i+1}(e_{i+1}) \le d_{i+1}(j)$ for any $j\in \x^i_k$ in the first equality, and the fact that agent $(i+1)$ EFX-envies agent $k$ in the allocation $\x^i$ in the second inequality. This proves the second claim of invariant (ii) and also shows that agent $(i+1)$ does EFX-envy agent $k$ immediately after the swap. Consider some other agent $h\notin \{i+1, k\}$. By the choice of agent $k$, we have $d_{i+1}(\x^i_k) \le d_{i+1}(\x^i_h)$. Thus, $\max_{j\in \x^{i+1}_{i+1}} d_{i+1}(\x^{i+1}_{i+1}\setminus\{j\}) = d_{i+1}(\x^i_k) \le d_{i+1}(\x^i_h)$. This shows that agent $(i+1)$ is EFX after the swap in iteration $(i+1)$.
\item[(iii)] Consider the allocation $\x^{i+1}$ after iteration $(i+1)$. Invariant (ii) shows that agent $(i+1)$ is EFX in $\x^{i+1}$. Moreover, agent $k$ is EFX in $\x^{i+1}$ since she has a single chore. Since we argued above that $k\in [i]$, the agents in $N_1$ continue to have a single chore and are EFX. 

Thus, it only remains to be shown that an agent $h\in [i]\setminus \{k\}$ who was EFX in the allocation $\x^i$ remains EFX in the allocation $\x^{i+1}$ after the $(i+1, k)$ swap. First note that $h$ is EFX towards $\x^{i+1}_{h'}$ for any $h'\notin \{i+1, k\}$, since $\x^i_{h'} = \x^{i+1}_{h'}$ and $h$ is EFX in $\x^i$.

Next, observe that agent $h$ is EFX towards the bundle $\x^i_k$. Since $\x^i_k \subset \x^{i+1}_{i+1}$, agent $h$ is EFX towards $\x^{i+1}_{i+1}$ as well. 

Finally, we show that $h$ is EFX towards the bundle $\x^{i+1}_k = \{j_{i+1}\}$. If agent $h$ underwent a swap of the form $(i', h)$ during iteration $i'\in[h+1, i+1]$, then agent $h$ has a single chore and will be EFX in $\x^{i+1}$. Hence we assume agent $h$ did not undergo any swap during iterations $[h+1,i+1]$, and hence $\x^{i+1}_h = \x^h_h$. Now, observe that:
\[ 
\max_{j\in \x^h_h} d_h(\x^h_h\setminus \{j\}) \le d_h(j_h).
\]
This is true, because if $h$ does not undergo a swap in iteration $h$, we have $\x^h_h = \{e_h, j_h\}$, and $d_h(e_h)\le d_h(j_h)$. If $h$ does undergo a swap in iteration $h$, invariant (ii) implies the same observation. Now observe that:
\begin{align*}
\max_{j\in \x^{i+1}_h} d_h(\x^{i+1}_h\setminus \{j\}) &= \max_{j\in \x^h_h} d_h(\x^h_h\setminus \{j\}) \tag{since $\x^{i+1}_h = \x^h_h$} \\
&\le d_h(j_h) \tag{observed above}\\
&\le d_h(j_{i+1}) \tag{by the order in which $h$ picked chores}\\
&= d_h(\x^{i+1}_k), \tag{since $\x^{i+1}_k = \{j_{i+1}\}$}     
\end{align*}
which shows that $h$ is EFX towards the bundle $\x^{i+1}_k$. In conclusion, invariant (iii) holds.
\end{itemize}
By induction, the invariants hold for all $i\in [r]$.
\end{proof}

With \cref{lem:efx-small-invariants} in hand, \cref{thm:efx-small} follows immediately.
\begin{proof}[Proof of \cref{thm:efx-small}]
For $m\le n$, \cref{thm:efx-small} only executes Phase 2 which assigns a single chore to each agent, thus returning an EFX. For $m>n$, Invariant (iii) of \cref{lem:efx-small-invariants} for $i=r$ implies that in the allocation returned by \cref{alg:efx-small}, agents in $N_2$ are EFX. Moreover, agents in $N_1 = [n]\setminus [r]$ have a single chore and hence are EFX. This proves both properties (1) and (2) claimed by \cref{thm:efx-small}.

Finally, note that \cref{alg:efx-small} runs in $O(n^2)$-time: Phases 1 and 2 involve $m \le 2n$ steps of identifying an agent's favorite chore ($O(m)$ time each), and Phase 3 involves at most $r \le n$ swap steps.
\end{proof}

\subsection{Computing a 4-EFX Allocation: Algorithm Overview}\label{sec:efx-overview}
We now prove the existence of $4$-EFX allocations for chore allocation instances with $m \ge 2n$. We design a polynomial time algorithm, \cref{alg:const-efx}, which returns a $4$-EFX allocation for an instance when given an ER equilibrium of the instance with earning limit $\beta=\frac{1}{2}$ as input.

\Cref{alg:const-efx} first runs \Cref{alg:rounding} on the given ER equilibrium $(\y,\p)$ with earning limit $\beta=\frac{1}{2}$ to obtain a $2$-\eftwo and \fpo allocation $(\x,\p)$. As in \cref{sec:rounding}, we classify chores based on their payments as $L = \{j\in M: p_j \le \frac{1}{2}\}$ and $H = \{j\in M: p_j > \frac{1}{2}\}$. Thus $L$ is the set of low paying chores, and $H$ is the set of high paying chores, whose payment exceeds the earning limit $\beta = \frac{1}{2}$. 

We partition the bundle of each agent $i$ as $\x_i = S_i \cup H_i$, where $S_i \subseteq L$ and $H_i \subseteq H$. Let $N_H$ denote the set of agents who are assigned high paying chores, and let $N_0 = N\setminus N_H$. The following lemma records properties of the allocation $\x$.

\begin{restatable}{lemma}{lemInitialAlloc}\label{lem:initial-alloc}
The allocation $(\x,\p)$ returned by \cref{alg:rounding} satisfies:
\begin{itemize}
\item[(i)] For any $i\in N$: $\p(\x_i) = \p(S_i \cup H_i) \ge \frac{1}{2}$.
\item[(ii)] For $i\in N_H$: $\p(S_i) \le 1$.
\item[(iii)] For $i\in N_0$: $\p(\x_i) = \p(S_i) \le \frac{3}{2}$.
\end{itemize}
\end{restatable}
\begin{proof}
For $\beta=\frac{1}{2}$, Lemmas~\ref{lem:earning-ub} and \ref{lem:earning-lb} imply that for all agents $i\in N$, $(\x,\p)$ satisfies $\p(\x_i) \ge \frac{1}{2}$, and either $\p_{-1}(\x_i) \le 1$, or $|\x_i \cap H| = 2$ and $\p_{-2}(\x_i)\le \frac{1}{2}$. The lemma then follows from the definitions of the partition of agents $N = N_H \sqcup N_0$ and the chores bundles $\x_i = S_i \sqcup H_i$. 
\end{proof}

\begin{algorithm}[!t]
\caption{Computes a $4$-EFX allocation}\label{alg:const-efx}
\textbf{Input:} Instance $(N,M,D)$ with $m\ge 2n$ and its ER equilibrium $(\y,\p)$ with $\beta = \frac{1}{2}$\\ 
\textbf{Output:} An integral allocation $\x$   
\begin{algorithmic}[1] 
\Statex \textit{--- Phase 1: Compute a $2$-\eftwo and \po allocation ---}
\State $(\x,\p) \gets$ \cref{alg:rounding}$(\y,\p)$
\State $L = \{j \in M: p_j \le \frac{1}{2}\}$, $H = \{j \in M: p_j > \frac{1}{2}\}$ \Comment{Low, High paying chores}
\State Partition each $\x_i = S_i \cup H_i$, where $S_i \subseteq L$ and $H_i \subseteq H$
\State $N_H \gets \{i\in N: H_i \neq \emptyset\}$, and $N_0 \gets \{i\in N: H_i = \emptyset\}$. 
\Statex \textit{--- Phase 2: Re-allocate $H$ ---}
\State Re-order agents s.t. agents in $N_H$ are ordered before agents in $N_0$
\State $\z' \gets $ \efx allocation of $H$ to $N$ using \cref{alg:efx-small}, with agents ordered as above
\State Partition agents into $N_L, N_H^1, N_H^2$ using \cref{def:classification-efx}
\State $H'\gets \cup_{i\in N_H^1} \z'_i$
\State $(\z,\q) \gets$ Min cost matching of $H'$ to $N_H^1$, and associated dual variables
\State For each $i\in N_H^1$, $\x_i \gets S_i \cup \z_i$
\State For each $i\notin N_H^1$, $\x_i \gets S_i \cup \z'_i$
\Statex \textit{--- Phase 3: Perform $(i,\ell)$ swaps for $i\in N_H^1$ ---}
\While{$\exists i\in N_H^1$ not $4$-\efx}
\State $i\gets\arg\min\{\q(\z_h) : h\in N_H^1 \text{ not $4$-\efx}\}$
\State $\ell \gets \arg\min\{d_i(\x_h) : h\in N\}$
\State Perform $(i,\ell)$ swap: $\x_i \gets \x_i \cup \x_\ell \setminus \z_i$, $\x_\ell \gets \z_i$
\EndWhile
\State \Return $\x$
\end{algorithmic}
\end{algorithm}

Notice that the above lemma implies that agents in $N_0$ are $3$-EFX. Thus, $\x$ may not be $O(1)$-\efx only because of agents in $N_H$, who are assigned one or more high paying chores. Therefore, our algorithm must address the $O(1)$-\efx envy of agents in $N_H$.

\paragraph{Chore swaps.} For simplicity, let us assume for the moment that all agents in $N_H$ are assigned a single high paying chore. Consider an agent $i$ with $H_i = \{j_i\}$, who is not $4$-\efx in $\x$. To fix the $4$-\efx envy that $i$ has towards other agents, we re-introduce the idea of a \textit{`chore swap'}:

\begin{definition}\label{def:swap}
Consider an allocation $\x$ in which an agent $i\in N_H$ is not $4$-\efx. Let $j_i$ be the high paying chore in $\x_i$. Let $\ell$ be the agent who $i$ envies the most, i.e. $\ell = \arg\min\{h\in N: d_i(\x_h)\}$. An $(i,\ell)$ swap on the allocation $\x$ results in an allocation $\x'$ obtained by transferring all the chores of $\ell$ to $i$, and transferring the chore $j_i$ from $i$ to $\ell$. That is, $\x'_i = \x_i \cup \x_\ell \setminus \{j_i\}$, $\x_\ell = \{j_i\}$, and $\x'_h = \x_h$ for all $h\neq \{i,\ell\}$. 
\end{definition}

Note the similarity to the definition of chore swaps involved in \cref{alg:efx-small}. Similar to the analysis in \cref{lem:efx-small-invariants}, we claim that immediately after the $(i,\ell)$ swap, $i$ is $4$-\efx towards all agents. To see this, let us scale the disutility function of each agent so that every agent has MPB ratio 1. This allows us to measure payments and disutilities on the same scale. Let $\x'$ be the allocation resulting from an $(i,\ell)$ swap on $\x$. Note that $\x'_i = \x_i \cup \x_\ell \setminus \{j_i\} = S_i \cup \x_\ell$, and $\x'_\ell = \{j_i\}$, while $\x'_h = \x_h$ for all $h\neq \{i, \ell\}$. Since $\p(S_i) \le 1$ and $\p(\x_\ell) \ge \frac{1}{2}$, we have that $d_i(\x'_i) \le 3\cdot d_i(\x_\ell)$. By the choice of $\ell$, for any $h\neq\{i,\ell\}$, $d_i(\x_\ell) \le d_i(\x_h) = d_i(\x'_h)$. Thus $d_i(\x'_i) \le 3\cdot d_i(\x'_h)$, showing that $i$ does not $3$-\efx envy (and hence $4$-\efx envy) agent $h$ after the swap. Similarly, the fact that $i$ is $4$-\efx envious of the bundle $\x_\ell$ establishes a lower bound on the disutility of $j_i$ for $i$, which we can use to prove that $i$ will not $4$-\efx envy $\x'_\ell = \{j_i\}$ after the swap. Moreover, agent $\ell$ is \efx after the swap since she has a single chore.

In conclusion, after an $(i,\ell)$ chore swap, both agents $i$ and $\ell$ are $4$-\efx, i.e., the $O(1)$-\efx envy of agent $i$ is temporarily resolved. The above idea suggests repeatedly performing chore swaps until the allocation is $O(1)$-\efx. However, two things remain unclear: (i) how to address agents in $N_H$ with two high paying chores, and  (ii) whether an agent $i$ who underwent a swap develops \efx-envy subsequently in the run of the algorithm. Our algorithm addresses both these issues by separately treating high paying chores and the agents to whom they are assigned, and using clever design choices. 

\paragraph{Re-allocating high paying chores.} Observe that $\p(S_i) \le O(1)\cdot p_j$ for any agent $i$ and high paying chore $j\in H$. This means that for any agent $i$, the chores in $S_i$ have cumulatively less payment than \textit{any} single high paying chore, up to a constant factor. Thus we should `balance' out the envy created among the agents due to an imbalanced allocation of the high paying chores. To do this, we compute an \efx allocation $\z'$ of the high paying chores $H$ using \cref{alg:efx-small}. This is possible since there are at most $2n$ high paying chores, i.e., $|H| \le 2n$ as each agent has at most two high paying chores in the rounded allocation. In our invocation of \cref{alg:efx-small}, \textit{we order the agents in $N_H$ to appear before the agents in $N_0$}. 

We then classify the agents based on the EFX allocation $\z'$ as follows:
\begin{definition}\label{def:classification-efx}(Classification of Agents in \cref{alg:const-efx})
Agents are classified as:
\begin{itemize}
\item $N_L = \{i\in N: |\z'_i| = 0\}$, i.e. agents with no high paying chores.
\item $N_H^1 = \{i\in N: |\z'_i| = 1\}$, i.e. agents with a single high paying chore.
\item $N_H^2 = \{i\in N: |\z'_i| \ge 2\}$, i.e. agents with at least two high paying chores.
\end{itemize}
\end{definition}
Note that agents in $N_H^2$ can have more than two high paying chores since we re-allocated the high paying chores via the \efx allocation $\z'$. However, since $\z'$ is \efx, it cannot be that both $N_H^2 \neq \emptyset$ and $N_L \neq \emptyset$. 

More importantly, by ordering the agents in $N_H$ before agents in $N_0$, we can leverage \cref{thm:efx-small} to show that $N_H^2 \subseteq N_H$. That is, if an agent $i$ obtains two or more high paying chores after re-allocating the high paying chores, then $i$ must have had a high paying chore to begin with. Recall from \cref{lem:initial-alloc} that the earning of such agents is at most $1$ from chores in $L$, i.e., that $\p(S_i)\le 1$ for each such agent $i\in N_H$. This property is useful to bound the total EFX-envy of agents in $N_H^2$.

Having re-allocated the high paying chores $H$, we add back the chores from $L$ to obtain the allocation $\x'$ given by $\x'_i = S_i \cup \z'_i$ for all agents $i$. We note that each agent agent $i\in N_H^2$ is actually $O(1)$-\efx in $\x'$: since $\p(S_i) \le O(1)\cdot \p(\z'_i)$, we have that $d_i(\x'_i) \le O(1)\cdot d_i(\z'_i)$. Since $\z'$ is \efx and $i$ has at least two high paying chores, we have that $d_i(\z'_i) \le 2\cdot d_i(\z'_k)$ for all $k$. Thus $d_i(\x'_i) \le O(1)\cdot d_i(\x'_k)$. This proves a surprising property of the allocation $\x'$: the agents in $N_H^2$ who have two or more high paying chores are actually $O(1)$-\efx! The \efx re-allocation of the $H$ chores thus leaves us to tackle the agents in $N_H^1$ with exactly one high paying chore. For these agents, we use chore swaps as described earlier. 

\paragraph{Performing swaps involving $N_H^1$ agents.}
Consider an $(i,\ell)$ swap between an agent $i\in N_H^1$ who was not $4$-\efx and the agent $\ell$ who $i$ envied the most. We argued that after the swap $D_i \le 4\cdot d_i(j_i)$, where $D_i$ is the disutility of $i$ after the swap, and $j_i\in H$ is the high paying chore of $i$ that was transferred to $\ell$. Consider a subsequent swap $(h, k)$ between $h\in N_H^1$ and $k\in N$, after which the high paying chore $j_h \in H$ of agent $h$ is (the only chore) assigned to $k$. Roughly speaking, since $i$ was $4$-\efx after the $(i,\ell)$ swap, $i$ does not $4$-\efx envy $k$'s bundle before the swap. Hence, $i$ will not envy $h$ after the $(h,k)$ swap. However, it could happen that $i$ develops $O(1)$-\efx envy towards $k$ after the $(h,k)$ swap, if the swaps are made arbitrarily. But observe that if $i$'s disutility for $j_h$ is at least that of $j_i$, then we will have that $D_i \le 4\cdot d_i(j_i) \le 4\cdot d_i(j_h)$, showing that $i$ will not $4$-\efx envy $k$, who has $j_h$ after the $(h,k)$ swap. This observation suggests that we can avoid agents who participate in a swap from becoming envious again \textit{by performing swaps in a carefully chosen order}. This order depends on the disutilities of $N_H^1$ agents for the set $H'$ of high paying chores assigned to them, i.e., $H' = \cup_{i\in N_H^1} \z'_i$. 

To determine this order, we re-allocate $H'$ to the $N_H^1$ agents by computing an fPO allocation $(\z, \q)$, where every agent in $N_H^1$ gets exactly one chore of $H'$ in $\z$, and $\q$ is the set of supporting payments. We show in \cref{lem:overpaying-fpo} that such an allocation $\z$ can be found by solving a linear program for minimum cost matching, and the payments $\q$ can be computed from the dual variables of this program. \cref{alg:const-efx} then performs chore swaps in the following order: at each time step $t$, among all the agents in $N_H^1$ who are not $4$-\efx, we pick the agent $i$ with the high paying chore with the minimum payment $\q(\z_i)$, and perform an $(i, \ell)$ swap. An involved analysis shows that this design choice ensures \cref{alg:const-efx} does not cause an agent in $N_H^1$ to re-develop $4$-\efx envy. \cref{alg:const-efx} thus terminates in at most $n$ steps.

With the above ideas, we argue that the resulting allocation $\x$ is $4$-\efx: (i) agents in $N_H^1$ are addressed via swaps, (ii) agents in $N_H^2$ remain $4$-\efx during swaps, and (iii) agents in $N_L$ are $3$-\efx since their earning is at most $\frac{3}{2}$ and every agent has an earning of at least $\frac{1}{2}$; the latter property is maintained during every swap as well. Thus, \cref{alg:const-efx} returns a $4$-\efx allocation.

\subsection{Computing a 4-EFX Allocation: Algorithm Analysis}\label{sec:efx-analysis}
We begin by providing a recap of \cref{alg:const-efx}. In Phase 1, we compute a $2$-\eftwo and \fpo allocation $(\x', \p)$ using the \cref{alg:rounding} algorithm. We normalize the disutilities so that the MPB ratio of each agent is $1$ for the payment vector $\p$. We then partition the chores of each agent $i$ as $\x'_i = S_i \cup H$, where $S_i \subseteq L$ contains low paying chores and $H_i\subseteq H$ contains high paying chores. This partitions the set of agents as $N = N_H \sqcup N_0$, where agents in $N_H$ receive one or two high paying chores and agents in $N_0$ receive none. In Phase 2, we re-allocate $H$ by computing an \efx allocation $\z'$ using \cref{alg:efx-small} with agents ordered as $N_H$ first followed by $N_0$. We then categorize agents into $N_H^1, N_H^2$ and $N_L$ depending on the number of $H$ chores they are assigned (see \cref{def:classification-efx}) in $\z'$. The following is a useful property of $\z'$. 

\begin{restatable}{lemma}{lemEFXNHTwo}\label{lem:efx-nh2}
With agents ordered as $N_H$ first followed by $N_0$, \cref{alg:efx-small} computes an EFX allocation $\z'$ of the high paying chores $H$ s.t. for each agent $i$ with $|\z'_i|\ge 2$, we have $\p(S_i) \le 1$. In other words, $N_H^2 \subseteq N_H$.
\end{restatable}
\begin{proof}
From \cref{lem:initial-alloc}, we know that for $i\in N_H$, we have $\p(S_i) \le 1$. Hence it suffices to argue that if $|\z'_i| \ge 2$, then $i\in N_H$. We will prove the contrapositive statement: if $i\in N_0$, then $|\z_i'| \le 1$. 

First note that if $|H| \le n$, then $|\z'_i| \le 1$ for all $i\in N$. If $|H| > n$, then \cref{thm:efx-small} shows that \cref{alg:efx-small} returns an EFX allocation in which agents with index greater than $r := |H| - n$ have a single chore. These are the last $2n - |H|$ agents in the order. Thus, it suffices to prove that $|N_0| \le 2n - |H|$, since the agents in $N_0$ appear last in our order. 

Clearly, $n = |N_H| + |N_0|$. Since agents in $N_H$ have exactly one or two high paying chores, we have $|H| \le 2\cdot |N_H|$. This gives $|N_0| + |H| \le |N_0| + 2\cdot |N_H| \le 2n$. This implies $|N_0| \le 2n-|H|$, which is what we aimed to show.

In conclusion, any agent $i\in N_0$ has $|\z'_i| \le 1$, thus proving the lemma.
\end{proof}

Next, we re-compute a matching $\z$ of the chores $H' = \cup_{i\in N_H^1} \z'_i$ to agents in $N_H^1$, and let $\q$ be a set of payments of chores in $H'$ such that $(\z,\q)$ is on MPB. The following lemma shows that such an allocation $(\z,\q)$ is computable in polynomial time.

\begin{restatable}{lemma}{lemOverpayingfPO}\label{lem:overpaying-fpo}
Given a chore allocation instance with $m = n$ chores, an MPB allocation $(\z,\q)$ where $|\z_i| = 1$ for each $i\in N$ can be computed in polynomial time.
\end{restatable}
\begin{proof}
We show that the required allocation $\z$ can be computed via the following linear program for finding a minimum cost matching.

\begin{equation}\label{eq:matching}
\begin{aligned}
\text{min} &\sum_{i\in N} \sum_{j\in M} x_{ij} \log d_{ij} \\
\forall j\in M: &\sum_{i\in N} x_{ij} = 1 \\
\forall i\in N: &\sum_{j\in M} x_{ij} = 1 \\
\forall i\in N, j\in M: &\quad  x_{ij} \ge 0.
\end{aligned}
\end{equation}
Note that the objective is well-defined since $d_{ij} > 0$ for all $i\in N, j\in M$. Since the matching polytope is integral, there exists a integral optimal solution $\z$ with $|\z_i| = 1$ for all $i\in N$. We now show that we can compute chore payments $\q$ such that $(\z,\q)$ is on MPB by using dual variables of the above program. Let $\lambda_j$ and $\alpha_i$ be the dual variables corresponding to the constraints corresponding to chore $j$ and agent $i$ respectively. The stationarity KKT condition corresponding to the variable $x_{ij}$ implies:
\[ 
\log d_{ij} + \lambda_j + \alpha_i \ge 0.
\]
This implies that for all $i\in N$ and $j\in M$, $\frac{d_{ij}}{e^{-\lambda_j}} \ge e^{-\alpha_i}$. Moreover, the complementary slackness condition implies that the above inequality is an equality when $x_{ij} > 0$, i.e., $x_{ij} > 0 \Rightarrow \frac{d_{ij}}{e^{-\lambda_j}} = e^{-\alpha_i}$. We set the chore payments $\q$ as $q_j = e^{-\lambda_j} > 0$. The above observations then imply that $(\z,\q)$ is on MPB with $e^{-\alpha_i}$ denoting the MPB ratio of agent $i$.
\end{proof}

In Phase 3, we perform \textit{swaps} (\cref{def:swap}) involving agents in $N_H^1$ which are not $4$-\efx. In each such swap step, we pick the $4$-\efx envious agent $i\in N_H^1$ with the minimum $\q(\z_i)$. 

Let $\x$ be the allocation computed at the end of Phase 2, before any Phase 3 swaps are performed. For $t\in \Z_{\ge 0}$, we use the phrase `at time step $t$' to refer to the $t^{th}$ iteration of the while loop of \cref{alg:const-efx}, and use it interchangeably with `just before the swap at time step $t$'. Let $\x^t$ denote the allocation at time step $t$, with $\x = \x^0$.  We first prove a few basic invariants maintained by \cref{alg:const-efx}.

\begin{restatable}{lemma}{lemEFXInvariants}\label{lem:invariants-nh2}
At any time step $t$ in the run of \cref{alg:const-efx}, we have:
\begin{enumerate}
\item[(i)] For any agent $i\in N$, $\p(\x_i^t) \ge \frac{1}{2}$.
\item[(ii)] If an agent $i\in N_H^2\cup N_L$ has participated in a swap at time $t'<t$, then $i$ is \efx at $t$.
\item[(iii)] Any agent $i \in N_H^2$ is $4$-\efx. 
\item[(iv)] Any agent $i \in N_L$ is $3$-\efx. 
\end{enumerate}
\end{restatable}
\begin{proof} 
We first prove claim (i) by an inductive argument. For the allocation $(\x',\p)$ obtained by rounding the ER equilibrium, we have that $\p(\x'_i)\ge \frac{1}{2}$. This remains true at time step $t=0$ after the high paying chores are re-allocated in Phase 2, since each high paying chore pays at least $\frac{1}{2}$. Suppose claim (i) holds at time step $t$ before an $(i,\ell)$ swap takes place. After the swap, $\ell$ is assigned a high paying chore, hence $\p(\x^{t+1}_\ell) > \frac{1}{2}$. Moreover, since $i$ receives chores earlier assigned to $\ell$, we have $\p(\x^{t+1}_i) \ge \p(\x^t_\ell) \ge \frac{1}{2}$ using the inductive hypothesis at time $t$. Thus, claim (i) holds at every time step in the run of the algorithm.

For claim (ii), observe that an agent $i\in N_H^2 \cup N_L$ can only participate in a swap of the form $(h, i)$, where $h\in N_H^1$. Then $i$ is \efx immediately after the swap since $i$ is assigned a single chore. This remains true even after subsequent swaps that $i$ participates in, and hence $i$ remains \efx at any time step of the algorithm. 

Since the allocation of high paying chores $\z'$ is \efx, we have that $N_H^2 \neq \emptyset$ and $N_L \neq \emptyset$ cannot both be true. We first assume $N_H^2 \neq \emptyset$ and prove claim (iii). In this case, $N_L = \emptyset$. Consider an agent $i\in N_H^2$. If $i$ participated in an $(h,i)$ swap before $t$, then $i$ is \efx at $t$ due to claim (ii). Hence we assume that $i\in N_H^2$ did not participate in any swap, and thus $\x_i^t = S_i \cup \z'_i$. Let $j_0 = \arg\min_{j\in\z'_i} d_i(j)$. Since $i\in N_H^2$, $\z'_i\setminus \{j\} \neq \emptyset$. Consider any other agent $h\in N$ at time step $t$. We have that $\x_h^t \supseteq \z'_k$ for some $k\in N$ since an agent participating in a chore swap always swaps all of her high paying chores. We now show that $i$ is $4$-\efx towards $h$ as follows.
\begin{equation*}
\begin{aligned}
d_i(\x_i^t) &= d_i(S_i) + d_i(\z'_i) & \text{(since $i$ did not undergo any swap)}\\
            &= \p(S_i) + d_i(\z'_i \setminus \{j_0\}) + d_i(j_0) & \text{(using the MPB condition)}\\
            &\le \p(S_i) + 2\cdot d_i(\z'_i \setminus \{j_0\}) &\text{(since $\z'_i\setminus\{j_0\} \neq \emptyset$ and $j_0 = \arg\min_{j\in\z'_i} d_i(j))$}\\
            &\le \p(S_i) + 2\cdot d_i(\z'_k) &\text{(since $\z'$ is \efx)}\\
            &\le 1 + 2\cdot d_i(\z'_k) &\text{(since $\p(S_i) \le 1$ using \cref{lem:efx-nh2})}\\
            &\le 4\cdot d_i(\z'_k) &\text{(using $\p(\z'_k) > \frac{1}{2}$ since $\z'_k\subseteq H$)} \\
            &\le 4\cdot d_i(\x_h^t) &\text{(since $\x_h^t \supseteq \z'_k$ )}.
\end{aligned}
\end{equation*}

Finally, we assume $N_L\neq\emptyset$ and prove claim (iv). In this case, $N_H^2 = \emptyset$. Consider an agent $i\in N_L$. As before, if $i$ participated in a swap at a time before $t$, then $i$ is \efx at $t$ due to claim (ii). Hence we assume that $i\in N_L$ did not participate in any swap, and thus $\x_i^t = S_i$. Consider any other agent $h\in N$ at time step $t$. We show that $i$ is $3$-\efx towards $h$ as follows.
\begin{equation*}
\begin{aligned}
d_i(\x_i^t) &= \p(S_i) &\text{(since $\x_i^t = \x'_i = S_i$ is on MPB)}\\
            &\le \frac{3}{2} &\text{(since $\p(S_i) \le 1$ using \cref{lem:efx-nh2})} \\
            &< 3 \cdot \p(\x_h^t) &\text{(since claim (ii) shows $\p(\x^h_t) > \frac{1}{2}$)}\\
            &\le 3 \cdot d_i(\x_h^t). &\text{(using the MPB condition)}
\end{aligned}
\end{equation*}
This proves the lemma.
\end{proof}

The above lemma shows that $N_H^2$ and $N_L$ agents are $4$-\efx and we need to address the $N_H^1$ agents. Let us examine the change in disutility of agent $i\in N_H^1$ after an $(i,\ell)$ swap.

\begin{restatable}{lemma}{lemAfterSwapConst}\label{lem:after-swap-const}
Suppose $i\in N_H^1$ participates in an $(i,\ell)$ swap for the first time at time step $t$. Then, $d_i(\x_i^{t+1}) < 4\cdot d_i(\z_i)$, and  $i$ is $4$-\efx immediately after the swap. 
\end{restatable}
\begin{proof}
Since $i$ has not undergone a swap until time step $t$, we have $\x_i^t = S_i \cup \z_i$. By the definition of a swap, we have $\x_i^{t+1} = S_i \cup \x^t_\ell$, $\x^{t+1}_\ell = \z_i$, and $\x_h^{t+1} = \x_h^t$ for all $h\notin\{i,\ell\}$. Since $i$ is not $4$-\efx towards $\ell$ at time $t$, we know $d_i(\x^t_i) > 4\cdot d_i(\x^t_\ell)$. We prove the first part of the lemma using the above observations.
\begin{equation*}
\begin{aligned}
d_i(\x^{t+1}_i) &= d_i(S_i) + d_i(\x^t_\ell) \\
&< d_i(S_i) + \frac{d_i(\x_i^t)}{4} & \text{(using $d_i(\x^t_i) > 4\cdot d_i(\x^t_\ell)$)} \\
&= d_i(S_i) + \frac{d_i(S_i) + d_i(\z_i)}{4} & \text{(since $\x_i^t = S_i \cup \z_i$)} \\
&= \frac{5}{4}\cdot d_i(S_i) + \frac{1}{4}\cdot d_i(\z_i) \\
&= \frac{5}{4}\cdot\p(S_i) + \frac{1}{4}\cdot d_i(\z_i) & \text{(since $d_i(S_i) = \p(S_i)$ using the MPB condition)} \\
&< \frac{15}{4} \cdot d_i(\z_i) + \frac{1}{4}\cdot d_i(\z_i) & \text{(using $\p(S_i) \le \frac{3}{2}$ and $d_i(\z_i) > \frac{1}{2}$)} \\
&= 4\cdot d_i(\z_i).
\end{aligned}
\end{equation*}

By the choice of $\ell$, it holds that $d_i(\x^t_h) \ge d_i(\x^t_\ell)$ for $h \in N \setminus \{\ell\}$. We next prove that $i$ is $4$-\efx in the allocation $\x^{t+1}$ after the swap. 
\begin{itemize}[leftmargin=*]
\item $i$ is $4$-EFX towards $\ell$, since $d_i(\x^{t+1}_i) < 4\cdot d_i(\z_i)$ as argued previously and $\x^{t+1}_\ell = \z_i$.
\item $i$ is $4$-EFX towards an agent $h\in N \setminus \{\ell\}$, since:
\begin{equation*}
\begin{aligned}
d_i(\x^{t+1}_i) &= d_i(S_i) + d_i(\x^t_\ell) &\text{(since $\x_i^{t+1} = S_i \cup \x^t_\ell$)}\\
&= \p(S_i) + d_i(\x^t_\ell) &\text{(since $d_i(S_i) = \p(S_i)$ using the MPB condition)}\\
&\le \frac{3}{2} + d_i(\x^t_\ell) &\text{(using $\p(S_i) \le \frac{3}{2}$)} \\
&\le 3\cdot \p(\x^t_{\ell}) + d_i(\x^t_\ell) &\text{(since $\p(\x^t_\ell) \ge \frac{1}{2}$ by \cref{lem:invariants-nh2})} \\
&\le 4\cdot d_i(\x^t_\ell) &\text{(since the MPB condition implies $d_i(\x^t_\ell) \ge \p(\x^t_\ell)$)} \\
&\le 4\cdot d_i(\x^t_h) = 4\cdot d_i(\x^{t+1}_h). & \text{(by choice of $\ell$)}
\end{aligned}
\end{equation*} \qedhere
\end{itemize}
\end{proof}

The above lemma shows that an agent $i\in N_H^1$ is $4$-\efx immediately after the first $(i,\ell)$ swap she participates in. Next, we argue that such an agent cannot develop $4$-\efx envy again. The key idea is to choose among all $N_H^1$ agents who are not $4$-\efx, the agent $i$ with minimum $\q(\z_i)$. Let $\alpha_i$ denote the MPB ratio of $i$ in $(\z,\q)$. Note that $|\z_i| = 1$ for all $i\in N_H^1$. For an $(i,\ell)$ swap at time step $t$, we let $q_t = \q(\z_i)$ denote the payment of the high paying chore $\z_i$ transferred from $i$ to $\ell$. We now prove the following set of invariants of \cref{alg:const-efx}.

\begin{restatable}{lemma}{lemInvariantsConst}\label{lem:invariants-const}
At any time step $t$ in the while loop of \cref{alg:const-efx}, the following hold:
\begin{itemize}
\item[(i)] Every agent $i$ has participated in at most one $(i,\ell)$ swap until time $t$.
\item[(ii)] If an agent $i$ has participated in an $(i,\ell)$ swap at time $t'<t$, then $i$ is $4$-\efx at $t$.
\item[(iii)] If an agent $i\in N_H^1$ is not $4$-\efx at $t$, then $q(\z_i) \ge q_{t-1}$.
\item[(iv)] $q_t \ge q_{t-1} \ge \cdots \ge q_1 \ge q_0 := 0$.
\end{itemize}
\end{restatable}

\begin{proof}
We prove this by induction on $t$. Since no agent has participated in a swap before $t=1$ and $q_0 = 0$, claims (i)-(iv) are vacuously true at $t=1$. Suppose claims (i)-(iv) hold true at some time step $t \ge 1$. Consider a swap $(i,\ell)$ taking place at $t$. We prove that claims (i)-(iv) hold at time $(t+1)$ after the swap has taken place.
\begin{enumerate}
\item[(i)] Suppose $i$ has already participated in a swap at time $t'<t$, then claim (ii) of the induction hypothesis implies that $i$ is $4$-\efx at time $t$, contradicting the fact that an $(i,\ell)$ swap takes place at $t$. Thus $i$ participates in her first swap at $t$. 

\item[(ii)] We will prove that every agent $h$ who has participated in a swap at time $t' < (t+1)$ is $4$-\efx at $(t+1)$. We first consider the case of $h=i$. Note that \cref{lem:after-swap-const} implies that $i$ is $4$-\efx after the swap at $t$, i.e., $i$ is $4$-\efx at time $(t+1)$. 

Next we prove the claim for agents $h\neq i$. Suppose an agent $h\neq i$ participated in a swap of the form $(h,k)$ at time $t'<(t+1)$. Since $h\neq i$, and the $(i,\ell)$ swap takes place at time $t$, we know $t'<t$. We therefore apply claim (ii) of the induction hypothesis to obtain that $h$ is $4$-\efx at time $t$. In particular, this shows that at time $(t+1)$, agent $h$ remains $4$-\efx towards all agents $h'\neq\{i,\ell\}$ who don't participate in the $(i,\ell)$ swap at time $t$. Moreover, since $h$ is $4$-\efx towards the bundle $\x^t_\ell$ and $\x^t_{i+1} \supseteq \x^t_\ell$, $h$ remains $4$-\efx towards agent $i$ after the swap at $(t+1)$.

It remains to be shown that $h$ is $4$-\efx towards agent $\ell$ at time $(t+1)$. By claim (i) of the induction hypothesis, $h$ does not participate in swaps during times $[t'+1, t]$. Thus $\x^{t+1}_h = \x^{t'+1}_h$. Then, \cref{lem:after-swap-const} implies that $d_h(\x^{t'+1}_h) \le 4\cdot d_h(\z_h)$. Using the MPB condition for the \fpo allocation $(\z,\q)$, we get that $d_h(\z_h) = \alpha_h \q(\z_h)$. Using $\q(\z_h) = q_{t'}$, we conclude that $d_h(\x^{t+1}_h) \le 4\cdot \alpha_h q_{t'}$. With this, the following chain of inequalities shows that $h$ remains $4$-\efx towards $\ell$ at $(t+1)$.

\begin{equation*}
\begin{aligned}
d_h(\x^{t+1}_h) &\le 4\cdot \alpha_h q_{t'} &\text{(as argued above)} \\
&\le 4\cdot \alpha_h q_t &\text{(using claim (iv) of the induction hypothesis)} \\
&= 4\cdot \alpha_h \q(\z_i), & \text{(using $\q(\z_i) = q_t$)} \\
&\le 4\cdot d_h(\z_i), & \text{(using the MPB condition)} \\
&\le 4\cdot d_h(\x^{t+1}_\ell). & \text{(since $\x^{t+1}_\ell = \z_i$)} \\
\end{aligned}
\end{equation*}
This proves that claim (ii) holds at time $(t+1)$.

\item[(iii)] We prove that for an agent $h\in N_H^1$ who is not $4$-\efx at time $(t+1)$, it holds that $\q(\z_h) \ge q_{t}$. Clearly, $h \neq i$. Moreover, claim (ii) proved above shows that $h$ has not participated in a swap at any time step $t'<t+1$. Thus, $\x^{t+1}_h = S_h \cup \z_h$. Thus we have:
\begin{equation}\label{eq:claim3}
\begin{aligned}
d_h(\x^{t+1}_h) &= d_h(S_h) + d_h(\z_h) &\text{(using $\x^{t+1}_h = S_h \cup \z_h$)} \\
&= \p(S_h) + d_h(\z_h) &\text{(since $d_h(S_h) = \p(S_h)$)} \\ 
&\le \frac{3}{2} + d_h(\z_h) &\text{(using $\p(S_h) \le \frac{3}{2}$)} \\
&\leq 4 \cdot d_h(\z_h) &\text{(using $d_h(\z_h) > \frac{1}{2}$)} \\
&= 4 \cdot \alpha_h \q(\z_h). &\text{(using the MPB condition for $(\z,\q)$)}
\end{aligned}
\end{equation}

If $h$ was not $4$-\efx at $t$, then $\q(\z_h) \ge \q(\z_i)$, since \cref{alg:const-efx} chose to perform a swap involving $i$ instead of $h$. Thus $\q(\z_h) \ge \q(\z_i) = q_t$ in this case. 

On the other hand, suppose $h$ was $4$-\efx at $t$. Since $h$ is not $4$-\efx at $(t+1)$, $h$ became envious at $(t+1)$ due to the $(i,\ell)$ swap at $t$. Since $\x^{t+1}_i \supseteq \x^t_{\ell}$, $h$ does not $4$-\efx envy $i$ at $(t+1)$. This implies that $h$ $4$-\efx envies $\ell$ at time $(t+1)$. 

Thus, 
\begin{equation}
\begin{aligned}
4 \cdot \alpha_h \q(\z_h)
&\geq d_h(\x_h^{t+1}) &\text{(using \cref{eq:claim3})} \\
&> 4 \cdot d_h(\x^{t+1}_\ell) &\text{(since $h$ $4$-\efx envies $\ell$)} \\
&=  4 \cdot d_h(\z_i) &\text{(using claim (i))} \\
&= 4 \cdot \alpha_h q_t. &\text{(using the MPB condition of $(\z, \q)$)}
\end{aligned}
\end{equation}
Dividing each side by $4 \cdot \alpha_h$, we obtain $\q(\z_h) \ge q_t$, as claimed.

\item[(iv)] Consider a swap $(h,k)$ taking place at $(t+1)$. Since $h$ is not $4$-\efx at $(t+1)$, we have that $\q(\z_h) \ge q_t$ by claim (iii) proved above. Thus $q_{t+1} := \q(\z_h) \ge q_t$. With claim (iv) of the induction hypothesis at $t$, we obtain $q_{t+1}\ge q_t \ge \cdots \ge q_1\ge q_0$, as desired. \qedhere
\end{enumerate}
\end{proof}

We can now prove the main result of this section.
\thmConstEFX*
\begin{proof}
\cref{alg:const-efx} runs as long as there is an agent $i\in N_H^1$ who is not $4$-\efx. \cref{lem:invariants-const} shows that once an agent $i\in N_H^1$ participates in an $(i,\ell)$ swap, she remains $4$-\efx in the subsequent run of the algorithm. Thus, there can only be $n$ swap steps before the algorithm terminates. 

We argue that the resulting allocation is $4$-\efx. If an agent $i\in N_H^1$ is not $4$-\efx, then the algorithm would not have terminated. Thus, all agents in $N_H^1$ are $4$-\efx upon termination of the algorithm. Finally, \cref{lem:invariants-nh2} shows that the agents in $N_H^2$ are $4$-\efx and agents in $N_L$ are $3$-\efx throughout the run of the algorithm. In conclusion, given an ER equilibrium for a chore allocation instance with $m\ge 2n$, \cref{alg:const-efx} returns a $4$-\efx allocation in polynomial time.
\end{proof}

\section{Approximate-EFX and PO for Bivalued Instances}\label{sec:bivalued}

We now turn to the problem of computing (approximately-)\efx and \po allocations for \textit{bivalued} instances. Recall that in a bivalued instance $(N,M,D)$ there exist $a,b\in \R_{> 0}$ s.t. $d_{ij} \in \{a,b\}$.  Note that we can re-scale the disutilities so that $d_{ij} \in \{1,k\}$, where $k>1$. We refer to such an instance as a $\{1,k\}$-bivalued instance. The main result of this section is that:

\thmefxpomain*

To prove \cref{thm:efxpo-main}, we design and analyze two algorithms which separately handle the cases of $m > 2n$ and $m\le 2n$: \cref{alg:bivalued}, which computes a $3$-\efx and \po allocation when $m > 2n$; and \cref{alg:bivalued-small} which computes an \efx and \po allocation when $m\le 2n$. Both algorithms begin with initial allocations with certain desirable properties and perform subsequent chore transfers to achieve (approximate-)\efx and \po. \cref{alg:bivalued} begins with the $2$-\eftwo and \po allocation obtained by rounding an ER equilibria using \cref{alg:rounding}, while \cref{alg:bivalued-small} begins with the balanced allocation obtained using \cref{alg:balanced}. Before discussing our algorithms, we note that the bivalued nature of the instance allows us to prove some important properties of any competitive equilibrium $(\x,\p)$.
\begin{restatable}{lemma}{lemBivaluedPayments}\label{lem:bivalued-payments}
Let $(\x,\p)$ be a CE of a $\{1,k\}$-bivalued instance with $\rho = \min_j p_j$. Then:
\begin{itemize}
\item[(i)] For every $j\in M$, $\rho \le p_j \le \rho k$. 
\item[(ii)] Let $j\in \x_i$ be s.t. $p_j \in (\rho, \rho k)$. Then for all $j'\in \x_i$, $p_{j'} = p_j$.
\end{itemize}
\end{restatable}
\begin{proof}
For any $j\in M$, $p_j \ge \rho$ follows from the definition of $\rho$. Suppose for some $j_1\in M$, $p_{j_1} > \rho k$. Let $j_0\in \x_\ell$ be such that $p_{j_0} = \rho$. Then the MPB condition for agent $\ell$ implies that $\frac{d_{\ell{j_0}}}{p_{j_0}} \le \frac{d_{\ell j_1}}{p_{j_1}}$. This implies $\frac{d_{\ell {j_0}}}{d_{\ell {j_1}}} \le \frac{p_{j_0}}{p_{j_1}} < \frac{1}{k}$. However this is a contradiction since $d_{ij} \in \{1,k\}$ for all $i\in N,j\in M$. This proves (i).

For (ii), suppose $\exists j, j'\in \x_i$ s.t. $p_j \in (\rho, \rho k)$ and $p_{j'} \neq p_j$. Then the MPB condition for $i$ implies that $\frac{d_{ij}}{p_j} = \frac{d_{ij'}}{p_{j'}}$, implying that $\frac{p_j}{p_{j'}} = \frac{d_{ij}}{d_{i{j'}}}$. Since $d_{ij}, d_{ij'} \in \{1,k\}$, we know $\frac{p_j}{p_{j'}} \in \{1/k, 1, k\}$. Since $p_j\neq p_{j'}$, we have $\frac{p_j}{p_{j'}} \in \{1/k , k\}$. Thus $p_{j'} = k p_j$ or $p_{j'} = p_j/k$. Since $p_j\in (\rho, \rho k)$ and $k>1$, this implies either $p_{j'} > k\rho$ or $p_{j'} < \rho$, both of which contradict (i).
\end{proof}

\subsection{3-EFX and PO for $m > 2n$}\label{sec:bivalued-large}
We first present \cref{alg:bivalued}: a polynomial-time algorithm that computes a $3$-\efx and \fpo allocation for a bivalued instance with $m > 2n$, given its ER equilibrium $(\y,\p)$ as input. \cref{alg:bivalued} first rounds $(\y,\p)$ using \cref{alg:rounding} with the chore earning limit set as $\beta = \frac{1}{2}$. The resulting allocation $\x^0$ is already a good starting point: it is \fpo, and Lemmas~\ref{lem:earning-ub} and \ref{lem:earning-lb} with $\beta=\frac{1}{2}$ show its fairness properties.
\begin{lemma}\label{lem:bivalued-initial}
The allocation $(\x^0,\p)$ is \fpo and satisfies:
\begin{itemize}
\item[(i)] For all $i\in N$, $\p(\x^0_i) \ge \frac{1}{2}$.
\item[(ii)] For all $h\in N$, either $\p_{-1}(\x^0_h) \le 1$, or $|\x^0_h\cap \{j:p_j>\frac{1}{2}\}| = 2$ and $\p_{-2}(\x^0_h) \le \frac{1}{2}$.
\end{itemize}
\end{lemma}

Let $\rho = \min_j p_j$ be the minimum chore payment.
\begin{restatable}{lemma}{lemmasmallp}\label{lem:small-p}
The minimum chore payment satisfies $\rho < \frac{1}{2}$.
\end{restatable}
\begin{proof}
Let $q_j = \min\{p_j, \frac{1}{2}\}$ be the earning from chore $j\in M$. If $\rho \ge \frac{1}{2}$, then $p_j \ge \frac{1}{2}$ for all $j$, implying that $q_j = \frac{1}{2}$. Hence, the total earning from chores is $\sum_j q_j = \frac{m}{2}$. Since $\sum_j q_j = \sum_i e_i = n$, we obtain that $n = \frac{m}{2}$, which contradicts our assumption that $m > 2n$.
\end{proof}

The next lemma shows that if the largest chore payment is small, $\x^0$ is already fair. 
\begin{restatable}{lemma}{lemsmallpk}\label{lem:small-pk}
If $\rho k \le \frac{1}{2}$, then $\x^0$ is $3$-\textup{EF}.
\end{restatable}
\begin{proof}
If $\rho k \le \frac{1}{2}$, then $p_j \le \frac{1}{2}$ for all $j\in M$ by \cref{lem:bivalued-payments}. \cref{lem:bivalued-initial} then implies that:
\begin{enumerate}[label=(\roman*)]
\item For all $i\in N$, $\p(\x^0_i) \ge \frac{1}{2}$.
\item For all $h\in N$, $\p_{-1}(\x^0_h) \le 1$, or $\p_{-2}(\x^0_h)\le \frac{1}{2}$. Thus $\p(\x^0_h) \le \max\{1+\frac{1}{2}, \frac{1}{2}+2\cdot\frac{1}{2}\} = \frac{3}{2}$.
\end{enumerate}
We therefore have that $(\x^0,\p)$ is $3$-\s{EF}, since for any $i,h\in N$, we have $\p(\x^0_h) \le \frac{3}{2} \le 3\cdot \p(\x^0_i)$.
\end{proof}

Thus, \cref{alg:bivalued} simply returns $(\x^0,\p)$ if $\rho k \le \frac{1}{2}$. Hence, we assume $\rho k > \frac{1}{2}$ in the subsequent discussion. Note that $\rho < \frac{1}{2}$ by \cref{lem:small-p}.

\begin{definition}\label{def:classification}(Classification of Chores and Agents in \cref{alg:bivalued}) \normalfont
Chores are categorized as:
\begin{itemize}
\item $L = \{j\in M: p_j = \rho\}$, i.e., low paying or $L$-chores. Note $p_j = \rho < \frac{1}{2}$ for all $j\in L$.
\item $H = \{j\in M: p_j = \rho k\}$, i.e., high paying or $H$-chores. Note $p_j = \rho k > \frac{1}{2}$ for all $j\in H$.
\item $M' = M\setminus(L\cup H) = \{j\in M: p_j \in (\rho, \rho k)\}$, or $M'$-chores. 
\end{itemize}
\cref{lem:bivalued-payments} (ii) shows that agents can either be assigned chores from $M'$ or from $L$ and $H$, but not both. With this observation, we classify agents into four categories given an allocation $(\x,\p)$:
\begin{itemize}
\item $N_L = \{ i\in N: \x_i \subseteq L\}$, i.e., agents who are only assigned $L$-chores.
\item $N_H^1 = \{i\in N: |\x_i \cap H| = 1\}$, i.e., agents who are assigned exactly one $H$-chore.
\item $N_H^2 = \{i\in N: |\x_i \cap H| = 2\}$, i.e., agents who are assigned exactly two $H$-chores.
\item $N_0 = \{i\in N: \x_i \subseteq M'\}$, i.e., agents who are only assigned $M'$-chores.
\end{itemize}
\end{definition}
Let $N_H = N_H^1 \cup N_H^2$. We begin by exploring the source of \efx-envy in $\x^0$. We prove a general lemma concerning the \efx-envy of agents in $N_0$.

\begin{algorithm}[!t]
\caption{$3$-\efx + \po for bivalued instances with $m > 2n$}\label{alg:bivalued}
\textbf{Input:} $\{1,k\}$-bivalued instance with $m > 2n$, its ER equilibrium $(\y,\p)$ with $\beta = \frac{1}{2}$\\ 
\textbf{Output:} An integral allocation $\x$   
\begin{algorithmic}[1]
\State $(\x,\p) \gets$ Run \cref{alg:rounding} with $(\y,\p)$
\State $\rho \gets \min_j p_j$
\If{$\rho k \le \frac{1}{2}$} \Return $\x$ \Comment{$\x$ is $3$-\efx by \cref{lem:small-pk}}
\EndIf
\State $L = \{j \in M: p_j = \rho\}$, $H = \{j \in M: p_j = \rho k\}$ \Comment{Low, High paying chores}
\State Classify agents as $N_L, N_H^1, N_H^2, N_0$ (See \cref{def:classification})
\Statex \textit{--- Phase 1: Address $N_H^2$ agents ---}
\While{$\exists i\in N_H^2$ not 3-\efx}
\State $\ell \gets$ agent 3-\efx-envied by $i$ \Comment{\cref{lem:efx-envy-mpb-edges} shows $\ell\in N_L$}
\If{$\p(\x_\ell) > 1$}
$S\gets j_1$ for some $j_1\in\x_\ell$ \Else{} $S\gets \emptyset$
\EndIf
\State $j \in \x_i \cap H$
\State $\x_\ell \gets \x_\ell \setminus S \cup \{j\}$
\State $\x_i \gets \x_i \cup S \setminus \{j\}$
\State $N_H^1 \gets N_H^1 \cup \{i,\ell\}$, $N_H^2 \gets N_H^2\setminus \{i\}$, $N_L \gets N_L\setminus \{\ell\}$
\EndWhile
\Statex \textit{--- Phase 2: Address $N_H^1$ agents ---}
\While{$\exists i\in N_H^1$ not 3-\efx}
\State $\ell \gets \arg\min\{\p(\x_h) : h\in N \text{ s.t. $i$ 3-\efx envies $h$}\}$ \Comment{\cref{lem:efx-envy-mpb-edges} shows $\ell\in N_L$}
\State $j \in \x_i \cap H$
\State $\x_i \gets \x_i \cup \x_\ell \setminus \{j\}$
\State $\x_\ell \gets \{j\}$
\State $N_H^1 \gets N_H^1 \cup \{\ell\}\setminus\{i\}$, $N_L \gets N_L \cup \{i\}\setminus\{\ell\}$
\EndWhile
\State \Return $\x$
\end{algorithmic}
\end{algorithm}

\begin{restatable}{lemma}{lemNZeroEFXenvy}\label{lem:n0-efx-envy}
Consider an allocation $(\x,\p)$ s.t. $\x_i = \x_i^0$ for all $i\in N_0$ and $\p(\x_h)\ge \frac{1}{2}$ for all $h\in N$. Then $\x$ is $2$-\efx for any $i\in N_0$.
\end{restatable}
\begin{proof}
Consider an agent $i\in N_0$. As per \cref{lem:bivalued-initial}, we consider two cases regarding $\x^0_i$:
\begin{itemize}
\item[(i)] $\p_{-1}(\x_i^0) \le 1$. Since \cref{lem:invariants} implies that $\p(\x_h) \ge \frac{1}{2}$ for any $h\in N$, we obtain that $\p_{-1}(\x_i^0) \le 2\cdot\p(\x_h)$. For $i \in N_0$ we have that $\p_{-X}(\x_i^0) = \p_{-1}(\x_i^0)$, which gives us that $\p_{-X}(\x_i^0) \le 2\cdot\p(\x_h)$. Thus, $i$ is $2$-\pefx and hence $2$-\efx towards any $h\in N$ by \cref{lem:pEF1impliesEF1}.
\item[(ii)] $|\x^0_i \cap \{j: p_j > \frac{1}{2}\}| = 2$ and $\p_{-2}(\x^0_i) \le \frac{1}{2}$. Let $\x^0_i \cap \{j: p_j > \frac{1}{2}\} = \{j_1,j_2\}$. By \cref{lem:bivalued-payments}, all chores in $\x^0_i$ have the same payment $\rho'\in(\rho,\rho k)$. Hence $p_{j_1} = p_{j_2} = \rho' > \frac{1}{2}$. Thus $\p_{-2}(\x^0_i) \le \frac{1}{2}$ implies that $\x^0_i \setminus \{j: p_j > \frac{1}{2}\} = \emptyset$, i.e., $\x^0_i = \{j_1, j_2\}$.

We claim now that $d_{ij_1} = d_{ij_2} = 1$. Without loss of generality, suppose otherwise that $d_{ij_1} = k$. By the definition of $\rho$ there exists a chore $j$ such that $p_j = \rho$. Then, since $p_{j_1} < \rho k$, we have that:
\[
\frac{d_{ij}}{p_j} \leq \frac{k}{\rho} = \frac{d_{ij_1}}{\rho} < \frac{d_{ij_1}}{\rho k} < \frac{d_{ij_1}}{p_{j_1}}.
\]
This implies that $j_1$ is not MPB for $i$, a contradiction, so it must be that $d_{ij_1} = 1$. Since $\p(\x_h)\ge \frac{1}{2}$ for all $h\in N$, all bundles are non-empty. Thus, for all $h \in N$, we have $\max_{j'\in\x_i} d_i(\x_i\setminus\{j'\}) = 1 \leq d_i(\x_h)$, showing that $i$ is in fact EFX towards all agents. \qedhere
\end{itemize}
\end{proof}

Since $\p(\x^0_h) \ge \frac{1}{2}$ for all $h\in N$ by \cref{lem:bivalued-initial}, the above lemma shows that $\x^0$ is $2$-\efx for agents in $N_0$. Next, we show that $\x^0$ is also $2$-\efx for agents in $N_L$.

\begin{restatable}{lemma}{lemNLinit}\label{lem:nl-init}
$\x^0$ is $2$-\efx for agents in $N_L$. 
\end{restatable}
\begin{proof}
Consider an agent $i\in N_L$. Since $\x^0_i\subseteq L$, $\p_{-X}(\x^0_i) = \p_{-1}(\x^0_i)$. \cref{lem:bivalued-initial} implies that $\p_{-1}(\x^0_i) \le 1 \le 2\cdot \p(\x^0_h)$ for any $h\in N$. This shows $i$ is $2$-\efx towards any $h\in N$. 
\end{proof}

If $H = \emptyset$, $N = N_0 \cup N_L$. Thus $\x^0$ is $2$-\efx, and \cref{alg:bivalued} will simply return $\x^0$. Hence, we assume $H\neq \emptyset$ in the subsequent discussion. Lemmas~\ref{lem:n0-efx-envy} and \ref{lem:nl-init} show that $\x^0$ is $2$-\efx for agents in $N_0 \cup N_L$. Hence if $\x^0$ is not $3$-\efx, some agent in $N_H$ must $3$-\efx-envy another agent. Intuitively, an agent $i\in N_H$ $2$-\efx-envies another agent $\ell$ since $i$ has one or two high paying $H$-chores in addition to some low paying $L$-chores. \cref{alg:bivalued} addresses the \efx-envy of these agents by swapping some chores between agents $i$ and $\ell$, and does so in two phases.

In Phase 1, \cref{alg:bivalued} addresses agents in $N_H^2$. An agent $i \in N_H^2$ has two $H$-chores, and earns at most $\frac{1}{2}$ from $L$-chores. We show in \cref{lem:efx-envy-mpb-edges} that if $i$ $3$-\efx-envies an agent $\ell$, we must have $\ell\in N_L$. We then transfer one $H$-chore from $i$ to $\ell$, and if needed, transfer a single $L$-chore from $\ell$ to $i$ so that both agents earn at most $1$ from their $L$-chores. We show that such a swap preserves that the allocation is MPB. After the swap, both $i$ and $\ell$ are added to $N_H^1$ and removed from $N_H^2$ and $N_L$ respectively. This implies that Phase 1 terminates after at most $n/2$ swaps, after which the allocation is $3$-\efx for all agents in $N_H^2$. 

In Phase 2, \cref{alg:bivalued} addresses agents in $N_H^1$. An agent $i \in N_H^1$ has one $H$-chore, and earns at most $1$ from $L$-chores. Once again, \cref{lem:efx-envy-mpb-edges} shows that if $i$ $3$-\efx-envies an agent $\ell$, then $\ell\in N_L$. We then transfer the $H$-chore from $i$ to $\ell$, and transfer all the chores of $\ell$ to $i$. As before, we argue that such a swap preserves that the allocation is MPB. After the swap, $i$ gets added to $N_L$ and removed from $N_H^1$, while $\ell$ is added to $N_H^1$ and removed from $N_L$. Since $\ell$ is now assigned a single $H$-chore, $\ell$ does not \efx-envy any agent. This implies that Phase 2 terminates after at most $n$ swaps since the number of agents in $N_H^1$ who are not $2$-\efx strictly decreases. The resulting allocation is $3$-\efx for all agents in $N_H^1$.

Lastly, we show that throughout the algorithm, agents in $N_0$ are $3$-\efx towards all other agents (\cref{lem:n0-efx-envy}), agents in $N_L$ are $3$-\efx towards all other agents (\cref{lem:nl-envy}), and those in $N_H$ are $3$-\efx towards agents in $N_H \cup N_0$ (\cref{lem:efx-envy-mpb-edges}). Since the algorithm addresses $3$-\efx-envy from agents in $N_H$ towards those in $N_L$ in at most $3n/2$ swaps, it terminates with a $3$-\efx and \fpo allocation. 

We now formally prove the above claims. We begin by recording a lemma regarding the MPB ratio $\alpha_i$ of an agent $i\in N_L\cup N_H$. 
\begin{restatable}{lemma}{lemMPBRatios}\label{lem:mpb-ratios}
Assume $H\neq \emptyset$. Then:
\begin{itemize}
\item[(i)] For all $i\in N_L$, $\alpha_i = 1/\rho$. Moreover for every $j\in H$, $d_{ij} = k$ and $j\in \mpb_i$.
\item[(ii)] For all $i\in N_H$, $\alpha_i \in \{\frac{1}{\rho}, \frac{1}{\rho k}\}$.
\item[(iii)] For all $i\in N_H$, if $\x_i \setminus H \neq \emptyset$ then $\alpha_i = 1/\rho$.
\end{itemize}
\end{restatable}
\begin{proof}
Let $j_0 \in H$ with $p_{j_0} = \rho k$. For (i), consider $i\in N_L$, and let $j\in\x_i$. Since $j\in L$, $p_j = \rho$. The MPB condition for $i$ implies $\frac{d_{ij}}{p_j} \le \frac{d_{ij_0}}{p_{j_0}}$. This gives $kd_{ij} \le d_{ij_0}$. Since $d_{ij},d_{ij_0}\in \{1,k\}$, the above inequality must be an equality and $d_{ij} = 1$ and $d_{ij_0} = k$. Thus $\alpha_i = 1/\rho$ for $i\in N_L$. Now consider any $j'\in H$. The MPB condition for $i$ implies $\alpha_i \le \frac{d_{ij'}}{p_{j'}}$. This implies $d_{ij'} \ge k$. Since $d_{ij'}\in\{1,k\}$, we have $d_{ij'}=k$ and $j'\in\mpb_i$.

For (ii), let $i\in N_H$ and $j\in\x_i\cap H$. Then $\alpha_i = \frac{d_{ij}}{p_j} \in \{\frac{1}{\rho}, \frac{1}{\rho k}\}$, since $d_{ij}\in\{1,k\}$ and $p_j = \rho k$. 

For (iii), consider $i\in N_H$ with $j_1 \in \x_i\setminus H$ and $j_2\in \x_i \cap H$. The MPB condition for $i$ implies $\frac{d_{ij_1}}{p_{j_1}} = \frac{d_{ij_2}}{p_{j_2}}$, which gives $d_{ij_1} = 1$ and $d_{ij_2} = k$. Thus $\alpha_i = 1/\rho$.
\end{proof}

We next show that \cref{alg:bivalued} maintains the following invariants.
\begin{restatable}{lemma}{lemInvariants}\label{lem:invariants} (Invariants of Alg.\ref{alg:bivalued})
Let $(\x,\p)$ be an allocation in the run of \cref{alg:bivalued}. Then:
\begin{enumerate}[label=(\roman*)]
\item $(\x,\p)$ is an MPB allocation.
\item For all $i\in N$, $\p(\x_i) \ge \frac{1}{2}$.
\item For all $i\in N_L$, $\p_{-1}(\x_i) \le 1$ during Phase 1.
\item For all $i\in N_L$, $\p(\x_i) \le \frac{4}{3} + \frac{\rho k}{3}$.
\item For all $i\in N_H^1$, $\p_{-1}(\x_i) \le 1$.
\item For all $i\in N_H^2$, $\p_{-2}(\x_i) \le \frac{1}{2}$.
\end{enumerate}
\end{restatable}
\begin{proof}
We prove this using Lemmas~\ref{lem:efx-envy-mpb-edges}, \ref{lem:phase1-inv} and \ref{lem:phase2-inv} below.    
\end{proof}

We first prove that some conditions must hold if an agent in $N_H$ $3$-\efx-envies another agent.
\begin{restatable}{lemma}{lemEFXenvyMPBedges}\label{lem:efx-envy-mpb-edges}
Consider an allocation $(\x,\p)$ satisfying the invariants of \cref{lem:invariants}. Then if $i\in N_H$ $3$-\efx-envies $\ell$, then $\alpha_i = 1/\rho$, $\ell\in N_L$, and $\x_\ell \subseteq \mpb_i$.
\end{restatable}
\begin{proof}
Consider $i\in N_H$ who $3$-\efx-envies $\ell\in N$. We know from \cref{lem:mpb-ratios} that $\alpha_i \in\{\frac{1}{\rho}, \frac{1}{\rho k}\}$. Suppose $\alpha_i = \frac{1}{\rho k}$. Then $d_{ij} = 1$ for all $j\in \x_i$. By the contrapositive of \cref{lem:mpb-ratios} (iii), we get $\x_i \subseteq H$. Since $|\x_i\cap H|\le 2$, we get $|\x_i| \le 2$. Thus, $\max_{j\in\x_i} d_i(\x_i\setminus\{j\}) \le 1 \le d_i(\x_\ell)$, since $\x_\ell\neq\emptyset$ and the instance is bivalued. Thus, $i$ is \efx towards $\ell$ if $\alpha_i = \frac{1}{\rho k}$, which implies that $\alpha_i = 1/\rho$. 

If $\exists j\in \x_\ell$ such that $d_{ij} = k$, then observe that:
\begin{align*}
d_i(\x_i)& = \alpha_i\cdot\p(\x_i) \tag{using the MPB condition} \\
&\le \frac{1}{\rho}\cdot \max\{1+\rho k, \frac{1}{2} + 2\rho k\} \tag{using $\alpha_i = \frac{1}{\rho}$ and invariants (v) and (vi)} \\
&= \max\bigg\{\frac{1}{\rho} + k, \frac{1}{2\rho} + 2k\bigg\} \\
&< 3k \tag{using $\rho k > \frac{1}{2}$} \\
&\le 3 d_i(\x_\ell). \tag{since $j\in \x_\ell$}
\end{align*}
This shows that $i$ is $3$-\efx towards $\ell$. Thus it must be that for all $j\in \x_\ell$, $d_{ij} = 1$. The MPB condition for $i$ implies that $\alpha_i \le d_{ij}/p_j$, showing that $p_j \le \rho$. \cref{lem:bivalued-payments} implies that $p_j = \rho$ for all $j\in \x_\ell$. Thus $\ell\in N_L$. Moreover, for any $j\in \x_\ell$, $\alpha_i = d_{ij}/p_j$, and hence $\x_\ell\subseteq\mpb_i$.
\end{proof}

The next two lemmas establish the invariants claimed by \cref{lem:invariants}.
\begin{restatable}{lemma}{lemPhase1Inv}\label{lem:phase1-inv}
The invariants of \cref{lem:invariants} are maintained during Phase 1 of \cref{alg:bivalued}.
\end{restatable}
\begin{proof}
We prove the statement inductively. We first show that the invariants hold at $(\x^0,\p)$. Invariants (i), (ii), (iii), (v) and (vi) follow from \cref{lem:bivalued-initial}. For invariant (iv), note that for any $i\in N_L$, we have $\p_{-1}(\x^0_i) \le 1$. By using $\rho < \frac{1}{2}$ and $\rho k > \frac{1}{2}$, we obtain:
\[
\p(\x^0_i) \le 1 + \rho < \frac{3}{2} = \frac{4}{3} + \frac{1}{6} < \frac{4}{3} + \frac{\rho k}{3},
\]
proving invariant (iv).

Suppose the invariants hold at an allocation $(\x,\p)$ during Phase 1. Consider a Phase 1 swap involving agents $i\in N_H^2$ and $\ell\in N$. Given that \cref{alg:bivalued} performed the swap, $i$ must $3$-\efx-envy $\ell$. \cref{lem:efx-envy-mpb-edges} implies that $\ell \in N_L$ and hence $\x_\ell \in L$. As per \cref{alg:bivalued}, if $\p(\x_\ell) > 1$, then $S = \{j_1\}$ for some $j_1\in\x_\ell$, otherwise $S = \emptyset$. Let $j\in \x_i\cap H$.

Let $\x'$ be the resulting allocation. Thus $\x'_i = \x_i\setminus\{j\}\cup S$, $\x'_\ell = \x_\ell \setminus S\cup \{j\}$, and $\x'_h = \x_h$ for all $h\notin\{i,\ell\}$. We show that the invariants hold at $(\x',\p)$. Since a Phase 1 step removes agents $i$ and $\ell$ from $N_H^2$ and $N_L$ respectively, invariants  (iii), (iv), (vi) continue to hold. For the rest, observe:
\begin{enumerate}
\item[(i)] $(\x',\p)$ is on MPB. This is because \cref{lem:efx-envy-mpb-edges} implies $S\subseteq \x_\ell \subseteq \mpb_i$, showing $\x'_i\subseteq \mpb_i$. Since $\ell\in N_L$ at $(\x,\p)$ and $j\in H$, \cref{lem:mpb-ratios} shows $j\in \mpb_\ell$ and hence $\x'_\ell \subseteq \mpb_\ell$.
\item[(ii)] Since $|\x'_i\cap H| = |\x'_\ell\cap H| = 1$, we have $\p(\x'_i) \ge \rho k$ and $\p(\x'_\ell)\ge \rho k$. Invariant (ii) follows by noting that $\rho k > \frac{1}{2}$.
\item[(v)] For agent $i$, note that $\p_{-1}(\x'_i) \le \p(\x_i\setminus \{j\} \cup S) = \p_{-1}(\x_i\setminus\{j\}) + \p(S)$. Invariant (iii) implies $\p_{-1}(\x_i\setminus\{j\}) = \p_{-2}(\x_i) \le \frac{1}{2}$, and $\p(S)\le \frac{1}{2}$ by construction. Hence $\p_{-1}(\x'_i) \le 1$.

For agent $\ell$, note that $\p_{-1}(\x'_\ell) = \p(\x_\ell\setminus S) = \p(\x_\ell)-\p(S)$. If $\p(\x_\ell) \le 1$, then $S=\emptyset$, implying that $\p_{-1}(\x'_\ell) \le 1$. On the other hand suppose $\p(\x_\ell) > 1$. Since invariant (v) holds at $(\x,\p)$, we have $\p_{-1}(\x_\ell)\le 1$, which gives $\p(\x_\ell)\le 1+\rho$. With $\p(S) = p_{j_1} = \rho$, we obtain $\p_{-1}(\x'_\ell) = \p(\x_\ell)-\p(S) \le 1$.
\end{enumerate}
Since the swap does not affect any $h\notin\{i,\ell\}$, the invariants continue to hold for $h$ after the swap. By induction, we have shown that the invariants of \cref{lem:invariants} hold after any Phase 1 swap.
\end{proof}

\begin{restatable}{lemma}{lemPhase2Inv}\label{lem:phase2-inv}
The invariants of \cref{lem:invariants} are maintained during Phase 2 of \cref{alg:bivalued}. Moreover, agents in $N_H^2$ remain $3$-\efx towards other agents.
\end{restatable}
\begin{proof}
We prove the statement inductively. \cref{lem:phase1-inv} shows the invariants hold at the end of Phase 1. Suppose the invariants hold at an allocation $(\x,\p)$ during Phase 2. Consider a Phase 2 swap involving agents $i\in N_H^1$ and $\ell\in N$. Given that \cref{alg:bivalued} performed the swap, $i$ must $3$-\efx-envy $\ell$. \cref{lem:efx-envy-mpb-edges} implies that $\ell \in N_L$ and hence $\x_\ell \in L$. Let $j\in \x_i\cap H$.

Let $\x'$ be the resulting allocation. Thus $\x'_i = \x_i\setminus\{j\}\cup \x_\ell$, $\x'_\ell = \{j\}$, and $\x'_h = \x_h$ for all $h\notin\{i,\ell\}$. We now show that the invariants hold at $(\x',\p)$. Since we are in Phase 2, invariant (iii) does not apply, and since Phase 2 swaps do not alter the allocation of agents in $N_H^2$, invariant (vi) continues to hold. For the rest, observe:
\begin{enumerate}
\item[(i)] $(\x',\p)$ is on MPB. This is because \cref{lem:efx-envy-mpb-edges} implies $\x_\ell \subseteq \mpb_i$, showing $\x'_i\subseteq \mpb_i$. Since $\ell\in N_L$ at $(\x,\p)$ and $j\in H$, \cref{lem:mpb-ratios} shows $j\in \mpb_\ell$ and hence $\x'_\ell \subseteq \mpb_\ell$.
\item[(ii)] For agent $i$, $\p(\x'_i) \ge \p(\x_\ell) \ge \frac{1}{2}$, since invariant (ii) holds in $(\x,\p)$. For agent $\ell$, note that $\p(\x'_\ell) = p_j = \rho k > \frac{1}{2}$.

\item[(iv)] We want to show that $\p(\x'_i) \le \frac{4}{3} + \frac{\rho k}{3}$. To see this note that since $i$ $3$-\efx-envies $\ell$ in $\x$, $i$ must $3$-\s{pEF}-envy $\ell$ in $(\x,\p)$. Thus $\p(\x_i) > 3\cdot\p(\x_\ell)$. Now $\p_{-1}(\x_i) \le 1$ due to invariant (v), which shows $\p(\x_i) \le 1+\rho k$. We therefore obtain $\p(\x_\ell) < \frac{\p(\x_i)}{3} \leq \frac{1+\rho k}{3}$.

Now $\p(\x'_i) = \p(\x_i\setminus\{j\}) + \p(\x_\ell) \le 1 + \frac{1+\rho k}{3} = \frac{4}{3} + \frac{\rho k}{3}$, where we once again use $\p_{-1}(\x_i) = \p(\x_i\setminus\{j\}) = 1$. The invariant thus follows. 

\item[(v)] Note that $\ell\in N_H^1$ in $(\x',\p)$, and $\p_{-1}(\x_\ell) = 0 < 1$.
\end{enumerate}
The swap does not affect an agent $h\notin\{i,\ell\}$ and hence the invariants continue to hold for $h$ after the swap. By induction, we conclude that the invariants of \cref{lem:invariants} hold after any Phase 2 swap.

We now show that $i\in N_H^2$ cannot $3$-\efx-envy an agent $\ell\in N$. \cref{lem:efx-envy-mpb-edges} implies that $\ell\in N_L$ and hence $\x_\ell \subseteq L$. Let $\x^1$ be the allocation at the end of Phase 1. Note that the bundle $\x_\ell \subseteq L$ is obtained via a series of Phase 2 swaps initiated with some agent $\ell_1$ in $(\x^1,\p)$. Here, $\ell_1\in N_L$ at $(\x^1,\p)$. Thus $\x_\ell \supseteq \x^1_{\ell_1}$. Agent $i\in N_H^2$ did not $3$-\efx-envy $\ell_1$ in $\x^1$, otherwise \cref{alg:bivalued} would have performed a Phase 1 swap between agent $i$ and $\ell_1$. Since $\x_i = \x^1_i$ as \cref{alg:bivalued} does not alter allocation of agents in $N_H^2$ and $\x_\ell \supseteq \x^1_{\ell_1}$, $i$ will not $3$-\efx-envy $\ell$ in $\x$ either. Thus, all agents in $N_H^2$ continue to remain $3$-\efx during Phase 2.
\end{proof}

We require one final lemma showing that $N_L$ agents do not $3$-\efx-envy any other agent.
\begin{restatable}{lemma}{lemNLenvy}\label{lem:nl-envy}
At any allocation $(\x,\p)$ in the run of \cref{alg:bivalued}, $\x$ is $3$-\efx for every agent in $N_L$.
\end{restatable}
\begin{proof}
\cref{lem:nl-init} shows that the initial allocation $\x^0$ is $3$-\efx for agents in $N_L$. Let $\x$ be the earliest allocation in the run of \cref{alg:bivalued} in which an agent $i\in N_L$ $3$-\efx-envies another agent $h\in N$. Using $\alpha_i = 1/\rho$ from \cref{lem:mpb-ratios}, the bound on $\p(\x_i)$ from \cref{lem:invariants} (iv), and $\rho k > \frac{1}{2}$, we note:
\begin{equation}\label{eq:nl-envy}
d_i(\x_i) = \alpha_i\cdot\p(\x_i) \le \frac{1}{\rho}\cdot \bigg(\frac{4}{3}+\frac{\rho k}{3}\bigg) < 3k.
\end{equation}

Thus if $\exists j\in \x_h$ s.t. $d_{ij} = k$, then by \eqref{eq:nl-envy}, $d_i(\x_i) < 3k \le 3 d_i(\x_h)$, showing that $i$ does not $3$-\efx-envy $h$ in $\x$. Hence it must be that for all $j\in\x_h$, $d_{ij} = 1$. This also implies $\x_h \subseteq \mpb_i$, since $\alpha_i = 1/\rho = d_{ij}/p_j$ for any $j\in \x_h$. We now consider three cases based on the category of $h$.
\begin{itemize}[leftmargin=*]
\item $h\in N_0$. For $j\in \x_h$, the MPB condition of $i$ implies $\alpha_i \le d_{ij}/p_j$, implying $d_{ij} \geq p_j/\rho$. Since $h\in N_0$, we have $j\in M'$ and $p_j > \rho$. Thus $d_{ij} = k$ for $j\in\x_h$, which is a contradiction.
\item $h\in N_H$. By definition of $N_H$, $\exists j\in \x_h$ s.t. $j\in H$. Since $i\in N_L$, by \cref{lem:mpb-ratios} (i) we get $d_{ij} = k$, which is a contradiction.
\item $h\in N_L$. Since $\x^0$ is $3$-\efx for agents in $N_L$, and Phase 1 swaps only remove agents from $N_L$, it cannot be that $i$ starts $3$-\efx-envying $h\in N_L$ during Phase 1. Let $\x'$ be the allocation immediately preceding $\x$, from which \cref{alg:bivalued} performed a Phase 2 swap. It must be that in $\x'$, agent $i$ was in $N_H^1$ and was involved with a Phase 2 swap with another agent $\ell\in N_L$. Since $\x'_h = \x_h \subseteq \mpb_i$, we must have $\p(\x'_\ell) \le \p(\x'_h)$ by the choice of $\ell$ at $(\x',\p)$. Else, if $\p(\x'_\ell) > \p(\x'_h)$, then $d_i(\x'_h) = \p(\x'_h) < \p(\x'_\ell) \leq d_i(\x'_\ell)$ and \cref{alg:bivalued} would have chosen $h$ over $\ell$ for the swap in $\x'$.

Note that $\x_i = (\x'_i \setminus H) \cup \x'_\ell$. By \cref{lem:invariants} (v), we know that $\p_{-1}(\x'_i) = \p(\x'_i\setminus H) \le 1$. Thus:
\[ 
\p(\x_i) = \p(\x'_i\setminus H) + \p(\x'_\ell) \le 1 + \p(\x'_h) \le 3\p(\x_h),
\]
where the last inequality uses $\x_h = \x'_h$ and $\p(\x_h)\ge \frac{1}{2}$. Thus $i$ is actually $3$-\efx towards $h$.
\end{itemize}
Since these cases are exhaustive, we conclude that it is not possible for an agent $i\in N_L$ to $3$-\efx-envy any other agent during the course of \cref{alg:bivalued}.
\end{proof}

We are now in a position to summarize and conclude our analysis of \cref{alg:bivalued}. \begin{restatable}{theorem}{thmBivaluedER}\label{thm:bivalued-er}
Given an ER equilibrium of a bivalued instance with $m > 2n$ and chore earning limit $\beta = \frac{1}{2}$, \cref{alg:bivalued} returns a $3$-\efx and \fpo allocation in polynomial time.
\end{restatable}
\begin{proof}
Let $(\x^0,\p)$ be the initial allocation obtained by using \cref{alg:rounding} on the ER equilibrium, and let $\rho = \min_j p_j$ be the minimum payment. \cref{lem:small-pk} shows that $\x^0$ is $3$-\efx if $\rho k \le \frac{1}{2}$, hence we assume otherwise. Since any allocation $(\x,\p)$ during the course of \cref{alg:bivalued} satisfies invariant (ii) of \cref{lem:invariants}, \cref{lem:n0-efx-envy} implies that $\x$ is $3$-\efx for agents in $N_0$. 

\cref{lem:nl-envy} shows that any allocation $\x$ in the course of \cref{alg:bivalued} is $3$-\efx for agents in $N_L$. Any potential \efx-envy is therefore from some agent $i\in N_H$. \cref{lem:efx-envy-mpb-edges} shows that if $i\in N_H$ is not $3$-\efx towards $\ell$, then $\ell\in N_L$. If $i\in N_H^2$, $i$ participates in a Phase 1 swap with agent $\ell$, after which $i$ and $\ell$ get removed from $N_H^2$ and $N_L$ respectively. This implies that Phase 1 terminates after at most $n/2$ swaps, and the resulting allocation is $3$-\efx for all agents in $N_H^2$. If $i\in N_H^1$, $i$ participates in a Phase 2 swap with agent $\ell$, after which $\ell$ is added to $N_H^1$ and is assigned a single chore and $\ell$ does not have \efx-envy. This implies that Phase 2 terminates after at most $n$ swaps, since the number of agents in $N_H^1$ who are not $3$-\efx strictly decreases. The resulting allocation is $3$-\efx for all agents in $N_H^1$. \cref{lem:phase2-inv} also shows that Phase 2 swaps do not cause $N_H^2$ agents to start $3$-\efx-envying any agent in $N_L$. Thus the allocation on termination of \cref{alg:bivalued} is $3$-\efx. By invariant (i) of \cref{lem:invariants}, $\x$ is also \fpo. Since there are at most $3n/2$ swaps, \cref{alg:bivalued} terminates in polynomial time.
\end{proof}

\subsection{EFX and PO for $m\le 2n$}\label{sec:bivalued-small}
We design \cref{alg:bivalued-small} for bivalued instances with $m\le 2n$. \cref{alg:bivalued-small} begins with a balanced allocation computed using \cref{alg:balanced}, and essentially runs \cref{alg:bivalued}. Since the number of chores is limited, a careful analysis shows that the guarantee of the resulting allocation can be improved to \efx and \fpo. 
\begin{restatable}{theorem}{thmBivaluedSmall}\label{thm:bivalued-small}
Given a bivalued instance with $m \le 2n$, \cref{alg:bivalued-small} returns an \efx and \fpo allocation in polynomial time.
\end{restatable}
The main ideas of the analysis are similar to those presented in \cref{sec:bivalued-large}, and hence we defer this section in a self-contained \cref{app:bivalued-small}. Surprisingly, we show through \cref{ex:aEFX-2-ary} shows that if we slightly generalize the class to 2-ary instances, an $\alpha$-\efx and \fpo allocation need not exist for any constant $\alpha\ge 1$. 

\section{Existence of Earning-Restricted Equilibria}\label{sec:er}
We prove \cref{thm:er-existence} in this section.
\thmerexistence*

In what follows, we assume ER instances satisfy the feasible earning condition $\sum_i e_i \le \sum_j c_j$. We prove \cref{thm:er-existence} by designing a linear complementarity problem (LCP) whose solution corresponds to an ER equilibrium. We begin with some background on LCPs.

\subsection{Linear Complementarity Problems and Lemke's Scheme}\label{sec:lemke-prelim}
A Linear Complementary Problem (LCP) is a generalization of linear programming (LP) with complementary slackness conditions: given a matrix $A\in\R^{n\times n}$ and a vector $\b\in\R^n$, the problem is:
\begin{equation}\label{eq:lcp-def}
\text{LCP($A,\b$): Find } \y\ge \boldsymbol{0} \text{ such that } A\y \le \b, \text{ and } y_i\cdot(A\y - \b)_i = 0 \text{ for all } i\in [n].  
\end{equation}
We use the shorthand notation $(A\y)_i \le b_i \: \perp \: y_i$ to represent the constraints $A\y \le b_i$, $y_i \ge 0$, and $y_i \cdot (A\y-\b)_i$. If $\b \ge \boldsymbol{0}$, then $\y = \boldsymbol{0}$ is a trivial solution to the LCP. If $\b\not\ge \boldsymbol{0}$, then the LCP may not have a solution; indeed, LCPs are general enough to capture NP-hard problems \cite{cottlepang-lcpbook}.

\paragraph{Lemke's scheme.} 
Let $\mathcal{P} = \{\y\in \R^n: A\y\le \b, \y\ge\boldsymbol{0}\}$. We assume that the polyhedron $\mathcal{P}$ is non-degenerate, i.e., exactly $n - d$ constraints hold with equality on any $d$-dimensional face of $\mathcal{P}$. With this assumption, each solution to \eqref{eq:lcp-def} corresponds to a vertex of $\mathcal{P}$ since exactly $n$ equalities must be satisfied. Lemke's scheme finds such a vertex solution by working with an \textit{augmented} LCP which adds a scalar variable $z$ to LCP($A,\b$), resulting in the following program:
\begin{equation}\label{eq:lcp-aug}
\text{Augmented LCP($A,\b$): } z\ge 0; \text{ and } A\y - z\cdot\boldsymbol{1} \le \b, \text{ and } y_i\cdot((A\y-\b)_i - z) = 0 \text{ for all } i\in[n].
\end{equation}
Note that $(\y,z=0)$ is a solution to \eqref{eq:lcp-aug} iff $\y$ is solution to \eqref{eq:lcp-def}. Let 
$\mathcal{P'} = \{(\y,z) \in \R^{n+1} : A\y - z\cdot\boldsymbol{1} \le \b, \y \ge 0, z\ge 0\}$.  Assuming $\mathcal{P'}$ is non-degenerate, a solution to \eqref{eq:lcp-aug} still satisfies $n$ constraints of $\mathcal{P'}$ with equality. Since $\mathcal{P'}$ is $(n+1)$-dimensional, this means the set of solutions $S$ is a subset of the 1-skeleton of $\mathcal{P'}$,  i.e., edges (1-dimensional faces) and vertices (0-dimensional faces). Moreover, $\y$ is a solution of \eqref{eq:lcp-def} iff $(\y,0)$ is a vertex of $\mathcal{P'}$.

The set $S$ of solutions to the augmented LCP has some important structural properties. We say that the label $i$ is present at $(\y,z)\in \mathcal{P'}$ if $y_i = 0$ or $(A\y)_i -z=b_i$. Every solution in $S$ is \textit{fully labelled} where all the labels are present. A solution $s\in S$ contains \textit{double label} $i$ if $(A\y)_i -z=b_i$ for $i\in [n]$. Since there are only two ways to relax the double label while keeping all other labels, there are two edges of $S$ incident to $s$. The above observations imply that $S$ consists of paths and cycles. Clearly, any solution $s$ to \eqref{eq:lcp-aug} with $z = 0$ contains no double labels. Relaxing $z = 0$ gives the unique edge incident to $s$ at this vertex. We note that some of the edges in $S$ are unbounded. An unbounded edge of $S$ incident to vertex $(\y^*,z^*)$ with $z^* > 0$ is called a ray. Formally, a ray $R$ has the form $R= \{ (\y^*,z^*) + \alpha\cdot(\y',z') : \alpha\ge 0 \}$, where $(\y', z') \neq 0$ solves \eqref{eq:lcp-aug} with $\b=0$. The \textit{primary ray} is the ray $\{(\boldsymbol{0}, z) : z\ge |\min_i b_i|\}$, which contains solutions with $\y=0$ and $z$ sufficiently large to satisfy \eqref{eq:lcp-aug}. All other rays are called \textit{secondary rays}. 

Starting from the primary ray, Lemke’s scheme follows a path on the 1-skeleton of $\mathcal{P'}$ with a guarantee that it never revisits a vertex. If a vertex $s$ is non-degenerate, i.e., has a \textit{unique} double label, then Lemke's scheme \textit{pivots} by relaxing one of the two constraints and travelling along the edge of $\mathcal{P'}$ to the next vertex solution. Therefore, if the vertices are non-degenerate, Lemke's scheme eventually either reaches a vertex with $z = 0$ (which is a solution of the original LCP \eqref{eq:lcp-def}) or ends up on a secondary ray. In the latter case, the algorithm fails to find a solution; in fact, the problem may not have a solution. Note that it suffices to introduce $z$ in the $(A\y)_i\le b_i$ constraint only if $b_i < 0$, without changing the role of $z$.

\subsection{Basic LCP for ER Equilibrium}
We first capture ER equilibria in an instance $(N,M,D,e,c)$ via the following LCP with variables $\p=\{p_j\}_{j\in M}$, $\q = \{q_{ij}\}_{i\in N,j\in M}$, $\r = \{r_i\}_{i\in N}$, and $\bbeta = \{\beta_j\}_{j\in M}$.

\begin{subequations}\label{lcp:1}
\begin{eqnarray}
\forall i \in N: & \ e_i \leq \sum_{j} q_{ij} & \perp  \ \ r_i  \label{eq:ag_earn_1}\\
\forall j\in M: & \ \sum_{i} q_{ij} \leq p_j - \beta_j & \perp \ \ p_j \label{eq:ch_earn_1} \\
\forall i\in N, j\in M: & \ p_j \le d_{ij} r_i & \perp \ \ q_{ij} \label{eq:mpb_1} \\
\forall j\in M: & \ p_j - \beta_j \le c_j & \perp \ \ \beta_j \label{eq:earn_res_1}
\end{eqnarray}
\end{subequations}

\textit{Notation}.
We use the following notational convention. For a constraint labelled $L$, we represent its complementarity constraint expressing the non-negativity of a variable by $L$'. For example, \eqref{eq:ag_earn_1} is the constraint $e_i\le \sum_j q_{ij}$ for agent $i$, and \eqref{eq:ag_earn_1}' is the constraint $r_i \ge 0$.

\textit{Interpretation of the LCP.} In the above LCP, $p_j$ denotes the payment of chore $j$, $q_{ij}$ denotes the earning of agent $i$ from chore $j$, $r_i$ denotes the reciprocal of the MPB ratio of agent $i$, and $\beta_j$ denotes the excess payment of chore $j$, i.e., $q_j := p_j - \beta_j$ is the total earning from chore $j$. 

Constraint \eqref{eq:ag_earn_1} imposes that each agent $i$ earns at least their earning requirement of $e_i$.  Constraint \eqref{eq:ch_earn_1} imposes that the earning $\sum_i q_{ij}$ from each chore $j$ is at most $q_j = p_j - \beta_j$. Constraint \eqref{eq:mpb_1} enforces the MPB condition. Constraint \eqref{eq:earn_res_1} enforces the earning restriction on each chore. Constraints \eqref{eq:ag_earn_1}', \eqref{eq:ch_earn_1}', \eqref{eq:mpb_1}', \eqref{eq:earn_res_1}' enforce non-negativity of the LCP variables. The constraints \eqref{eq:ag_earn_1}-\eqref{eq:earn_res_1} and \eqref{eq:ag_earn_1}'-\eqref{eq:earn_res_1}' together define LCP\eqref{lcp:1}.

We now demonstrate the correspondence between the ER equilibria and certain solutions to LCP\eqref{lcp:1}.

\begin{restatable}{lemma}{lemLcpOneErToLcp}\label{lem:lcp1-er-to-lcp}
Any ER equilibrium can be used to construct a solution to LCP\eqref{lcp:1}.   
\end{restatable}
\begin{proof}
Let $(\x,\p)$ be an ER equilibrium. Let $\alpha_i$ be the MPB ratio of agent $i$ in $(\x,\p)$. Since all $d_{ij} > 0$, we have $\alpha_i > 0$ for all $i\in N$. Define $r_i = \alpha_i^{-1}$. Since $\p > 0$ in any ER equilibrium, we have that $r_i > 0$ for all $i\in N$. Let $q_{ij} = p_jx_{ij}$ be the earning of agent $i$ from chore $j$, and let $q_j = \sum_i q_{ij}$. Finally, define $\beta_j = \max\{0, p_j - c_j\}$ for each $j\in M$. We show that $(\p,\q,\r,\bbeta)$ is a solution to LCP\eqref{lcp:1} as follows:

(Constraint~\eqref{eq:ag_earn_1}) For all $i\in N$, $e_i = \sum_j q_{ij}$ since $(\x,\p)$ is an ER equilibrium (Def.~\ref{def:er-main} (i)). Also for all $i\in N$, $r_i > 0$.

(Constraints~\eqref{eq:mpb_1}) Since $\x$ is an MPB allocation, for all $i\in N, j\in M$ we have $d_{ij}/p_j \ge \alpha_i$, with equality if $x_{ij} > 0$. We then note that $\alpha_i^{-1} = r_i$ and $q_{ij} = p_j x_{ij}$.

(Constraints~\eqref{eq:ch_earn_1},~\eqref{eq:earn_res_1}) We consider two cases for each chore $j\in M$. If $p_j \le c_j$, then $q_j = p_j$ and $\beta_j = 0$. Otherwise, $p_j > c_j$, and we have $q_j = c_j$ and $\beta_j = p_j - c_j$. In both cases, the complementarity constraints~\eqref{eq:ch_earn_1} and \eqref{eq:earn_res_1} hold for each $j\in M$.
\end{proof}

\begin{restatable}{lemma}{lemLcpOneLcpToEr}\label{lem:lcp1-lcp-to-er}
Any solution $(\p,\q,\r,\bbeta)$ to LCP\eqref{lcp:1} with $\p>0$ can be used to construct an ER equilibrium $(\x,\p,\q)$.
\end{restatable}
\begin{proof}
We first argue that $\forall i\in N: r_i > 0$. Suppose $r_i = 0$ for some $i\in N$. Then constraint~\eqref{eq:mpb_1} implies $p_j = 0$ for all $j\in M$. In turn, with constraint~\eqref{eq:ch_earn_1} this implies that $q_{ij} = 0$ for all $i,j$. Then constraint \eqref{eq:ag_earn_1} cannot be satisfied for agent $i$ since $e_i > 0$, leading to the contradiction that $(\p,\q,\r,\bbeta)$ is a solution to LCP\eqref{lcp:1}. Therefore $\forall i\in N: r_i > 0$.  

Define an allocation $\x$ as $x_{ij} = q_{ij}/p_j$ and let $q_j = \sum_i q_{ij}$. We show $(\x,\p,\q)$ is an ER equilibrium by showing that it satisfies the conditions of \cref{def:er-main} as follows.

(Agents) The complementarity constraint \eqref{eq:mpb_1} implies that $(\x,\p)$ is an MPB allocation with $1/r_i$ being the MPB ratio of agent $i$. Moreover, constraint \eqref{eq:ag_earn_1} implies $e_i = \sum_j q_{ij}$ for all $i\in N$, since $r_i > 0$ for all $i\in N$. 

(Chores) Since $p_j > 0$ for all $j\in M$ by assumption, we have $q_j = p_j - \beta_j$. We consider two cases for each chore $j\in M$. If $\beta_j = 0$, then $q_j = p_j$ (from \eqref{eq:ch_earn_1}) and also $q_j \le c_j$ (from \eqref{eq:earn_res_1}). Thus $q_j = \min\{p_j,c_j\}$. Otherwise, $\beta_j > 0$ and $q_j = p_j - \beta_j < p_j$ (from \eqref{eq:ch_earn_1}) and $q_j = c_j$ (from \eqref{eq:earn_res_1}'). Thus $q_j = \min\{p_j,c_j\}$ in this case as well.
\end{proof}

\subsection{Main LCP for ER Equilibrium}
However, Lemke's scheme for LCP\eqref{lcp:1} may not converge to a solution with $\p>0$. We address 
this issue by performing a change of variables. First, we show that that chore payments can be assumed to be upper bounded by some constant $P$. 

\begin{restatable}{lemma}{lemPaymentUB}\label{lem:payment-ub}
For every ER instance $(N,M,D,e,c)$, there exists a constant $P$ such that for every ER equilibrium $(\x,\p,\q)$ there exists a scaled ER equilibrium $(\x',\p',\q)$ s.t. $\forall j\in M, p'_j \le P$.
\end{restatable}
\begin{proof}
We set the constant $P$ to $\frac{d_{max}}{d_{min}}\cdot c_{max}$, where $d_{max} = \max_{i,j} d_{ij}$, $d_{min} = \min_{i,j} d_{ij} > 0$, and $c_{max} = \max_j c_j$. 
If there exists some chore $j\in \x_i$ with $p_j \le c_j$, then the MPB condition implies that for every chore $k\in M$, $\frac{d_{ij}}{p_j} \le \frac{d_{ik}}{p_k}$, showing that $p_k \le \frac{d_{max}}{d_{min}}\cdot p_j \le P$.

Therefore, suppose $p_j > c_j$ for all $j\in M$. Then we uniformly decrease the payments as $p'_j = \frac{p_j}{\min_k p_k/c_k}$. We therefore have $p'_j \ge c_j$ for every $j$, but $p'_k = c_k$ for some $k$. We update the allocation $\x'$ s.t. $x'_{ij}\cdot p'_j = x_{ij}\cdot p_j$ for all $i,j$, ensuring that the earning vector of $(\x',\p')$ stays the same as $\q$. Since payments are decreased uniformly and $x'_{ij} > 0$ iff $x_{ij} > 0$, $(\x',\p')$ is MPB. Thus $(\x',\p',\q)$ is an ER equilibrium. The MPB condition for an agent $i$ s.t. $x_{ik} > 0$ implies that $\frac{d_{ik}}{p'_k} \le \frac{d_{ij}}{p'_j}$ for any $j\in M$. Thus, for any $j\in M$, $p'_j \le \frac{d_{ij}}{d_{ik}}\cdot p'_k \le \frac{d_{max}}{d_{min}}\cdot c_{max} = P$, as desired. 
\end{proof}

The upper bound $P$ on payments implies an upper bound $R$ on the reciprocal of the MPB ratios. Let $R$ be chosen so that $R\cdot \min_{i,j} d_{ij}> P$. We replace variable $p_j$ with $P-p_j$ and $r_i$ with $R-r_i$, while keeping the complementary constraints the same. Finally, we uniformly scale the input parameters $e$ and $c$ and obtain the following LCP. 

\begin{subequations}\label{lcp:2}
\begin{eqnarray}
\forall i \in N: & \ e_i\cdot\frac{\sum_j (P-p_j-\beta_j)}{\sum_h e_h} \leq \sum_{j} q_{ij} & \perp  \ \ r_i  \label{eq:ag_earn_2}\\
\forall j\in M: & \ \sum_{i} q_{ij} \leq (P-p_j - \beta_j) & \perp \ \ p_j \label{eq:ch_earn_2} \\
\forall i\in N, j\in M: & \ P-p_j \le d_{ij} (R-r_i) & \perp \ \ q_{ij} \label{eq:mpb_2} \\
\forall j\in M: & \ P-p_j-\beta_j \le c_j\cdot\frac{\sum_k (P-p_k-\beta_k)}{\sum_i e_i} & \perp \ \ \beta_j \label{eq:earn_res_2}
\end{eqnarray}
\end{subequations}

Similar to \cref{lem:lcp1-lcp-to-er}, we first establish a correspondence between certain solutions to LCP\eqref{lcp:2} and ER equilibria. We call a solution $(\p,\q,\r,\bbeta)$ \textit{`good'} if $\forall j\in M: p_j + \beta_j < P$ and $\forall i\in N: r_i < R$.

\begin{restatable}{lemma}{lemLcpTwoLcpEr}\label{lem:lcp2-lcp-er}
Any good solution of LCP\eqref{lcp:2} can be used to construct an ER equilibrium with all payments at most $P$, and vice versa. 
\end{restatable}
\begin{proof}
Let $\v =(\p,\q,\r,\bbeta)$ be a good solution of LCP\eqref{lcp:2}. Let $Q := \frac{\sum_j (P-p_j-\beta_j)}{\sum_i e_i}$. Since $\v$ is good, $Q > 0$. We first show that for all $i\in N, r_i > 0$. For the sake of contradiction, let $r_i = 0$ for some $i\in N$. Since $d_{ij}R > P$, constraint \eqref{eq:mpb_2} is not tight, which implies due to complementarity that $q_{ij} = 0$ for all $j\in M$. Hence $\sum_j q_{ij} = 0$, which implies that \eqref{eq:ag_earn_2} cannot hold, as $e_i\cdot Q > 0$. Thus, $\forall i: r_i > 0$. By complementarity, this means constraints \eqref{eq:ag_earn_2} must be tight: $\forall i\in N: e_i = \sum_j q_{ij}/Q$.

We next show that constraints \eqref{eq:ch_earn_2} must be tight. Suppose for some $j\in M$, $\sum_i q_{ij} < P-p_j -\beta_j$. Using inequalities \eqref{eq:ag_earn_2} and \eqref{eq:ch_earn_2}, we have:
\[
\sum_i e_i \cdot Q = \sum_i \sum_j q_{ij} = \sum_j \sum_i q_{ij} < \sum_j (P-p_j-\beta_j) = \sum_i e_i\cdot Q,
\]
which is a contradiction. Thus for all $j\in M, \sum_i q_{ij} = P-p_j-\beta_j$.

For $i\in N$ and $j\in M$, define:
\[\hat{p}_j := \frac{P-p_j}{Q}, \quad \hat{q}_{ij} := \frac{q_{ij}}{Q}, \quad \hat{q}_j := \frac{P-p_j-\beta_j}{Q}, \quad \hat{x}_{ij} = \frac{\hat{q}_{ij}}{e_i} \]

We show $(\hat{\x},\hat{\p},\hat{\q})$ is an ER equilibrium by showing it satisfies the conditions of \cref{def:er-main}.

(Agents) Since $e_i = \sum_j q_{ij}/Q$, we have $e_i = \sum_j \hat{q}_{ij}$ for all $i\in N$. Moreover the complementarity constraint \eqref{eq:mpb_2} implies $(\hat{\x},\hat{\p})$ is an MPB allocation and $Q/(R-r_i)$ is the MPB ratio of agent $i$. 

(Chores) Since $\sum_i q_{ij} = P-p_j-\beta_j$, we have $\hat{q}_j = \sum_i\hat{q}_{ij}$ for each $j\in M$. We consider two cases for each $j\in M$. If $\beta_j = 0$, then $\hat{q}_j = \hat{p}_j$ (by definition), and also $\hat{q}_j \le c_j$ (from \eqref{eq:earn_res_2}). Thus $\hat{q}_j = \min\{\hat{p}_j,c_j\}$. Otherwise $\beta_j > 0$, and $\hat{q}_j = \hat{p}_j - \beta_j/Q < \hat{p}_j$ (by definition) and $\hat{q}_j = c_j$ (from \eqref{eq:earn_res_2}). Then $\hat{q}_j = \min\{\hat{p}_j,c_j\}$ in this case as well.

Thus a good solution of LCP\eqref{lcp:2} can be used to construct an ER equilibrium. In the other direction, the argument of \cref{lem:lcp1-er-to-lcp} with the appropriate change of variables shows that an ER equilibrium with payments at most $P$ can be used to construct a good solution to LCP\eqref{lcp:2}.
\end{proof}

We now give the augmented LCP for LCP\eqref{lcp:2} so that we can apply Lemke's scheme as discussed in \cref{sec:lemke-prelim}. As is standard practice, we add the variable $z$ in the constraints whose right hand side is negative. We thus obtain LCP\eqref{lcp:3}.

\begin{subequations}\label{lcp:3}
\begin{eqnarray}
\forall i \in N: & \ e_i\cdot\frac{\sum_j (-p_j-\beta_j)}{\sum_h e_h} - \sum_{j} q_{ij} - z \le -e_i \cdot\frac{mP}{\sum_h e_h} & \perp  \ \ r_i  \label{eq:ag_earn_3}\\
\forall j\in M: & \ p_j + \beta_j + \sum_{i} q_{ij} \leq P & \perp \ \ p_j \label{eq:ch_earn_3} \\
\forall i\in N, j\in M: & \ d_{ij}r_i - p_j \le d_{ij} R - P & \perp \ \ q_{ij} \label{eq:mpb_3} \\
\forall j\in M: & \ -p_j-\beta_j + c_j\cdot\frac{\sum_k (p_k+\beta_k)}{\sum_i e_i} -z \le -P + c_j\cdot\frac{mP}{\sum_i e_i} & \perp \ \ \beta_j \label{eq:earn_res_3} \\
& z\ge 0
\end{eqnarray}
\end{subequations}

Let $\P$ be the polyhedron defined by the constraints of LCP\eqref{lcp:3}. 
The primary ray in Lemke's algorithm will set $\p, \q, \r, \bbeta$ to zero and $z = \max\{\max_i \frac{e_i\cdot mP}{\sum_h e_h}, \max_j (c_j\cdot \frac{mP}{\sum_i e_i} - P)\}$ as the initial vertex solution. Lemke's scheme involves pivoting from \textit{non-degenerate} vertices of $\P$, i.e., those with a unique double label (see \cref{sec:lemke-prelim}). 

\begin{definition}[Non-degenerate instance]\label{def:non-degenerate} An instance $(N,M,D,e,c)$ is non-degenerate if there is no polynomial relation between the input parameters, i.e., there is no polynomial $\phi$ s.t. $\phi(D,e,c) = 0$. 
\end{definition}

We can assume our instance is non-degenerate without loss of generality, as there are standard ways of handling degeneracy in the input parameters like the lexico-minimum test \cite{cottlepang-lcpbook}.

\begin{restatable}{lemma}{lemNonDegenerate}\label{lem:non-degenerate}
For a non-degenerate instance $(N,M,D,e,c)$, every vertex $\v = (\p, \q, \r, \bbeta, z)$ encountered in Lemke's scheme which is good and satisfies $z>0$ is non-degenerate.
\end{restatable}
\begin{proof}
For the sake of contradiction, suppose $\v = (\p, \q, \r, \bbeta, z)$ is a vertex encountered in Lemke's scheme which is good, where $z>0$, and which is degenerate. We show that the parameters of the instance have a polynomial relation, contradicting the instance being non-degenerate. 

Let $N$ be the number of variables in LCP\eqref{lcp:2}. Hence the augmented LCP has $N+1$ variables, with $z$ being the additional variable. Let the augmented LCP polyhedron be given by $\mathcal{P} = \{\y \in \R^{N},z\in \R : A\y \le \b, \y \ge 0, z\ge 0\}$, where $A\in \R^{N\times (N+1)}$, and $\b\in \R^N$. Thus, there are $N$ linear constraints given by the rows of $A$, and each such constraint may involve $(N+1)$ variables, including $z$.

Since $\v$ is a vertex of $\mathcal{P}$, exactly $(N+1)$ inequalities (out of the $2N+1$ inequalities, including the non-negativity constraints) must be tight at $\v$. Let $I$ be the set of non-zero variables of $\v$ excluding $z$, and let $|I| = N'$. By complementarity, the constraint $(A\v - \b)_i = 0$ for each such variable $i\in I$. Consider the subsystem of $A\y \le \b$ corresponding to the variables in $I$ and $z$. This can be represented as a collection of equalities given by $A'\cdot\v' = \b'$, where $A'\in\R^{N'\times (N'+1)}$, $\b\in\R^{N'}$ and $v'\in\R^{N'+1}$. Note that $\v' = (\v'', z)$ is simply the subvector of $\v$ with non-zero entries. By separating out the terms involving $z$, we can transform the above system into an equation of the form $z\cdot \boldsymbol{\gamma} + A''\cdot \v'' = \b'$, where $A''\in\R^{N'\times N'}$ and $\boldsymbol{\gamma} \in \R^{N'}$. This gives $\v'' = (A'')^{-1}\b' - z\cdot (A'')^{-1}\boldsymbol{\gamma}$, which expresses each non-zero variable in the set $I$ as a linear term in $z$ with coefficients that are polynomials in the input parameters.

Now observe that the degeneracy of $\v$ implies that $\v$ has at least two double labels (see \cref{sec:lemke-prelim}). That is, there are two variables $i, j \not\in I$ s.t. $y_i = y_j = (A\v-\b)_i = (A\v-\b)_j = 0$. We use one of these equalities to solve for $z$ by replacing each non-zero variable with its linear expression in $z$ obtained earlier. We then substitute this value of $z$ into the second equality to obtain a polynomial relation in the input parameters. This contradicts the fact that the instance is non-degenerate.
\end{proof}

\subsection{Convergence of Lemke's Scheme}
In this section, we show that Lemke's scheme converges to a good solution $(\p,\q,\r,\bbeta,z)$ with $z=0$ for LCP\eqref{lcp:3} for non-degenerate instances. A solution to LCP\eqref{lcp:3} with $z=0$ is a solution to LCP\eqref{lcp:2}. With \cref{lem:lcp2-lcp-er}, this implies the existence of ER equilibria and proves \cref{thm:er-existence}. Further, it provides an algorithm for computing an ER equilibrium.

To show convergence, we need to show that starting from the primary ray, Lemke's scheme only reaches good vertex solutions and does not reach a secondary ray. We prove the former using Lemmas~\ref{lem:conv1}, \ref{lem:conv2}, and \ref{lem:conv3}, and the latter using \cref{lem:conv4}. Recall that a solution $(\p,\q,\r,\bbeta,z)$ is good if $\forall j: p_j + \beta_j < P$ and $\forall i: r_i < R$.

\begin{restatable}{lemma}{lemConvOne}\label{lem:conv1}
Let $\v = (\p,\q,\r,\bbeta,z)$ be a vertex encountered by Lemke's scheme starting from the primary ray of LCP\eqref{lcp:3}. If $\forall j: p_j +\beta_j < P$ holds, then $\forall i: r_i < R$ also holds.
\end{restatable}
\begin{proof}
Suppose at a vertex $p_j + \beta_j < P$ for all $j$. Then for every $i$, \eqref{eq:mpb_3} implies $d_{ij}(r_i - R) \le p_j - P$. Since $p_j < P$ and $d_{ij} > 0$, we have $r_i < R$ for all $i$.
\end{proof}

\begin{restatable}{lemma}{lemConvTwo}\label{lem:conv2}
Let $\v = (\p,\q,\r,\bbeta,z)$ be a vertex encountered by Lemke's scheme starting from the primary ray of LCP\eqref{lcp:3}, and let $\v' = (\p',\q',\r',\bbeta',z')$ be the next vertex after a pivoting step. If $\forall j: p_j + \beta_j < P$ holds, then it cannot happen that $p'_j + \beta'_j = P$ holds for a strict non-empty subset of $M$. 
\end{restatable}
\begin{proof}
Let $M_0 \subseteq M$ be the set of chores in $\v'$ for which $p'_j + \beta'_j = P$. Assume $\emptyset \neq M_0 \subsetneq M$. Hence there are chores $g$ and $k$ s.t. $g\in M_0$ and $k\in M\setminus M_0$. Thus, $p'_g + \beta'_g = P$ and $p'_k + \beta'_k < P$. Let $E$ be the edge from $\v$ to $\v'$. Along $E$, at least one of \eqref{eq:ch_earn_3} and \eqref{eq:ch_earn_3}' and at least one of \eqref{eq:earn_res_3} and \eqref{eq:earn_res_3}' have to remain tight for chore $g$, due to complementarity. We now consider three cases:

\begin{itemize}[leftmargin=*]
\item The constraint \eqref{eq:earn_res_3} remains tight along $E$ for $g$. Thus, $P -p_g - \beta_g = c_g \frac{\sum_j (P-p_j-\beta_j)}{\sum_i e_i} + z$ holds along $E$. Since $P-p_g-\beta_g$ goes to $0$ along $E$, we must have that $c_g\frac{\sum_j (P-p_j-\beta_j)}{\sum_i e_i} + z$ goes to $0$ along $E$. Since $p_j+\beta_j \le P$ for all $j\in M$ due to constraint \eqref{eq:ch_earn_3} and $z\ge 0$, it must be that $P - p'_j -\beta'_j = 0$ for all $j\in M$ at $\v'$, contradicting that $p'_k + \beta'_k < P$.
\item The constraints \eqref{eq:ch_earn_3} and \eqref{eq:earn_res_3}' remain tight along $E$ for $g$. Thus $\sum_i q_{ig} = P - p_g - \beta_g$ and $\beta_g = 0$ along $E$. Thus $p_g < P$ pivots to $p'_g = P$ along $E$ from $\v$ to $\v'$. Since $\sum_i q_{ij} = P - p_g > 0$ at $\v$ and along $E$, there is some agent $i$ s.t. $q_{ij} > 0$. By complementarity for \eqref{eq:mpb_3}, we have $d_{ig}r_i - p_g = d_{ig} R - P$. Moreover $d_{ik}r_i - p_k \le d_{ik}R-P$. This implies:
\[ R - r_i = \frac{P - p_g}{d_{ig}} \ge \frac{P - p_k}{d_{ik}}.\]
Since $p_g$ pivots to $P$ along $E$, the above inequality implies $R-r_i$ also pivots to $0$. Hence $p_k$ must pivot to $P$ along $E$, i.e., $p'_k = P$. This contradicts the fact that at $\v'$, $p'_k + \beta'_k < P$.
\item The constraints \eqref{eq:ch_earn_3}' and \eqref{eq:earn_res_3}' remain tight along $E$ for $g$. Thus, $p_g = \beta_g = 0$ along $E$. In particular at $\v'$, $p'_g + \beta'_g = 0$. This contradicts $p'_g + \beta'_g = P$. Therefore it cannot happen that both \eqref{eq:ch_earn_3}' and \eqref{eq:earn_res_3}' remain tight along $E$.
\end{itemize}
Thus, it cannot happen that at $\v'$, $p'_j + \beta'_j = P$ holds for a strict non-empty subset of chores. 
\end{proof}

\begin{restatable}{lemma}{lemConvThree}\label{lem:conv3}
Let $\v = (\p,\q,\r,\bbeta,z)$ be a vertex encountered by Lemke's scheme starting from the primary ray of LCP\eqref{lcp:3}, and let $\v' = (\p',\q',\r',\bbeta',z')$ be the next vertex after a pivoting step. If $\forall j: p_j + \beta_j < P$ holds, then $\forall j: p'_j + \beta'_j = P$ cannot hold. 
\end{restatable}
\begin{proof}
For sake of contradiction, suppose $\forall j\in M, p'_j + \beta'_j = P$ at $\v'$. Let $E$ be the edge from $\v$ to $\v'$. Then for each chore $j$, either $p_j$ or $\beta_j$ increases along $E$. We consider three cases.

\begin{itemize}[leftmargin=*]
\item Suppose $p_j$ increases along $E$ for each $j\in M$. By complementarity, constraint \eqref{eq:ch_earn_3} must be tight at $\v$. Thus $\forall j: \sum_i q_{ij} = P-p_j-\beta_j$, which implies:

\begin{equation}\label{eq:conv3-1}
\sum_j\sum_i q_{ij} = \sum_j (P-p_j-\beta_j).
\end{equation}

Let $N_1 = \{i\in N: r_i = 0\}$ and $N_2 = N\setminus N_1$. For $i\in N_1$, since $r_i = 0$, $p_j \ge 0$ and $d_{ij} R - P > 0$, \eqref{eq:mpb_3} is not tight. By complementarity, $q_{ij} = 0$ for all $j\in M$. Thus \eqref{eq:ag_earn_3} becomes $e_i\cdot\frac{\sum_j (P-p_j-\beta_j)}{\sum_h e_h} \le z$ for $i\in N_1$. For $i\in N_2$, since $r_i > 0$, complementarity implies \eqref{eq:ag_earn_3} must be tight. Thus $e_i\cdot\frac{\sum_j (P-p_j-\beta_j)}{\sum_h e_h} - z = \sum_j q_{ij}$ for $i\in N_2$. Using these observations, we have:
\begin{equation}\label{eq:conv3-2}
\begin{aligned}
\sum_i\sum_j q_{ij} &= \sum_{i\in N_2} \sum_j q_{ij} \\
&= \sum_{i\in N_2} \left(e_i\cdot\frac{\sum_j (P-p_j-\beta_j)}{\sum_h e_h} - z \right)\\
&= \frac{\sum_{i\in N_2} e_i}{\sum_h e_h}\cdot\sum_j (P-p_j-\beta_j) - |N_2|z.
\end{aligned}
\end{equation}

Putting \eqref{eq:conv3-1} and \eqref{eq:conv3-2} together and rearranging, we obtain:
\[ 
\frac{\sum_{i\in N_1} e_i}{\sum_h e_h}\cdot\sum_j (P-p_j-\beta_j) = - |N_2|z.
\]
Since $p_j+\beta_j < P$ for all $j\in M$, we have $\sum_j (P-p_j-\beta_j) > 0$. Moreover $z>0$ at $\v$. Since either $N_1\neq\emptyset$ or $N_2\neq \emptyset$, the above equality cannot hold.  

\item Suppose $\beta_j$ increases along $E$ for each $j\in M$. By complementarity, constraint \eqref{eq:earn_res_3} must be tight. Hence $P-p_j-\beta_j -z= c_j\cdot\frac{\sum_k (P-p_k-\beta_k)}{\sum_i e_i}$ for all $j\in M$. Summing over all $j$ and rearranging gives:
\[ 
\sum_j (P-p_j-\beta_j)\cdot\left(\frac{\sum_j c_j}{\sum_i e_i} - 1\right) = -mz.
\]
Since $p_j+\beta_j < P$ for all $j\in M$, we have $\sum_j (P-p_j-\beta_j) > 0$. Since $\sum_i e_i \le \sum_j c_j$, hence the left side of the above equation is non-negative. Since $z>0$ at $\v$, the right side is negative, and hence the above equality cannot hold.  

\item There are two chores $j$ and $k$ such that $p_j = 0$ and $\beta_k = 0$ hold all along the edge $E$. It cannot be that $j = k$, since $p'_j + \beta'_j = P$. Since $p_k + \beta_k$ increases to $P$ while $\beta_k = 0$ along $E$, it must mean that $p_k$ increases to $P$ along $E$. By complementarity, \eqref{eq:ch_earn_3} is tight along $E$. Thus, $\sum_i q_{ik} = P - p_k > 0$ along $E$, since $p_k < P$ along $E$. Thus, there is some agent $i\in N$ for which $q_{ik} > 0$ along $E$. By complementarity, \eqref{eq:mpb_3} implies that $d_{ik}r_i-p_k = d_{ik}R-P$ holds along $E$. Moreover the constraint \eqref{eq:mpb_3} also implies $d_{ij}r_i-p_j \le d_{ij}R-P$ holds along $E$. This implies that:
\[ 
R- r_i = \frac{P-p_k}{d_{ik}} \ge \frac{P-p_j}{d_{ij}}
\]
holds along the edge $E$. However, since $p_j = 0$ along $E$, we have $p_k \le P \cdot(1-d_{ik}/d_{ij}) < P$, where $d_{ik}\neq d_{ij}$ follows from the non-degeneracy of the instance. Thus, $p_k$ always remains strictly below $P$ along the edge $E$. Therefore, it cannot happen that $p'_k = P$ at $\v'$, which is a contradiction.
\end{itemize}
Since these cases are exhaustive, the lemma holds.
\end{proof}

\begin{restatable}{lemma}{lemConvFour}\label{lem:conv4}
Lemke's scheme starting from the primary ray of LCP\eqref{lcp:3} does not reach a secondary ray.
\end{restatable}
\begin{proof}
Suppose Lemke's scheme starting from the primary ray of LCP\eqref{lcp:3} reaches a vertex $\v^0 = (\p^0,\q^0,\r^0,\bbeta^0,z^0)$ and then pivots to a secondary ray given by $\mathcal{R} = \{\v^0 + \alpha\cdot\v' : \alpha \ge 0\}$, where $\v' = (\p',\q',\r',\bbeta',z')$ with $z'>0$. We first show that $\v' = \textbf{0}$ by arguing that if this is not the case then some constraint of LCP\eqref{lcp:3} will be violated at some point on the secondary ray.

If $p'_j < 0$ for some $j\in M$ then eventually the non-negativity constraint \eqref{eq:ch_earn_3}' will be violated. On the other hand if $p'_j > 0$ then eventually the constraint \eqref{eq:ch_earn_3} will be violated. Thus $\p' = 0$. By similar arguments considering constraints \eqref{eq:mpb_3}' and \eqref{eq:ch_earn_3} for $\q'$, \eqref{eq:ag_earn_3}' and \eqref{eq:mpb_3} for $\r'$, and \eqref{eq:earn_res_3}' and \eqref{eq:ch_earn_3} for $\bbeta'$, we can conclude that $\q' = 0$, $\r' = 0$ and $\bbeta' = 0$. If $z' < 0$ then the $z\ge 0$ constraint will be violated eventually. Suppose $z'>0$. Then \eqref{eq:ag_earn_3} becomes strict, implying that $r_i = r^0_i = 0$ for each $i\in N$. Since $d_{ij}R-P>0$ for all $i,j$, \eqref{eq:mpb_3} is strict. By complementarity, $q_{ij} = q^0_{ij} = 0$ for all $i\in N,j\in M$. Similarly, since $z'>0$ \eqref{eq:earn_res_3} eventually becomes strict and by complementarity $\beta_j = \beta^0_j = 0$ for all $j\in M$. Since $\v^0$ is a vertex encountered in Lemke's scheme starting from the primary ray, \cref{lem:conv2} and \cref{lem:conv3} imply that $\forall j: p^0_j + \beta^0_j < P$. This means that \eqref{eq:ch_earn_3} is strict, and hence by complementarity $p_j = p^0_j = 0$ for all $j\in M$. Thus, $\p^0 = 0$, $\q^0 = 0$, $\r^0 = 0$ and $\bbeta^0 = 0$. We therefore have $\mathcal{R} = \{(0,0,0,0,z^0) + \alpha\cdot(0,0,0,0,z') : \alpha \ge 0\}$ where $z^0 > 0$ and $z' > 0$. However this is the same as the primary ray, thus showing $\mathcal{R}$ cannot be a secondary ray.
\end{proof}

We conclude the above discussion in the following lemma.
\begin{restatable}{lemma}{lemLCPConv}\label{lem:lcp-conv}
Lemke's scheme starting from the primary ray of LCP\eqref{lcp:3} converges to a good solution $(\p,\q,\r,\bbeta,z)$ where $z=0$.
\end{restatable}
\begin{proof}
\cref{lem:conv4} shows that no secondary rays are encountered. \cref{lem:conv2} and \cref{lem:conv3} together show that if at a vertex $\v$ it holds that $\forall j: p_j+\beta_j < P$, then the same holds at the next vertex $\v'$ after pivoting. Since the primary ray sets all $p_j = \beta_j = 0$, this is true initially. Hence $\forall j: p_j+\beta_j < P$ at every vertex encountered by Lemke's scheme. Finally \cref{lem:conv1} shows that at every such vertex $\forall i: r_i < R$ also holds, hence such a vertex is good. \cref{lem:non-degenerate} shows that every good vertex with $z>0$ is non-degenerate. Hence pivoting to the next step is always possible and Lemke's scheme eventually reaches a good solution with $z=0$.
\end{proof}

\cref{lem:lcp2-lcp-er} and \cref{lem:lcp-conv} thus prove \cref{thm:er-existence}: the existence of an ER equilibrium under the feasible earning condition, and also show that Lemke's scheme can be used to compute it.

\subsection{Computing ER Equilibrium in Polynomial Time for Constantly Many Agents}\label{sec:er-const}
In this section, we prove a positive result regarding the computation of ER equilibria when the number of agents is a constant. 

\thmConstER*

Assuming $\sum_i e_i \le \sum_j c_j$, \cref{thm:er-existence} guarantees the existence of an ER equilibrium. For such an instance, our algorithm effectively explores the space of all competitive allocations by enumerating the set of all \textit{consumption graphs} \cite{branzei2019choresceei} of an instance. The consumption graph $G_\z$ of an allocation $\z$ is a bipartite graph $G_\z = (N, M, E)$ where $(i,j)\in E$ iff $z_{ij} > 0$.

\begin{definition}[Rich family of graphs, \cite{branzei2019choresceei}] A collection of bipartite graphs $\mathcal{G}$ is said to be \textit{rich} for a given instance $(N,M,V)$ if for any fPO utility vector $\mathbf{u}$, there is a feasible allocation $\z$ with $\mathbf{u}(\z) = \mathbf{u}$ such that the consumption graph $G_\z$ is in the collection $\mathcal{G}$.
\end{definition}
Thus, a rich family of graphs contains the consumption graphs of every fPO utility vector for the instance.  A rich family of graphs $\mathcal{G}$ can be constructed in polynomial time for every instance with constant $n$.
\begin{proposition}[\cite{branzei2019choresceei}] For constant number of agents $n$, a rich family of graphs $\mathcal{G}$ can be constructed in time $O(m^{\frac{n(n-1)}{2}+1})$ and has at most $(2m+1)^{\frac{n(n-1)}{2}+1}$ elements.
\end{proposition}

\paragraph{Enumeration algorithm.} Our algorithm first constructs a rich family of consumption graphs $\mathcal{G}$, where $|\mathcal{G}| = O(m^{n^2})$. Then, for each consumption graph $G$, our algorithm (described below) decides in polynomial time if there is an ER equilibrium $(\x,\p)$ such that $G = G_\x$. Since an ER equilibrium is guaranteed to exist due to \cref{thm:er-existence}, and is fPO, there is some consumption graph which supports it. Since our enumeration of consumption graphs is exhaustive, our algorithm is therefore guaranteed to find an ER equilibrium in polynomial time.

\paragraph{Algorithm.} We now describe a polynomial time algorithm which when given a consumption graph $G$ as input, identifies if there exists an ER equilibrium $(\x,\p,\q)$ such that $G = G_\x$. Consider any competitive allocation $\x$ such that $G = G_\x$. Then by the Second Welfare Theorem, there exist payments $\p > 0$ s.t. $(\x,\p)$ satisfies the MPB condition, i.e., $x_{ij} > 0$ implies $\frac{d_{ij}}{p_j} = \min_c \frac{d_{ic}}{p_c}$. By definition of consumption graph, $(i,j)\in E[G]$ implies $x_{ij} > 0$ for any $\x$ s.t. $G = G_\x$. Thus we obtain that for any fPO allocation $\x$ s.t. $G = G_\x$, there is a set of payments $\p>0$ such that $(i,j) \in E[G]$ implies $\frac{d_{ij}}{p_j} = \min_c \frac{d_{ic}}{p_c}$. This suggests that we search for $\p > 0$ satisfying the above condition.

To do this, we write a program as follows. For each connected component $C$ in $G$, we arbitrarily choose a representative chore $j_C$. Then for each chore $j\in C$, we use the MPB condition along edges in $C$ to write $p_j = \mu_j\cdot p_{j_C}$ for some constant $\mu_j$ depending only on the disutilities. In more detail, we identify a path $j_0 = j \rightarrow i_1 \rightarrow j_1 \rightarrow i_2 \rightarrow \cdots \rightarrow i_k \rightarrow j_k = j_C$ comprising of agents $i_1, \dots i_k$, and chores $j_0 = j, j_1, \dots, j_k = j_C$. Then using the MPB condition along these edges, we obtain that $p_j = \frac{d_{i_1 j}}{d_{i_1 j_1}} \cdot \frac{d_{i_2 j_1}}{d_{i_2 j_2}} \cdots \frac{d_{i_k j_{k-1}}}{d_{i_k j_C}}\cdot p_{j_C} = \mu_j \cdot p_{j_C}$.

Then, we search for an ER equilibrium by writing the following program for each component $C = (N', M', E')$:
\begin{subequations}\label{eq:const-program}
\begin{eqnarray}
\forall i\in N': & e_i = \sum_{j: (i,j) \in E'} q_{ij} 
\label{eq:ag_earn_const}\\
\forall j\in M': & q_j = \sum_{i: (i,j) \in E'} q_{ij} \label{eq:ch_earn_const} \\
\forall j\in M': & q_j = \min(c_j, \ \mu_j \cdot p_{j_C}) 
\label{eq:ch_restriction}\\
\forall i\in N', j\in M': & q_{ij} \ge 0
\end{eqnarray}
\end{subequations}

The variables of the program are $q_{ij}$ for $i\in N', j\in M'$ s.t. $(i,j)\in E'$ and $p_{j_C}$. The variable $q_{ij}$ denotes the amount of money agent $i$ earns from chore $j$, and the variable $p_{j_C}$ denotes the payment of the representative chore $j_C$ of component $C$. The first and second constraints express the market clearing conditions for every agent $i\in N'$ and chore $j\in M'$. The third constraint expresses that the payment $q_j$ from chore $j$ is the minimum of its payment $p_j = \mu_j \cdot p_{j_C}$ and the earning cap $c_j$. The final constraint simply expresses non-negativity of $q_{ij}$. We show how to efficiently solve the above program.
\begin{lemma}\label{lem:const-program}
The program \eqref{eq:const-program} can be solved in polynomial time.
\end{lemma}
\begin{proof}
The constraint \eqref{eq:ch_restriction} can be equivalently expressed as: $q_j = \mu_j\cdot \min(p_{j_C}, \frac{c_j}{\mu_j})$ for all $j\in M'$. Note that $c_j/\mu_j$ is a constant. Thus, we can sort and re-label the $m'=|M'|$ chores in $M'$ so that $\frac{c_{j_1}}{\mu_{j_1}} \le \frac{c_{j_2}}{\mu_{j_2}} \le \cdots \le \frac{c_{j_{m'}}}{\mu_{j_{m'}}}$. We consider the $(m'+1)$ segments given by $[0, \frac{c_{j_1}}{\mu_{j_1}})$, $[\frac{c_{j_1}}{\mu_{j_1}}, \frac{c_{j_2}}{\mu_{j_2}})$, $\dots$, $[\frac{c_{j_{m'}}}{\mu_{j_{m'}}}, \infty)$, and consider the possibility that $p_{j_C}$ lies in each one of these segments. In the case that $p_{j_C}$ belongs to the $k^{th}$ segment for $k\in[m'+1]$, we obtain that the constraint \eqref{eq:ch_restriction} is $q_j = \mu_j p_{j_C}$ for chores $j\ge k$, and $q_j = c_j$ for chores $j < k$. Once the constraints are fixed, we obtain a linear program. Thus, the program \eqref{eq:const-program} can be solved by iterating over the $(m'+1)$ segments, fixing the constraints \eqref{eq:ch_restriction} for the chores depending on the segment, and then solving the resulting linear program in polynomial time.
\end{proof}

Suppose the consumption graph has $r$ components $C_1, \dots, C_r$, with $m_1, \dots, m_r$ chores respectively. For each component $C_k$, we construct the $(m_k+1)$ segments as described in \cref{lem:const-program}. For component $C_k$ and its segment indexed $s\in [m_k+1]$, let $p_{ks}$ be the value of $p_{j_{C_k}}$ returned by the linear program; we set $p_{ks} = 0$ if the program was infeasible. Thus, when $p_{ks} > 0$, it denotes the payment of the representative chore of component $C_k$ which lies in segment indexed $s\in [m_k+1]$. 

We now iterate over each configuration $(s_1,\dots, s_r) \in [m_1+1]\times [m_2+1]\times\cdots\times[m_r+1]$ and compute the representative payments $p_{ks_k}$ for each component $k\in [r]$. When $p_{ks_k} > 0$ for all $k\in[r]$, we compute the payments of all chores in each component $C_k$ and then check that the MPB condition is satisfied for all agents and chores across components as well. If so, it is clear that we have found an ER equilibrium since payments and chore earnings are solutions to the program \eqref{eq:const-program}. If one of the conditions fail for each configuration $(s_1, \dots, s_r)$, we conclude that the given consumption graph does not admit any ER equilibrium, and hence we move to the next consumption graph.

Finally, we argue that the above procedure terminates in polynomial time. Note that each $m_k\le m$, hence the number of configurations $(s_1,\dots, s_r)$ we iterate over is at most $(m+1)^r$. Since each component has at least one agent, we have $r\le n$. Since $n$ is a constant, there are at most $(m+1)^n = \poly{m}$ configurations. For each configuration we solve $r\le n$ linear programs, and check the MPB conditions in $\poly{n,m}$ time. Since there are $O(m^{n^2})$ consumption graphs which can be enumerated in $O(m^{n^2})$ time, we conclude that our algorithm terminates in polynomial time with an ER equilibrium. 

\newpage
\appendix

\section{Algorithms for Computing Approximately-\efone and \po Allocations}\label{app:efone}
\begin{algorithm}[!b]
\caption{ER Rounding for $(n-1)$-\efone and \po}\label{alg:efone-rounding}
\textbf{Input:} Instance $(N,M,D)$, with $m \ge n$ for earning limit $\beta = 1$; an ER equilibrium $(\y,\p)$ \\
\textbf{Output:} An integral allocation $\x$
\begin{algorithmic}[1]
\State $(\z,\p) \gets \texttt{MakeAcyclic}(\y,\p)$
\State Let $G = (N, M, E)$ be the payment graph associated with $(\z,\p)$
\State Root each tree of $G$ at some agent and orient edges
\State $\x_i \gets \emptyset$ for all $i\in N$ \Comment{Initialize empty allocation}
\State $L = \{j \in M: p_j \le \frac{1}{2}\}$, $H = \{j \in M: p_j > \frac{1}{2}\}$ \Comment{Low, High paying chores}
\Statex \textit{--- Phase 1: Round leaf chores ---}
\For{all leaf chores $j$} 
\State $\x_i \gets \x_i \cup \{j\}$ for $i= \parent{j}$; delete $j$ from $G$
\EndFor
\Statex \textit{--- Phase 2: Allocate $L$ ---}
\For{every tree $T$ of $G$}
\For{every agent $i$ of $T$ in BFS order}
\If{$\p(\x_i) > 1$}
\For{every $j\in\child{i}\cap H$}
\State Assign $j$ to agent $h\in\child{j}$ earning most from $j$ among $\child{j}$; delete $j$
\EndFor
\EndIf
\While{$\exists j\in\child{i}\cap L$ s.t. $\p(\x_i\cup \{j\}) \le 1$}
\State $\x_i \gets \x_i \cup \{j\}$; delete $j$ from $G$
\EndWhile
\For{every $j\in\child{i}\cap L$}
\State Assign $j$ to arbitrary agent $h\in\child{j}$; delete $j$ from $G$
\EndFor
\EndFor
\EndFor
\Statex \textit{--- Phase 3: Pruning trees ---}
\For{chore $j\in V(G) \cap M$}
\If{a $i\in \child{j}$ does not earn the most from $j$ among agents in $\child{j}$}
\State Delete edge $(j, i)$ from $G$
\EndIf
\EndFor
\Statex \textit{--- Phase 4: Matching to allocate $H$ ---}
\For{every tree $T = (N(T) \cup M(T), E(T))$ of $G$}
\State $h\gets \arg\max_{i\in N(T)} \p(\x_i)$
\State Compute a matching $\sigma$ of $i\in N(T)\setminus\{h\}$ to $M(T)$
\For{$i\in N(T)\setminus \{h\}$}
\State $\x_i \gets \x_i \cup \{\sigma(i)\}$
\EndFor
\EndFor
\State \Return $\x$
\end{algorithmic}
\end{algorithm}

We first present an algorithm which returns a $2(n-1)$-\efone and \fpo allocation for instances with $m\ge n$. Our algorithm, \cref{alg:efone-rounding}, takes as input an ER equilibrium $(\y,\p)$ of an instance with $m\ge n$ and earning limit $\beta = 1$, and performs essentially the same rounding algorithm as in \cref{alg:rounding}, except that the chore sets $L$ and $H$ are defined differently as $L = \{j\in M: p_j \le \frac{1}{2} \}$ and $H = \{j\in M : p_j > \frac{1}{2}\}$. We note that \cref{lem:rounding-runtime} (polynomial run-time) and \cref{lem:rounding-fpo} (allocation is always \fpo) are still applicable to \cref{alg:efone-rounding}. 

Analogous to Lemmas~\ref{lem:earning-ub} and \ref{lem:earning-lb}, we prove upper and lower bounds on the earning of agents in the allocation returned by \cref{alg:efone-rounding}.
\begin{restatable}{lemma}{lemEarningUBOne}\label{lem:earning-ub-1}
Let $(\x,\p)$ be the allocation returned by \cref{alg:efone-rounding} with earning restriction $\beta=1$. Then for each $i\in N$, $\p_{-1}(\x_i) \le 1$.
\end{restatable}
\begin{proof}
Let $\x^t$ denote the allocation after Phase $t$, for $t\in[4]$; note that $\x^4 = \x$. Consider an agent $i\in N$. Let $\hat{\x}_i$ be the allocation when \cref{alg:rounding} visits $i$ in Phase 2. Suppose $\p(\hat{\x}_i) \le 1$. Then we have $\p(\x^2_i) \le 1$ at the end of Phase 2 after $i$ is assigned a subset of $\child{i}\cap L$. Subsequently, $i$ could be assigned one more chore in Phase 4. Hence we have $\p_{-1}(\x_i) \le 1$ in this case. 

On the contrary, suppose $\p(\hat{\x}_i) > 1$. Then \cref{alg:rounding} will not allocate any chore to $i$ in Phase 4, and hence $\x_i = \x^2_i = \hat{\x}_i$. Note that either $\hat{\x}_i = \x^1_i$ or $\hat{\x}_i = \x^1_i \cup \{j\}$, where $j = \parent{i}$. That is, $\hat{\x}_i$ includes the chores $\x^1_i$ allocated to $i$ in Phase 1, and may include $i$'s parent chore $j$. Recall that Phase 1 rounds leaf chores to their parent agents, hence $\x^1_i$ comprises of the leaf chores that are child chores of $i$. 

Suppose there exists a chore $j_1 \in \x_i^1$ such that $p_{j_1} > 1$, i.e., there is a leaf chore $j_1$ rounded to $i$ whose payment exceeds the earning limit $\beta=1$. Then agent $i$ earns $e_i = 1$ from $j_1$ and no other chore, implying that $\x_i = \hat{\x}_i = \{j_1\}$. Then $\p_{-1}(\x_i) = 0 \le 1$.

Otherwise, $\p(\x_i^1) \le 1$. Then $\p_{-1}(\x_i) = \p_{-1}(\hat{\x}_i) \le \p(\hat{\x}_i \setminus \{j\}) = \p(\x_i^1) \le 1$, showing that the claim holds in this case too.
\end{proof}

\begin{restatable}{lemma}{lemEarningLBOne}\label{lem:earning-lb-1}
Let $(\x,\p)$ be the allocation returned by \cref{alg:efone-rounding}. Let $T = (N(T) \cup M(T), E(T))$ be a Phase 3 tree rooted at agent $i_0$. 
\begin{itemize}
\item[(i)] If some agent in $N(T)$ lost a child chore, then for every $i\in N(T)$, $\p(\x_i)\ge \frac{1}{2}$.
\item[(ii)] If no agent in $N(T)$ lost a child chore and $i_0$ received $\parent{i_0}$ chore, then for every $i\in N(T)$, $\p(\x_i)\ge \frac{1}{|N(T)|}$.
\item[(iii)] If no agent in $N(T)$ lost a child chore and $i_0$ lost $\parent{i_0}$ chore, then for every $i\in N(T)$, $\p(\x_i)\ge \frac{1}{2|N(T)|}$.
\end{itemize}
\end{restatable}
\begin{proof}
Let $(\z,\p)$ be the acyclic ER equilibrium computed before Phase 1. Let $\x^t$ denote the allocation after Phase $t$ of \cref{alg:efone-rounding}, for $t\in[4]$. Note that $\x^2 = \x^3$ since Phase 3 does not assign any chores and only deletes edges in $G$. Also note $\x^4 = \x$. 

Consider a Phase 3 tree $T$ rooted at agent $i_0$. Since $T$ is a Phase 3 tree, $T$ has exactly $|N(T)|-1$ chores, all of which belong to $H$. Phase 4 identifies the agent $h\in\arg\max_{i\in N(T)} \p(\x^3_i)$, and assigns a chore $\sigma(i) \in H$ to every agent $i\in N(T)\setminus\{h\}$ by computing a matching of $M(T)$ to $N(T)\setminus \{h\}$. Since $p_j > \frac{1}{2}$ for $j\in H$, we have $\p(\x_i)\ge p_{\sigma(i)} > \frac{1}{2}$ for all $i\in N(T)\setminus \{h\}$. Hence we only need to prove lower bounds on the earning $\p(\x_h)$ of the agent $h$. Note that $\x_h = \x^3_h = \x^2_h$, since $h$ is not allocated any chores in Phase 3 or 4. By choice of $h$, we also have that $\p(\x_h) \ge \p(\x^3_i) = \p(\x^2_i)$ for all $i\in N(T)$. We now analyze three scenarios.

\begin{itemize}[]
\item[(i)] Some agent $i\in N(T)$ lost a child chore $j\in \child{i}$. Suppose $i$ lost $j$ in Phase 2. If $j\in H$, then it must be that $\p(\x^2_i) > 1$. If $j\in L$, then it must be that $\p(\x^2_i) \ge \frac{1}{2}$; otherwise we would have assigned $j$ to $i$ in Phase 2. In either case, we have $\p(\x^2_i) \ge \frac{1}{2}$, and hence $\p(\x_h) \ge \p(\x^2_i) \ge \frac{1}{2}$ by choice of $h$. Note that $i$ cannot lose $j\in \child{i}$ in Phase 3 since Phase 3 only deletes edges from a chore to some of its child agents. This proves (i).

\item[(ii)] No agent in $N(T)$ lost a child chore and $i_0$ received $j_0 = \parent{i_0}$; it is possible that $\parent{i_0}=\emptyset$. This implies that no agent in $N(T)$ has lost \textit{any} chore they were earning from in $(\z,\p)$. Since the earning of each agent in $(\z,\p)$ equals 1, the total earning of agents in $N(T)$ is at least $|N(T)|$. The earning from the $|N(T)|-1$ chores in $M(T)$ is at most $(|N(T)|-1)$ due to the earning restriction on each chore in $M(T)$. Hence there is at least one agent $i\in N(T)$ whose earning $\p(\x^2_i)$ satisfies:
\[ \p(\x^2_i) \ge \frac{|N(T)| - (|N(T)|-1)}{|N(T)|} = \frac{1}{|N(T)|}.\]
Since $\p(\x_h)\ge \p(\x^2_i)$ by choice of $h$, this proves (ii).
\item[(iii)] No agent in $N(T)$ lost a child chore and $i_0$ lost $j_0 = \parent{i_0}$. In this case, no agent in $N(T)$ except $i_0$ has lost any chore they were earning from in $(\z,\p)$. We evaluate the amount of earning $i_0$ loses due to losing $j_0$. Suppose $j_0 \in H$. Then $i_0$ must have lost $j_0$ in either Phase 2 or 3 to some agent $i'\in\child{j_0}$ since $i_0$ was not earning the most from $j_0$ among agents in $\child{i_0}$. Due to the earning limit, agents can earn at most $1$ from $j_0$. Hence the earning from $i_0$ from $j_0$ is at most $\frac{1}{2}$. On the other hand, if $j_0\in L$, then $i_0$ earns at most $p_{j_0} \le \frac{1}{2}$ from $j_0$. In either case, we find that $i_0$ has only lost $\frac{1}{2}$ in earning. Hence the total earning of agents in $N(T)$ is at least $|N(T)|-\frac{1}{2}$, while that from the chores in $M(T)$ is at most $(|N(T)|-1)$. Hence there is at least one agent $i\in N(T)$ whose earning $\p(\x^2_i)$ satisfies:
\[ \p(\x^2_i) \ge \frac{|N(T)| - \frac{1}{2} - (|N(T)|-1)}{|N(T)|} = \frac{1}{2|N(T)|}.\]
Since $\p(\x_h)\ge \p(\x^2_i)$ by choice of $h$, this proves (iii).\qedhere
\end{itemize}
\end{proof}

\thmTwonEFone*
\begin{proof}
Let $(\x,\p)$ be the allocation returned by \cref{alg:efone-rounding} with $\beta = 1$. Consider a Phase 3 tree $T = (N(T)\cup M(T), E(T))$ rooted at agent $i_0$. Clearly, $|N(T)| \le n$. If $i_0$ lost the chore $\parent{j_0}$ to another agent $i_1$, it must be that $|N(T)| \le n-1$ since $i_1\notin N(T)$. We use these facts together with \cref{lem:earning-lb-1} to obtain that for all $i\in N$:
\[ \p(\x_i) \ge \min\bigg\{\frac{1}{2}, \frac{1}{n}, \frac{1}{2(n-1)}\bigg\} = \frac{1}{2(n-1)},\]
since $n\ge 2$. Moreover \cref{lem:earning-ub-1} implies that $\p_{-1}(\x_h) \le 1$ for any $h\in N$. Thus, for any pair of agents $i,h$, we have:
\[ 
\p_{-1}(\x_h) \le 1 = 2(n-1)\cdot \frac{1}{2(n-1)} \le 2(n-1)\cdot \p(\x_i),
\]
thus showing that $\x$ is $2(n-1)$-\efone by \cref{lem:pEF1impliesEF1}. \cref{lem:rounding-fpo} implies $\x$ is \fpo and \cref{lem:rounding-runtime} shows \cref{alg:rounding} runs in polynomial time. 
\end{proof}

\subsection{An Improved Algorithm Guaranteeing $(n-1)$-\efone and \po}\label{app:rebalancing}
Next, we improve our previous result by proving \cref{thm:rebalancing}.
\thmRebalancing*

Let $(\z,\p)$ be the $2(n-1)$-\efone and \po allocation returned by \cref{alg:efone-rounding}. We obtained this fairness guarantee by showing that $\p_{-1}(\z_i) \le 1$ and $\p(\z_h) \ge \frac{1}{2(n-1)}$ for all agents $i, h \in N$. Improving the lower bound to $\p(\z_h) \ge \frac{1}{n-1}$ for all $h\in N$ would imply that $\z$ is $(n-1)$-\efone and \po. Our algorithm aims to construct such an allocation in the event that $\z$ is not already $(n-1)$-\efone.

To do so, we revisit \cref{lem:earning-lb-1}, which shows lower bounds on the earning of agents in the allocation resulting from the matching phase of \cref{alg:efone-rounding}.

We call a Phase 3 tree $T$ \textit{`problematic'} if after running Phase 4, some agent in $T$ has an earning strictly less than $\frac{1}{n-1}$ in the resulting allocation $(\z,\p)$. By \cref{lem:earning-lb-1}, if (i) some agent in $T$ lost a child chore, or (ii) $i_0$ received $\parent{i_0}$ and $|N(T)|\le n-1$, or (iii) if $|N(T)| \le \frac{n-1}{2}$, then $\p(\z_i)\ge \frac{1}{n-1}$ for every $i\in N(T)$, and hence $T$ is not-problematic. This leaves two possibilities for a problematic tree: (i) $N(T) = [n]$, or (ii) $T$ is \textit{large}, i.e., $|N(T)| > \frac{n-1}{2}$, and no agent in $T$ has lost a child chore, and its root $i_1$ lost its parent chore $j_1 = \parent{i_1}$.

We eliminate case (i) by showing that a Phase 3 tree $T$ with $n$ agents and $n-1$ chores is not problematic. Phase 4 selects an agent $h\in \arg\max_{i\in [n]} \p(\x_i)$, where $\x$ is the allocation at end of Phase 3. In the matching phase, each agent $i\in [n]\setminus\{h\}$ is assigned a single chore $j_i$, while $h$ is not assigned any chore. The resulting allocation $\z$ is therefore given by $\z_h = \x_h$ and $\z_i = \x_i\cup\{j_i\}$ for all $i\neq h$. The following shows that $\z$ is actually $2$-\efone.
\begin{itemize}
\item[(i)] $i\neq h$ does not \efone-envy $h$, as $\p_{-1}(\z_i) \le \p(\z_i\setminus\{j_i\}) = \p(\x_i) \le \p(\x_h) = \p(\z_h)$.
\item[(ii)] $i\in [n]$ does not $2$-\efone-envy $\ell\neq h$, as $\p_{-1}(\z_i) \le 1 \le 2\cdot p_{j_\ell} \le \p(\z_\ell)$, since $p_{j_\ell}\ge \frac{1}{2}$ as $j_\ell\in H$.
\end{itemize}

Therefore, a tree $T$ is problematic iff case (ii) holds. If the allocation returned by \cref{alg:efone-rounding} is not $(n-1)$-EF1, then there must exist a \textit{`problematic'} Phase 3 tree $T_1$. Then $|N(T_1)| > \frac{n-1}{2}$, no agent in $T_1$ lost a child chore, and the root $i_1$ lost the $j_1 = \parent{i_1}$ chore to another agent. We have two cases:

\paragraph{Case 1.} We first handle the case of $j_1\in L$. The $(j_1,i_1)$ edge must have been deleted in Phase 2 when $j_1$ was assigned another agent $i_2$ (who is either $\parent{j_1}$ or a sibling of $i_1$). Our algorithm `unrolls' parts of \cref{alg:efone-rounding} in the \textit{`old run'} and re-visits the event in Phase 2 which deleted the edge $(j_1, i_1)$. This must have happened during a BFS call to agent $i_0$ in Phase 2, which happened before the BFS call to agent $i_1$. Let $T_0\supseteq T_1$ be the Phase 1 tree containing $T_1$. At this point, we \textit{`re-run'} Phase 2 on $T_0$ by starting with $i_1$ as the root agent of $T_0$. The chore $j_1$ now becomes the child of $i_1$. We visit all child chores of $i_1$ \textit{before} visiting $j_1$. 

Let $T_1'$ be the Phase 3 tree rooted at $i_1$ in the new run. Since $T_1$ is a problematic Phase 3 tree, $i_1$ received all of her child $L$-chores during Phase 2 of the old run. Since these child chores are visited before $j_1$, $i_1$ receives all of them in the re-run as well. This shows that irrespective of whether $i_1$ is assigned $j_1$ or not, the Phase 3 tree $T'_1$ produced in the new run is such that $N(T'_1) = N(T_1)$.

If $i_1$ is assigned $j_1$, then $T'_1$ is not problematic, as its root has not lost its parent: $i_1$ has no parent chore in the re-run. On the other hand, suppose $i_1$ loses $j_1$ to another agent $i_2$ who is a part of a Phase 3 tree $T_2$. Once again, $T_1'$ is not problematic as an agent $i_1$ has lost a child chore $j_1$. Suppose $T_2$ is problematic. Then $|N(T_2)| > \frac{n-1}{2}$. Note that $T_2$ is disjoint from $T'_1$, and since $T_1$ is problematic we have $|N(T'_1)| = |N(T_1)| > \frac{n-1}{2}$ as well. Since $|N(T_1)| + |N(T_2)| \leq n$, the above inequalities can only hold if $n = 2n'$ for $n'\in \N$ and $|N(T_1)| = |N(T_2)| = n'$. This implies that the Phase 1 tree $T_0$ comprises of trees $T_1$ and $T_2$ rooted at $i_1$ and $i_2$ respectively, both of which have edges to the chore $j_1$. In this case, we simply round $j_1$ to the agent in $\{i_1,i_2\}$ who earns more from $j_1$. Without loss of generality, suppose this agent is $i_2$. Then $T_2$ is not problematic since its root received its parent chore. In $T_1$, agents have lost an earning of at most $1/4$, since $p_{j_1} \le \frac{1}{2}$ as $j_1\in L$, and $i_1$ earned at most as much as $i_2$ did from $j_1$. Hence every agent in $T_1$ earns at least $\frac{1-1/4}{n'} = \frac{3}{2n}$. For $n\ge 3$, $\frac{3}{2n} \ge \frac{1}{n-1}$, showing that agents get the desired lower bound of $\frac{1}{n-1}$ on their earning. For $n=2$, it is easy to see that the resulting allocation is in fact \efone.

\paragraph{Case 2.} We now handle the case of $j_1 \in H$. The $(j_1,i_1)$ edge must have been deleted either in Phase 2 (during the BFS call to agent $i_0 = \parent{j_1}$) or Phase 3 (because $i_1$ was not earning the most from $j_1$ among $\child{j}$). In either case, $j_1$ retains an edge to a sibling $i_2$ of $i_1$ in the Phase 3 tree $T'$ containing $i_2$. Let $T_2\subseteq T'$ be the subtree rooted at $i_2$. Since $T'$ is a Phase 3 tree, $T_2$ is a Phase 3 tree as well, i.e., every chore in $T_2$ is adjacent to exactly two agents.

Let $s_1$ and $s_2$ be the earning of agents $i_1$ and $i_2$ from $j_1$ respectively. Following the proof of \cref{lem:earning-lb-1}, we observe that agents in $T_1$ have lost at most $s_1$ total earning. Hence there must exist some agent in $T_1$ who earns at least $\frac{1-s_1}{|N(T_1)|}$ from the chores assigned integrally thus far. If $\frac{1-s_1}{|N(T_1)|} \ge \frac{1}{n-1}$, then the allocation $\z$ must have already been $(n-1)$-\efone, hence we assume $\frac{1-s_1}{|N(T_1)|} < \frac{1}{n-1}$. This gives: $s_1 > 1-\frac{|N(T_1)|}{n-1}$. Since the earning from each chore is at most 1, we have $s_1+s_2\le 1$. Thus we obtain $s_2 < \frac{|N(T_1)|}{n-1}$. 

Our algorithm now `unrolls' parts of \cref{alg:efone-rounding} by re-visiting the event which deleted the edge $(j_1,i_1)$. Instead, the edge $(j_1,i_1)$ is re-introduced and the edge $(j_1,i_2)$ is deleted. This results in a larger tree $T''$ which contains $T_1$, and the Phase 3 tree $T_2$. After Phase 4, the earning of every agent in $T_2$ is at least:
\begin{align*}
\frac{1-s_2}{|N(T_2)|} &> \frac{1-|N(T_1)|/(n-1)}{|N(T_2)|} \tag{using $s_2 < \frac{|N(T_1)|}{n-1}$ dervied earlier}\\
&= \frac{1}{n-1}\cdot \left(\frac{n-1-|N(T_1)|}{|N(T_2)|}\right) \\
&\ge \frac{1}{n-1},  
\end{align*}
where the final inequality used the fact that $|N(T_1)| + |N(T_2)| \le n-1$, since $i_0 = \parent{j_1} \notin T_1\cup T_2$. Thus $T_2$ is not problematic. If the larger tree $T''$ is problematic, we recurse and repeat our algorithm with $T''$ instead, i.e., set $T_1\gets T''$. Since $N(T'')\supseteq N(T_1)\cup\{i_0\}$, every recursive step increases the size of $T_1$. Eventually it must happen that $T_1$ is non-problematic, or its root agent has no grand-parent agent. In the latter case, this tree $T_1$ must be non-problematic, since its root has not lost its parent root, and $|N(T_1)| \le n-1$ since $T_2\subsetneq T_1$. The algorithm therefore terminates with at most $n$ recursive calls.

\section{Appendix to Section~\ref{sec:bivalued}}\label{app:bivalued-small}
\cref{alg:bivalued-small} essentially follows the same template as \cref{alg:bivalued}, except that it begins with a balanced allocation $(\x^0,\p)$ computed using \cref{alg:balanced}. When $m \le n$, $\x^0$ is \efx and \po since every agent gets at most one chore. Hence we assume $m > n$. Since the allocation is balanced, we know $1 \le |\x^0_i| \le 2$ for all $i\in N$. 

For bivalued instances, we can scale the payments to ensure that for all $j\in M$, $p_j\in\{1,k\}$.
\begin{restatable}{lemma}{bivalued-payments-2}\label{lem:bivalued-payments-2}
Let $\p$ be the payment vector at the end of \cref{alg:balanced} on a $\{1,k\}$-bivalued instance. Then there exists $r \in \Z_{\geq 0}$ such that $p_j \in \{k^r, k^{r+1}\}$ for all $j \in M$.
\end{restatable}
\begin{proof}
We show that all chore payments in the run of \cref{alg:balanced} are powers of $k$. This is true initially since all chores are allocated to agent $h$ and pay $1$ or $k$. If all chore payments are powers of $k$ and the two possible disutility values differ by a factor of $k$, then the payment raise coefficient $\beta$ must also be a power of $k$. Thus, all chore payments continue to be powers of $k$ after the payment raise. Thus, it must be that all chore payments remain a power of $k$ throughout the run of Algorithm~\ref{alg:balanced}. 

Now suppose for sake of contradiction that there exist chores $j_1$ and $j_2$ such that $p_{j_1} = k^r$ and $p_{j_2} = k^{r+s}$, where $s > 1$. Since we have seen that Algorithm~\ref{alg:balanced} maintains a CE, it must be that $j_1$ is MPB for the agent $i$ it is allocated to. However, we have that $\alpha_{ij_2} = \frac{d_{ij_2}}{p_{j_2}} \leq \frac{k}{k^{r+s}} < \frac{1}{k^r} \leq \frac{d_{ij_1}}{p_{j_1}} = \alpha_{ij_1}$. Thus, $j_1$ cannot be MPB for agent $i$, and we have a contradiction. It must then be that in fact $s = 1$, showing the result.
\end{proof}

We classify the chores as low paying, $L = \{j:p_j = 1\}$, and high paying $H = \{j: p_j = k\}$. As in \cref{def:classification}, we define classify agents into sets $N_L$, $N_H^1$, and $N_H^2$ depending on whether they have only $L$-chores, a single $H$-chore, or two $H$-chores. We first note that $\x^0$ is \efx for agents in $N_L$, since $\p_{-X}(\x_i) = \p_{-1}(\x_i) \le 1 \le \p(\x_h)$ for any $i\in N_L$ and $h\in N$. Thus, if $H=\emptyset$, $N = N_L$ and \cref{alg:bivalued-small} simply returns $\x^0$. We therefore assume $H \neq \emptyset$ subsequently. With this assumption, the following statement regarding the MPB ratios of agents analogous to \cref{lem:mpb-ratios} holds.

\begin{restatable}{lemma}{lemMPBRatiosSmall}\label{lem:mpb-ratios-small}
Assume $H\neq \emptyset$. Then:
\begin{itemize}
\item[(i)] For all $i\in N_L$, $\alpha_i = 1$. Moreover for every $j\in H$, $d_{ij} = k$ and $j\in \mpb_i$.
\item[(ii)] For all $i\in N_H$, $\alpha_i \in \{1, 1/k\}$.
\item[(iii)] For all $i\in N_H$, if $\x_i \setminus H \neq \emptyset$ then $\alpha_i = 1$.
\end{itemize}
\end{restatable}

\begin{algorithm}[!t]
\caption{\efx + \po for bivalued instances with $m\le 2n$}\label{alg:bivalued-small}
\textbf{Input:} $\{1,k\}$-bivalued instance $(N,M,D)$ with $m\le 2n$\\ 
\textbf{Output:} An integral allocation $\x$
\begin{algorithmic}[1]
\State $(\x,\p) \gets$ \cref{alg:balanced} on $(N,M,D)$ \Comment{For all $j\in M$, $p_j\in\{1,k\}$}
\If{$m \le n$} \Return $\x$
\EndIf
\State $L = \{j \in M: p_j = 1\}$, $H = \{j \in M: p_j = k\}$ \Comment{Low, High paying chores}
\State Classify agents as $N_L, N_H^1, N_H^2$ as before
\Statex \textit{--- Phase 1: Address $N_H^2$ agents ---}
\While{$\exists i\in N_H^2$ not \efx}
\State $\ell \gets$ agent \efx-envied by $i$ \Comment{\cref{lem:efx-envy-mpb-edges} shows $\ell\in N_L$}
\If{$\p(\x_\ell) > 1$}
$S\gets j_1$ for some $j_1\in\x_\ell$ \Else{} $S\gets \emptyset$
\EndIf
\State $j \in \x_i \cap H$
\State $\x_\ell \gets \x_\ell \setminus S \cup \{j\}$
\State $\x_i \gets \x_i \cup S \setminus \{j\}$
\State $N_H^1 \gets N_H^1 \cup \{i,\ell\}$, $N_H^2 \gets N_H^2\setminus \{i\}$, $N_L \gets N_L\setminus \{\ell\}$
\EndWhile
\Statex \textit{--- Phase 2: Address $N_H^1$ agents ---}
\While{$\exists i\in N_H^1$ not \efx}
\State $\ell \gets \arg\min\{\p(\x_h) : h\in N \text{ s.t. $i$ \efx envies $h$}\}$ \Comment{\cref{lem:efx-envy-mpb-edges} shows $\ell\in N_L$}
\State $j \in \x_i \cap H$
\State $\x_i \gets \x_i \cup \x_\ell \setminus \{j\}$
\State $\x_\ell \gets \{j\}$
\State $N_H^1 \gets N_H^1 \cup \{\ell\}\setminus\{i\}$, $N_L \gets N_L \cup \{i\}\setminus\{\ell\}$
\EndWhile
\State \Return $\x$
\end{algorithmic}
\end{algorithm}

If $\x^0$ is not \efx, some agent in $N_H$ must \efx-envy another $\ell$ agent. \cref{alg:bivalued} addresses the \efx-envy of agents in $N_H$ by swapping some chores between agents $i$ and $\ell$ by performing the same swap steps defined in \cref{alg:bivalued}. The only point of difference is that \cref{alg:bivalued-small} performs a swap if $i$ \efx-envies $\ell$, whereas \cref{alg:bivalued} performs it if $i$ 3-\efx-envies $\ell$. Since there is a limited number of chores, \cref{alg:bivalued-small} can ensure agents in $N_H$ do not have too much cost: agents in $N_H^2$ have exactly two $H$ chores and no other chores, while agents in $N_H^1$ have exactly one $H$ chore and at most one $L$ chore.

We now prove the above claims formally. 
\begin{restatable}{lemma}{lemInvariantsSmall}\label{lem:invariants-small} (Invariants of Alg.\ref{alg:bivalued-small})
Let $(\x,\p)$ be an allocation in the run of \cref{alg:bivalued-small}. Then:
\begin{enumerate}[label=(\roman*)]
\item $(\x,\p)$ is an MPB allocation.
\item For all $i\in N$, $\p(\x_i) \ge 1$.
\item For all $i\in N_L$, $\p_{-1}(\x_i) \le 1$ during Phase 1.
\item For all $i\in N_L$, $\p(\x_i) < 1 + k$.
\item For all $i\in N_H^1$, $|\x_i\setminus H| \le |\x_i\cap H| = 1$.
\item For all $i\in N_H^2$, $|\x_i| = |\x_i\cap H| = 2$.
\end{enumerate}
\end{restatable}

We prove the above lemma using Lemmas~\ref{lem:phase1-inv-small} and \ref{lem:phase2-inv-small} below. Like \cref{lem:efx-envy-mpb-edges}, we show that an agent in $N_H = N_H^1\cup N_H^2$ can only \efx-envy another agent in $N_L$.

\begin{restatable}{lemma}{lemEFXenvyMPBedgesSmall}\label{lem:efx-envy-mpb-edges-small}
Consider an allocation $(\x,\p)$ satisfying the invariants of \cref{lem:invariants-small}. If $i\in N_H$ \efx-envies $\ell$, then $\alpha_i = 1$, $\ell\in N_L$, and $\x_\ell \subseteq \mpb_i$.
\end{restatable}
\begin{proof}
Consider $i\in N_H$ who \efx-envies $\ell\in N$. We know from \cref{lem:mpb-ratios-small} that $\alpha_i \in\{1, 1/k\}$. Suppose $\alpha_i = 1/k$. Then $d_{ij} = 1$ for all $j\in \x_i$. By the contrapositive of \cref{lem:mpb-ratios} (iii), we get $\x_i \subseteq H$. Since $|\x_i\cap H|\le 2$, we get $|\x_i| \le 2$. Thus,
\[
\max_{j\in\x_i} d_i(\x_i\setminus\{j\}) \le 1 \le d_i(\x_\ell),
\]
since $\x_\ell\neq\emptyset$ and the instance is bivalued. Thus $i$ is \efx towards $\ell$ if $\alpha_i = 1/k$, which implies that $\alpha_i = 1$. 

Suppose $\exists j\in \x_\ell$ such that $d_{ij} = k$. Then invariants (v) and (vi) imply that $\max_{j'\in \x_i} d_i(\x_i\setminus\{j'\}) \le k \le d_{ij} \le d_i(\x_\ell)$, showing that $i$ is \efx towards $\ell$. Thus it must be that for all $j\in \x_\ell$, $d_{ij} = 1$. The MPB condition for $i$ implies that $\alpha_i \le d_{ij}/p_j$, showing that $p_j \le 1$, and hence $p_j = 1$ for all $j\in \x_\ell$. Thus $\ell\in N_L$. Moreover, for any $j\in \x_\ell$, $\alpha_i = d_{ij}/p_j$, and hence $\x_\ell\subseteq\mpb_i$.
\end{proof}

The next two lemmas establish the invariants of \cref{lem:invariants-small}.
\begin{restatable}{lemma}{lemPhase1InvSmall}\label{lem:phase1-inv-small}
The invariants of \cref{lem:invariants-small} are maintained during Phase 1 of \cref{alg:bivalued-small}.
\end{restatable}
\begin{proof}
We prove the statement inductively. We first show that the invariants hold at $(\x^0,\p)$. Invariants (i), (ii), (iii), (v) and (vi) follow from the fact that $\x^0$ is a balanced allocation. For (iv), note that for any $i\in N_L$, we have $\p_{-1}(\x^0_i) \le 1$. Thus $\p(\x^0_i) \le 2 < 1+k$. 

Suppose the invariants hold at an allocation $(\x,\p)$ during Phase 1. Consider a Phase 1 swap involving agents $i\in N_H^2$ and $\ell\in N$. Given that \cref{alg:bivalued} performed the swap, $i$ must \efx-envy $\ell$. \cref{lem:efx-envy-mpb-edges} implies that $\ell \in N_L$ and hence $\x_\ell \in L$. As per \cref{alg:bivalued}, if $\p(\x_\ell) > 1$, then $S = \{j_1\}$ for some $j_1\in\x_\ell$, otherwise $S = \emptyset$. Let $j\in \x_i\cap H$.

Let $\x'$ be the resulting allocation. Thus $\x'_i = \x_i\setminus\{j\}\cup S$, $\x'_\ell = \x_\ell \setminus S\cup \{j\}$, and $\x'_h = \x_h$ for all $h\notin\{i,\ell\}$. We show that the invariants hold at $(\x',\p)$. Since a Phase 1 step removes agents $i$ and $\ell$ from $N_H^2$ and $N_L$ respectively, invariants  (iii), (iv), (vi) continue to hold. For the rest, observe:
\begin{enumerate}
\item[(i)] $(\x',\p)$ is on MPB. This is because \cref{lem:efx-envy-mpb-edges-small} implies $S\subseteq \x_\ell \subseteq \mpb_i$, showing $\x'_i\subseteq \mpb_i$. Since $\ell\in N_L$ at $(\x,\p)$ and $j\in H$, \cref{lem:mpb-ratios} shows $j\in \mpb_\ell$ and hence $\x'_\ell \subseteq \mpb_\ell$.
\item[(ii)] Follows from $\x'_i \neq \emptyset$ and $\x'_\ell \neq \emptyset$.
\item[(v)] For agent $i$, note that $\x'_i$ contains exactly one $H$-chore and perhaps one $L$-chore $j_1$. Hence  $|\x'_i\setminus H| \le 1 = |\x'_i \cap H|$, proving invariant (v).  

For agent $\ell$, note  that $\x'_i$ contains exactly one $H$-chore $j$, hence $|\x'_i\cap H| = 1$. Since invariant (iii) implies $|\x_i\setminus H| \le 2$, $|\x'_i\setminus H| \le 1$ after the potential transfer of $j_1$, thus proving (v). 
\end{enumerate}
The swap does not affect an agent $h\notin\{i,\ell\}$ and hence the invariants continue to hold for $h$ after the swap. By induction, we have shown that the invariants of \cref{lem:invariants-small} hold after any Phase 1 swap.
\end{proof}

\begin{restatable}{lemma}{lemPhase2InvSmall}\label{lem:phase2-inv-small}
The invariants of \cref{lem:invariants-small} are maintained during Phase 2 of \cref{alg:bivalued-small}. Moreover, agents in $N_H^2$ remain \efx towards other agents.
\end{restatable}
\begin{proof}
We prove the statement inductively. \cref{lem:phase1-inv-small} shows the invariants hold at the end of Phase 1. Suppose the invariants hold at an allocation $(\x,\p)$ during Phase 2. Consider a Phase 2 swap involving agents $i\in N_H^1$ and $\ell\in N$. Given that \cref{alg:bivalued-small} performed the swap, $i$ must \efx-envy $\ell$. \cref{lem:efx-envy-mpb-edges-small} implies that $\ell \in N_L$ and hence $\x_\ell \in L$. Let $j\in \x_i\cap H$.

Let $\x'$ be the resulting allocation. Thus $\x'_i = \x_i\setminus\{j\}\cup \x_\ell$, $\x'_\ell = \{j\}$, and $\x'_h = \x_h$ for all $h\notin\{i,\ell\}$. We now show that the invariants hold at $(\x',\p)$. Since we are in Phase 2, invariant (iii) does not apply, and since Phase 2 swaps do not alter the allocation of agents in $N_H^2$, invariant (vi) continues to hold. For the rest, observe:
\begin{enumerate}
\item[(i)] $(\x',\p)$ is on MPB. This is because \cref{lem:efx-envy-mpb-edges-small} implies $\x_\ell \subseteq \mpb_i$, showing $\x'_i\subseteq \mpb_i$. Since $\ell\in N_L$ at $(\x,\p)$ and $j\in H$, \cref{lem:mpb-ratios-small} shows $j\in \mpb_\ell$ and hence $\x'_\ell \subseteq \mpb_\ell$.
\item[(ii)] Follows from $\x'_i \neq \emptyset$ and $\x'_\ell \neq \emptyset$.
\item[(iv)] We want to show $\p(\x'_i) \le 1+ k$. To see this note that since $i$ \efx-envies $\ell$ in $\x$, $i$ must \pefx-envy $\ell$ in $(\x,\p)$. Using invariant (v), this means that $\p_{-X}(\x_i) = k > \p(\x_\ell)$.

Now $\p(\x'_i) = \p(\x_i\setminus\{j\}) + \p(\x_\ell) < 1 + k$, where we used $\p(\x_i\setminus\{j\}) = 1$ since invariant (v) shows $|\x_i\setminus H| \le 1$. 
\item[(v)] Note that $\ell\in N_H^1$ in $(\x',\p)$, and $|\x'_\ell \setminus H| = 0 < 1 = |\x'_\ell \cap H|$.
\end{enumerate}
The swap does not affect an agent $h\notin\{i,\ell\}$ and hence the invariants continue to hold for $h$ after the swap. By induction, we conclude that the invariants of \cref{lem:invariants} hold after any Phase 2 swap.

We now show that $i\in N_H^2$ cannot \efx-envy an agent $\ell\in N$. \cref{lem:efx-envy-mpb-edges} implies that $\ell\in N_L$ and hence $\x_\ell \subseteq L$. Let $\x^1$ be the allocation at the end of Phase 1. Note that the bundle $\x_\ell \subseteq L$ is obtained via a series of Phase 2 swaps initiated with some agent $\ell_1$ in $(\x^1,\p)$. Here, $\ell_1\in N_L$ at $(\x^1,\p)$. Thus $\x_\ell \supseteq \x^1_{\ell_1}$. Agent $i\in N_H^2$ did not \efx-envy $\ell_1$ in $\x^1$, otherwise \cref{alg:bivalued-small} would have performed a Phase 1 swap between agent $i$ and $\ell_1$. Since $\x_i = \x^1_i$ as \cref{alg:bivalued-small} does not alter allocation of agents in $N_H^2$ and $\x_\ell \supseteq \x^1_{\ell_1}$, $i$ will not \efx-envy $\ell$ in $\x$ either. Thus, all agents in $N_H^2$ continue to remain \efx during Phase 2.
\end{proof}

We need one final lemma showing that $N_L$ agents do not \efx-envy any other agent.
\begin{restatable}{lemma}{lemNLenvySmall}\label{lem:nl-envy-small}
At any allocation $(\x,\p)$ in the run of \cref{alg:bivalued-small}, $\x$ is \efx for every agent in $N_L$.
\end{restatable}
\begin{proof}
We know that the initial allocation $\x^0$ is \efx for agents in $N_L$. Let $\x$ be the earliest allocation in the run of \cref{alg:bivalued-small} in which an agent $i\in N_L$ \efx-envies another agent $h\in N$. Using $\alpha_i = 1$ from \cref{lem:mpb-ratios-small}, the bound on $\p(\x_i)$ from \cref{lem:invariants-small} (iv), we note:
\begin{equation}\label{eq:nl-envy-small}
\max_{j'\in \x_i}d_i(\x_i\setminus\{j'\}) = \alpha_i\cdot\p_{-X}(\x_i) < (1+k)-1 = k.
\end{equation}

Thus if $\exists j\in \x_h$ s.t. $d_{ij} = k$, then by \eqref{eq:nl-envy-small}, $d_i(\x_i) < k \le d_i(\x_h)$, showing that $i$ does not \efx-envy $h$ in $\x$. Hence it must be that for all $j\in\x_h$, $d_{ij} = 1$. This also implies $\x_h \subseteq \mpb_i$, since $\alpha_i = 1 = d_{ij}/p_j$ for any $j\in \x_h$. We now consider two cases based on the category of $h$.
\begin{itemize}[leftmargin=*]
\item $h\in N_H$. By definition of $N_H$, $\exists j\in \x_h$ s.t. $j\in H$. Since $i\in N_L$, by \cref{lem:mpb-ratios-small} (i) we get $d_{ij} = k$, which is a contradiction.
\item $h\in N_L$. Since $\x^0$ is \efx for agents in $N_L$, and Phase 1 swaps only remove agents from $N_L$, it cannot be that $i$ starts \efx-envying $h\in N_L$ during Phase 1. Let $\x'$ be the preceding allocation, at which \cref{alg:bivalued-small} performed a Phase 2 swap. Since $\x$ is the earliest allocation in which $i$ \efx-envies $\ell$, it must in $\x'$, agent $i$ was in $N_H^1$ and was involved with a Phase 2 swap with another agent $\ell\in N_L$. Since $\x_h\subseteq \mpb_i$ and \cref{alg:bivalued-small} did \textit{not} perform a swap between agents $i$ and $h$ in the allocation $\x'$, we must have $\p(\x'_\ell) \le \p(\x'_h)$ by the choice of $\ell$ at $(\x',\p)$. 

Note that $\x_i = (\x'_i \setminus H) \cup \x'_\ell$. By \cref{lem:invariants} (i), we know $\p_{-1}(\x'_i) = \p(\x'_i\setminus H) \le 1$. Thus we observe:
\begin{align*}
\p_{-X}(\x_i) &= \p_{-1}(\x_i) \tag{since $i\in N_L$}\\
&=\p(\x'_i\setminus H) + \p(\x'_\ell)- 1 \tag{since $\x_i = (\x'_i \setminus H) \cup \x'_\ell$}\\
&\le \p(\x'_\ell) \tag{since $\p(\x'_i\setminus H) \le 1$}\\
&\le \p(\x'_h) \tag{by choice of $\ell$ at $(\x',\p')$}\\
&\le \p(\x_h). \tag{since $\x_h = \x'_h$}
\end{align*}
Thus, $i$ is \efx towards $h$.
\end{itemize}
We conclude that it is not possible for an agent $i\in N_L$ to \efx-envy any other agent during the course of \cref{alg:bivalued-small}.
\end{proof}

We are now in a position to summarize and conclude our analysis of \cref{alg:bivalued-small}.

\thmBivaluedSmall*
\begin{proof}
Let $(\x^0,\p)$ be the initial balanced allocation obtained by using \cref{alg:balanced}. If $m\le n$ or $H = \emptyset$, then $\x^0$ is \efx, hence we assume otherwise. 

\cref{lem:nl-envy-small} shows that any allocation $\x$ in the course of \cref{alg:bivalued-small} is \efx for agents in $N_L$. Any potential \efx-envy is, therefore, from some agent $i\in N_H$. \cref{lem:efx-envy-mpb-edges-small} shows that if $i\in N_H$ is not \efx towards $\ell$, then $\ell\in N_L$. If $i\in N_H^2$, $i$ participates in a Phase 1 swap with agent $\ell$, after which $i$ and $\ell$ get removed from $N_H^2$ and $N_L$ respectively. This implies that Phase 1 terminates after at most $n/2$ swaps, and the resulting allocation is \efx for all agents in $N_H^2$. If $i\in N_H^1$, $i$ participates in a Phase 2 swap with agent $\ell$, after which $\ell$ is added to $N_H^1$ and is assigned a single chore and $\ell$ does not \efx-envy. This implies that Phase 2 terminates after at most $n$ swaps since the number of agents in $N_H^1$ who are not \efx strictly decreases. The resulting allocation is \efx for all agents in $N_H^1$. \cref{lem:phase2-inv-small} also shows that Phase 2 swaps do not cause $N_H^2$ agents to start \efx-envying any agent in $N_L$. Thus the allocation on termination of \cref{alg:bivalued-small} is \efx. By invariant (i) of \cref{lem:invariants-small}, $\x$ is also \fpo. Since there are at most $3n/2$ swaps and \cref{alg:balanced} takes polynomial time, \cref{alg:bivalued-small} terminates in polynomial time.
\end{proof}

\section{Examples}\label{app:examples}
\begin{restatable}{example}{exPropApx}\label{ex:prop1-apx}
A \s{Prop} and \po allocation need not be $\alpha$-\s{EF}$k$ for any $\alpha, k \geq 1$.

\normalfont Consider an instance with three agents $a$, $b$, and $c$ and three types of chores, each with $s > k$ many copies. The disutility of each agent for each chore type is given below.

\begin{center}
\begin{tabular}{ |c||c|c|c| } 
 \hline
  & Type 1 & Type 2 & Type 3 \\
 \hline
 $a$ & $1$ & $t$ & $3t$ \\ 
 \hline
 $b$ & $1$ & $t$ & $3t$ \\ 
 \hline
 $c$ & $t$ & $t$ & $1$ \\
 \hline
\end{tabular}
\end{center}

We claim that for $t > \frac{\alpha \cdot s}{s-k}$, the allocation $\x$ in which agent $a$ receives all type 1 chores, $b$ receives all type 2 chores, and $c$ receives all type 3 chores is \s{Prop} + \po but fails to be $\alpha$-\s{EF}$k$. We first note that as $\alpha \geq 1$ and $\frac{s}{s-k} > 1$, we have $t > 1$. It is then easily verified that $\x$ is \s{Prop}. Additionally, since $\x$ is social welfare maximizing, it is necessarily \fpo. We now show that agent $b$ $\alpha$-\s{EF}$k$-envies agent $a$. We have that
\begin{equation*}
\min_{S\subseteq\x_b, |S|\le k} d_b(\x_b \setminus S) = t(s-k) > \frac{\alpha \cdot s}{s-k} \cdot (s-k) = \alpha \cdot s = \alpha \cdot d_b(\x_a),
\end{equation*}
showing the result.
\end{restatable}

\begin{restatable}{example}{exCERounding}\label{ex:ce-rounding} A competitive equilibrium from equal earning (CEEE) need not admit a rounding which is $\alpha$-EF$k$. 

\normalfont We construct a CEEE and show that it admits no rounded $\alpha$-EF$k$ allocation. We consider three agents $a$, $b$, and $c$ with identical disutility functions. There exists one shared chore $j$ among the agents such that $a$, $b$, and $c$ earn $\frac{1}{2}$, $1 - \frac{1}{5\alpha}$, and $1 - \frac{1}{5\alpha}$ from $j$, respectively. Each agent $i$ is integrally allocated a set of chores $S_i$ such that:
\begin{itemize}
    \item agent $a$ earns $\frac{1}{2}$ from $S_a$.
    \item agent $b$ earns $\frac{1}{5\alpha}$ from $S_b$.
    \item agent $c$ earns $\frac{1}{5\alpha}$ from $S_c$.
\end{itemize}
Specifically, we note that $S_a$ consists of $2k$ identical chores which each pay $\frac{1}{4k}$. In conjunction with their earning from $j$, we see that each agent earns $1$, showing equal earnings. We now show that $\alpha$-EF$k$-envy between agents persists regardless of whom the single shared chore $j$ is rounded to in the integral allocation $\x$. We have that:
\begin{equation*}
    \min_{S \subseteq S_a, |S|\le k} d_a(S_a \setminus S) = \alpha_a \cdot \p_{-k}(S_a) = \alpha_a \cdot \frac{1}{4} > \alpha_a \cdot \frac{1}{5} = \alpha_a \cdot \alpha \cdot \p(S_b) = \alpha \cdot d_a(S_b).
\end{equation*}
Note that for $i \in \{a, b, c\}$, $S_i \subseteq \x_i$. Thus, we have that:
\begin{equation*}
    \min_{S \subseteq \x_a, |S|\le k} d_a(\x_a \setminus S) \geq \min_{S \subseteq S_a, |S|\le k} d_a(S_a \setminus S) > \alpha \cdot d_a(S_b).
\end{equation*}
It follows then that if $\x_b = S_b$, agent $a$ will $\alpha$-\s{EF}$k$-envy agent $b$. An analogous argument shows that if $\x_c = S_c$, agent $a$ will $\alpha$-\s{EF}$k$-envy agent $c$. Since $j$ can only be rounded to one agent, it must be that either $\x_b = S_b$ or $\x_c = S_c$, so agent $a$ must $\alpha$-\s{EF}$k$-envy some agent. 
\end{restatable}

\begin{example}\label{ex:ef1-tight}
There exists an ER equilibrium for which no rounding is $(n-1-\delta)$-\efone.

\normalfont We construct an ER equilibrium with $n = 2k + 1$ agents $i_1, \ldots, i_{2k+1}$, $2k - 1$ shared chores $j_1, \ldots, j_{2k-1}$, and uniform chore earning limit $\beta = 1$. Note that agents may have other chores which are not shared with other agents.

We describe the structure of the payment graph $G$. $G$ is a forest with two trees. The first tree $T_1$ consists of the lone agent $i_{2k+1}$ and $\frac{1}{\varepsilon}$ many small, $\varepsilon$-paying chores which are integrally allocated to $i_{2k+1}$. The second tree $T_2$ contains agents $i_1$ to $i_{2k}$ and all of the $2k - 1$ shared chores. We note that each agent $i \in T_2$ earns $\frac{1}{2k}$ from a set of chores $S_i$ which is integrally allocated to $i$, so we focus our attention on the edges that are incident to the $2k - 1$ shared chores.

Let $T_2$ be rooted at chore $j_{2k-1}$ so that $j_{2k-1}$ has two agents $i_{2k-1}$ and $i_{2k}$ as children, each of whom earn $\frac{1}{2}$ from $j_{2k-1}$. Then, $i_{2k-1}$ and $i_{2k}$ each have $k-1$ children chores, with $i_{2k-1}$ having children $j_1, \ldots, j_{k-1}$ and $i_{2k}$ having children $j_k, \ldots, j_{2k-2}$. Both $i_{2k-1}$ and $i_{2k}$ earn $\frac{1}{2k}$ from each of their children. Finally, for $r \in \{1, \ldots, 2k-2\}$, chore $j_r$ has one child agent $i_r$ who earns $1 - \frac{1}{2k}$ from $j_r$. We verify that each agent meets their earning requirement in the following:
\begin{itemize}
    \item For $1 \leq r \leq 2k-2$, $e_{i_r} = \frac{1}{2k} + \Big( 1 - \frac{1}{2k} \Big) = 1$.
    \item For $r \in \{2k-1, 2k\}$, $e_{i_r} = \frac{1}{2k} + \frac{1}{2} + (k-1) \cdot \frac{1}{2k} = 1$.
    \item For $r = 2k+1$, $e_{i_r} = \frac{1}{\varepsilon} \cdot \varepsilon = 1$.
\end{itemize}
Additionally, each shared chore meets the earning limit:
\begin{itemize}
    \item For $1 \leq r \leq 2k-2$, $q_r = \frac{1}{2k} + \Big( 1 - \frac{1}{2k} \Big) = 1$.
    \item For $r = 2k-1$, $q_r = \frac{1}{2} + \frac{1}{2} = 1$.
\end{itemize}
It is trivial to have the chores in a set $S_i$ satisfy the earning limit by increasing the number of chores in $S_i$ and thus decreasing the individual earning from each chore. We now show that there is no rounding of the ER equilibrium that is better than $(n-1)$-\efone. We may assume that the payment of any chore $j$ is equal to its payout to the agents, i.e., for all $j \in M$, $p_j = q_j$. In any rounding there must be some agent $h$ in $T_2$ who receives no shared chore and thus earns only $\frac{1}{2k}$: this is because there are $2k$ agents in $T_2$ but only $2k-1$ shared chores. Suppose for the sake of contradiction that a rounded allocation $\x$ is $(n-1-\delta)$-\efone for some $\delta > 0$. Note that the bundle of agent $i_{2k+1}$ is the same in any rounding, as $i_{2k+1}$ does not share chores with any agent. Then, letting 
$\varepsilon < \frac{\delta}{n-1}$ 
and letting the disutility function of $i_{2k+1}$ be such that for all $S \subseteq M$, $d_{i_{2k+1}}(S) = \p(S)$ (so $\alpha_{i_{2k+1}} = 1$ and all chores are MPB for $i_{2k+1}$), we have that 
\begin{align*}
\p_{-1}(\x_{i_{2k+1}}) &= \min_{j \in \x_{i_{2k+1}}} d_{i_{2k+1}}(\x_{i_{2k+1}} \setminus \{j\}) \\
&\leq (n-1-\delta) \cdot d_{i_{2k+1}}(\x_h) \\
&= (n-1-\delta) \cdot \p(\x_h),
\end{align*}
where the first and last equalities stem from our definition of $d_{i_{2k+1}}(\cdot)$ and the inequality stems from the fact that $\x$ is $(n-1-\delta)$-\efone. Using $\p_{-1}(\x_{i_{2k+1}}) \leq (n-1-\delta) \cdot \p(\x_h)$, we have that $1 - \varepsilon \le (n-1-\delta) \cdot \frac{1}{2k} = (n-1-\delta) \cdot \frac{1}{n-1}$ and equivalently that $\varepsilon \ge \frac{\delta}{n-1}$, a contradiction. Thus, such an ER equilibrium has no rounding which is $(n-1-\delta)$-\efone for any $\delta > 0$.
\end{example}

\begin{example}\label{ex:ef2-tight}
There exists an ER equilibrium for which no rounding is $(2-\delta)$-\eftwo.

\normalfont We modify the ER equilibrium from Example~\ref{ex:ef1-tight} with $n = 2k + 1$ agents. Setting $\beta = \frac{1}{2}$, we aim to change the agent earnings for each chore so that each agent still receives their earning requirement $1$ but the total earning from any chore is at most $\frac{1}{2}$. The changes are as follows:
 \begin{itemize}
     \item each agent $i \in T_2$ earns $\frac{1}{2} + \frac{1}{4k}$ from their set of integrally allocated chores $S_i$,
     \item agents $i_{2k-1}$ and $i_{2k}$ each earn $\frac{1}{4}$ from their parent chore $j_{2k-1}$,
     \item agents $i_{2k-1}$ and $i_{2k}$ earn $\frac{1}{4k}$ from each of their $k-1$ children,
     \item for $r \in \{1, \ldots, 2k-2\}$, agent $i_r$ earns $\frac{1}{2} - \frac{1}{4k}$ from their parent chore.
 \end{itemize}
As in Example~\ref{ex:ef1-tight}, it can be verified that each agent meets her earning requirement and each chore satisfies the earning limit.
Then, also as in Example~\ref{ex:ef1-tight}, it must be that some agent in $T_2$ receives only their integrally allocated chores. That is, there exists some $i \in T_2$ such that $\x_i = S_i$ and $d_{i_{2k+1}}(\x_i) = \p(\x_i) = \frac{1}{2} + \frac{1}{4k}$. Recall that agent $i_{2k+1}$ is only allocated $\frac{1}{\varepsilon}$ many small, $\varepsilon$-paying chores and $\min_{S \subseteq \x_{i_{2k+1}}, |S|\le 2} d_{i_{2k+1}}(\x_{i_{2k+1}} \setminus S) = p_{-2}(\x_{i_{2k+1}})$. Then, for any $\delta > 0$, we may choose sufficiently large $k$ and sufficiently small $\varepsilon$ so that 
\begin{align*}
\p_{-2}(\x_{i_{2k+1}}) &\geq \varepsilon \cdot \Big( \frac{1}{\varepsilon} - 2 \Big) \\
&= 1 - 2\varepsilon > (2 - \delta) \cdot \Big( \frac{1}{2} + \frac{1}{4k} \Big) \\
&= (2 - \delta) \cdot \p(\x_i) \\
&= (2 - \delta) \cdot d_{i_{2k+1}}(\x_i)
\end{align*}
and thus agent $i_{2k+1}$ $(2-\delta)$-\eftwo-envies agent $i$.
\end{example}

\begin{example}\label{ex:aEFX-2-ary} 
For any $\alpha\ge 1$, there exists a 2-ary instances for which there is no allocation that is $\alpha$-\efx and \fpo.

\normalfont Consider the following 2-ary instance with two agents $\{a,b\}$, four chores $\{j_1,j_2,j_3,j_4\}$ with disutilities given by:
\begin{center}
\begin{tabular}{|c||c|c|c|c|}
\hline
& $j_1$ & $j_2$ & $j_3$ & $j_4$ \\
\hline
$a$ & 1 & 1 & $3\alpha$ & $3\alpha$ \\
\hline
$b$ & 1 & 1 & $3\alpha^2 + 3\alpha$ & $3\alpha^2 + 3\alpha$ \\
\hline
\end{tabular}
\end{center}

We show that this instance does not admit an allocation which is both $\alpha$-\efx and \fpo. Suppose that $\x$ is an $\alpha$-\efx allocation. Given that $\alpha \geq 1$, we see that neither agent may receive both $j_3$ and $j_4$ under $\x$, as for $i \in \{a, b\}$ we have 
\begin{equation*}
    \max_{j \in \{j_3, j_4\}} d_i(\{j_3, j_4\} \setminus \{j\}) \geq 3 \alpha > 2 \alpha = \alpha \cdot d_i(\{j_1, j_2\}). 
\end{equation*}
Thus, we assume without loss of generality that $j_3 \in \x_a$ and $j_4 \in \x_b$. We next argue that $j_1, j_2 \in \x_a$. Suppose for sake of contradiction (again w.l.o.g.) that $j_2 \in \x_b$. We have
\begin{align*}
\max_{j \in \x_b} d_b(\x_b \setminus \{j\}) &\geq d_b(\{j_2, j_4\} \setminus \{j_2\}) \\
&= 3\alpha^2 + 3\alpha \\
&> \alpha \cdot (3\alpha + 1) \\
&= \alpha \cdot d_b(\{j_1, j_3\}) \\
&\geq \alpha \cdot d_b(\x_a),  
\end{align*}
so agent $b$ would $\alpha$-\efx-envy agent $a$. Thus, it must be that $\x_a = \{j_1, j_2, j_3\}$ and $\x_b = \{j_4\}$. We now show however that $\x$ is not \fpo as it is dominated by the fractional allocation $\y$ where $\y_a = \{j_1, j_3, \frac{1}{3\alpha}j_4\}$ and $\y_b = \{j_2, \frac{3\alpha-1}{3\alpha}j_4\}$. Indeed, we have that $d_a(\y_a) = 3\alpha + 2 = d_a(\x_a)$ and $d_b(\y_b) = 3\alpha^2 + 2\alpha < 3\alpha^2 + 3\alpha = d_b(\x_b)$. Agent $a$'s disutility remains the same under $\y$ while agent $b$'s disutility strictly decreases under $\y$, so $\y$ dominates $\x$ and $\x$ not \fpo. Thus, the given instance admits no $\alpha$-\efx and \fpo allocation. 
\end{example}

\bibliographystyle{plainnat}
\bibliography{references}

\end{document}